\documentclass[11pt, a4paper, oneside]{report}
\usepackage{preamble}

\begin{document}

\begin{titlepage}
\AddToShipoutPicture*{\put(0,0){\includegraphics*{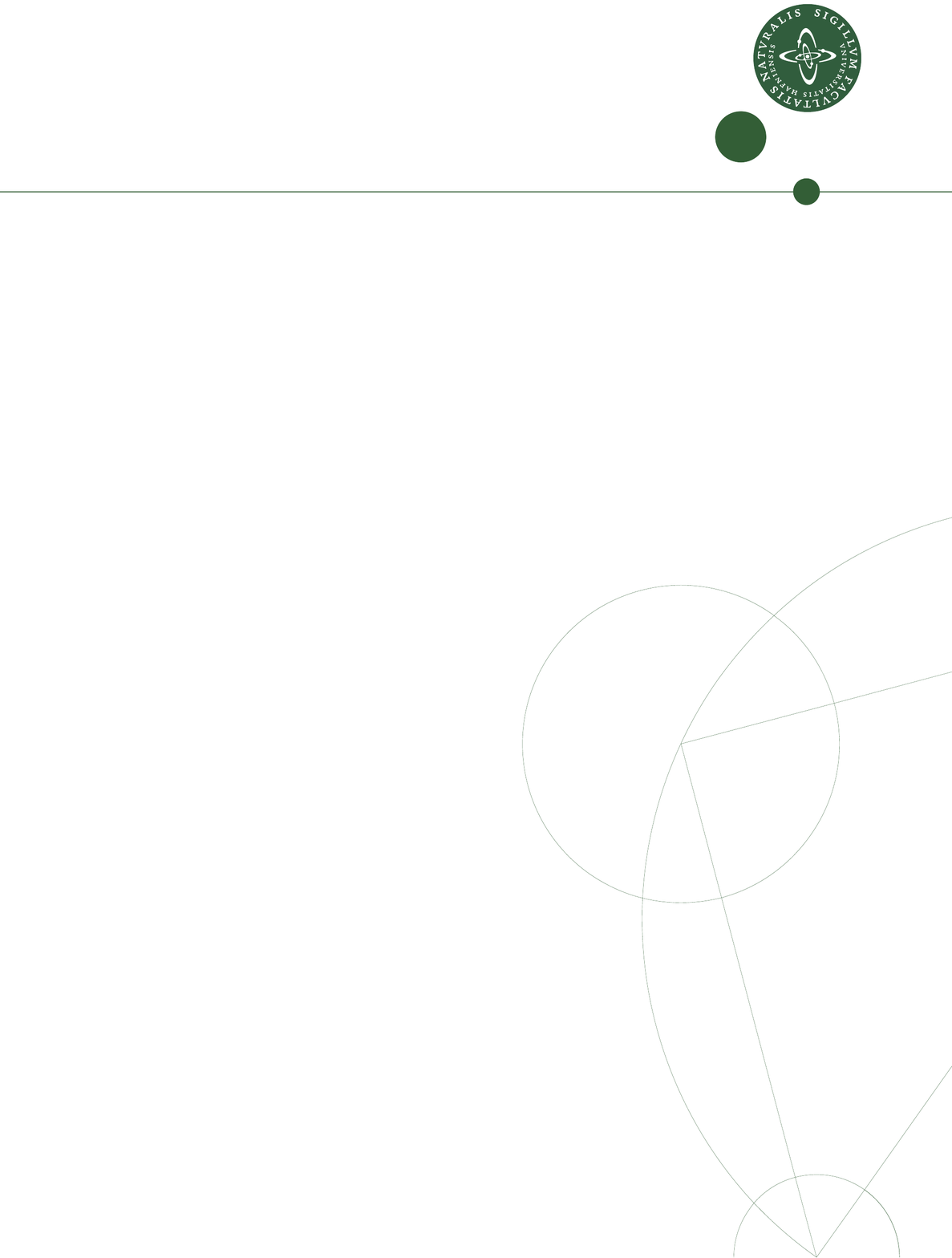}}}
\AddToShipoutPicture*{\put(0,0){\includegraphics*{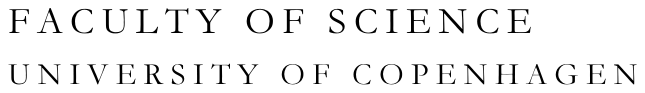}}}

\begin{flushleft}
\vspace*{3cm}
\textbf{\huge{Master's Thesis}}

\vspace*{3mm}
\textbf{Tue Haulund} - \texttt{qvr916@alumni.ku.dk}

\vspace*{4cm}
\textbf{\huge{Design and Implementation of a Reversible
Object-Oriented Programming Language}}
\vfill
\textbf{Main Supervisor:} Torben Ægidius Mogensen

\textbf{Co-Supervisor:} Robert Glück

\textbf{Submitted:} November 8, 2016
\end{flushleft}
\end{titlepage}
\newpage

\thispagestyle{empty}
\mbox{}
\setcounter{page}{0}
\newpage

\chapter*{Abstract}
\markright{Abstract}

High-level reversible programming languages are few and far between and in general offer only rudimentary abstractions from the details of the underlying machine. Modern programming languages offer a wide array of language constructs and paradigms to facilitate the design of abstract interfaces, but we currently have a very limited understanding of the applicability of such features for reversible programming languages.

We introduce the first reversible object-oriented programming language, ROOPL, with support for user-defined data types, class inheritance and subtype-polymorphism. The language extends the design of existing reversible imperative languages and it allows for effective implementation on reversible machines.

We provide a formalization of the language semantics, the type system and we demonstrate the computational universality of the language by implementing a reversible Turing machine simulator. ROOPL statements are locally invertible at no extra cost to program size or computational complexity and the language provides direct access to the inverse semantics of each class method.

We describe the techniques required for a garbage-free translation from ROOPL to the reversible assembly language PISA and provide a full implementation of said techniques. Our results indicate that core language features for object-oriented programming carries over to the field of reversible computing in some capacity.
\newpage

\chapter*{Preface}
\markright{Preface}

\epigraph{``A language that doesn't affect the way you think about programming, is not worth knowing``}{ -- Alan J. Perlis, \textit{Epigrams on Programming}~\cite{ap:epigrams}}

\vspace{1cm}

\noindent The present thesis constitutes a 30 ECTS workload and is submitted in partial fulfillment of the requirements for the degree of Master of Science in Computer Science at the University of Copenhagen (UCPH), Department of Computer Science (DIKU).
\vspace{3mm}

The thesis report consists of \pageref*{LastPage} numbered pages, a title page and a ZIP archive containing source code developed as part of the thesis work. The thesis was submitted for grading on November 8, 2016 and will be subject to an oral defense no later than December 6, 2016.
\vspace{3mm}

I would like to express my sincerest appreciation for the invaluable direction and encouragement of my primary academic supervisor, Torben Mogensen. I would also like to thank my co-supervisor Robert Glück, for introducing me to the fascinating field of reversible computing and for his help with the thesis subject. Finally - a heartfelt appreciation is owed to my loving partner Matilde, without whom this thesis would not have been possible.

\begin{flushleft}
Copenhagen, Autumn 2016
\vspace{3mm}

Tue Haulund
\end{flushleft}
\newpage

\tableofcontents
\newpage

\listoffigures
\newpage

\chapter{Introduction}
\label{chp:introduction}

Reversible computing is the study of time-invertible, two-directional models of computation. At any point during a reversible computation, there is at most one previous and one subsequent computational state, both of which are uniquely determined by the current state. The computational process follows a deterministic trajectory of these states in either direction of execution and carefully avoids erasing information such that previous states remain reachable and unique. As a result of this perfect preservation of information, reversible computing offers a possible solution to the heat dissipation problems faced by manufacturers of microprocessors~\cite{rl:irreversibility}.

To realize a fully reversible computing system, we need reversibility at every level of abstraction. Much headway has been made at the circuit and gate level, such as the realization of the reversible \textit{Pendulum} architecture~\cite{cv:pendulum} based on the reversible universal Fredkin and Toffoli gates~\cite{ef:conservative}. High-level reversible programming languages are also actively researched, most notably the imperative reversible language \textit{Janus}~\cite{cl:janus, ty:janus, ty:ejanus}, the procedural reversible language \textit{R}~\cite{mf:r, mf:reversibility} and the functional reversible languages \textit{RFUN}~\cite{ty:rfun} and \textit{Inv}~\cite{sm:inv}. Recently, translation of these languages to low-level reversible assembly languages has been the subject of some work~\cite{ha:translation, jsk:translation}. A reversible self-interpreter for the reversible imperative language \textit{R-WHILE} was shown in~\cite{rg:rwhile}.

Throughout this existing body of research, a reversible object-oriented language has yet to be formalized. The present thesis discusses the design of such a language as well as the techniques required to perform a clean (i.e.\ garbage-free) and correct translation from such a language to a low-level reversible assembly language. As is the case for any programming paradigm, reversible object-oriented programming has its own programming techniques and pitfalls, which we will explore in detail. The language will implement traditional OOP concepts such as encapsulation, subtype polymorphism and dynamic dispatch, albeit in a reversible context.

\section{Reversible Computing}
\label{sec:revcomp}

A great deal of effort is expended on minimizing the power consumption of modern microprocessors, to the point where it is now considered a first-class design constraint. However a theoretical lower limit does exist for our current model of computation. Known since the early 1960's, \textit{Landauer's principle} holds that:
\begin{displayquote}[rl:irreversibility][.]
\textelp{} any logically irreversible manipulation of information, such as the erasure of a bit or the merging of two computation paths, must be accompanied by a corresponding entropy increase in non-information-bearing degrees of freedom of the information-processing apparatus or its environment
\end{displayquote}
Put simply, Landauer's principle states that the erasure of information in a system is always accompanied by an increase in energy consumption. The exact amount of energy required to erase $n$ bits of information is $n \cdot k_B \cdot T \cdot \ln 2$, where $T$ is the temperature of the circuit in kelvin and $k_B$ is the \textit{Boltzmann constant} (approximately $1.38 \cdot 10^{-23}$ J/K)~\cite{cb:thermodynamics}.

This theoretical limit is known as the \textit{von Neumann-Landauer limit} and it places a lower bound on the energy consumption of any computation involving the erasure of information. In a reversible computation, information is never erased, which means reversible computing systems are not subject to the von Neumann-Landauer limit\footnote{Aside from its relationship to reversible computing, Landauer's principle also represents a compelling argument that Maxwell's Demon does not violate the second law of thermodynamics~\cite{cb:notes}}.

The naive approach to achieving reversibility is based on the idea of \textit{reversibilization} of a regular irreversible program. As the program is executing, intermediate values are preserved in a program history trace. Known as a \textit{Landauer embedding}, this technique achieves perfect preservation of information~\cite{rl:irreversibility}. Bennett showed that such an embedding can be created for any irreversible program~\cite{cb:reversibility}, however the space requirements for this technique grows proportionally to the length of time the program has been running. Given an irreversible program with running time $T$ and space complexity $S$, a semantically equivalent reversible program with running time $O(T^{1\ +\ \epsilon})$ and space complexity $O(S \ln T)$ can be constructed for some $\epsilon > 0$~\cite{cb:tradeoff}. These space requirements make this approach completely impractical for general purposes.

The Landauer embedding is an example of \textit{injectivization} of the function that our program computes. As we cannot accept the generation of this extraneous garbage data, we must limit ourselves to programs that compute functions that are already injective (i.e.\ one-to-one functions). Reversible programming languages are made up of individually reversible execution steps, each of which must also be injective when viewed as a mapping from one computational state to the next. This one-to-one mapping ensures that the language is both forwards and backwards deterministic, there is always at most one state the computation can transition to, regardless of the direction of execution.

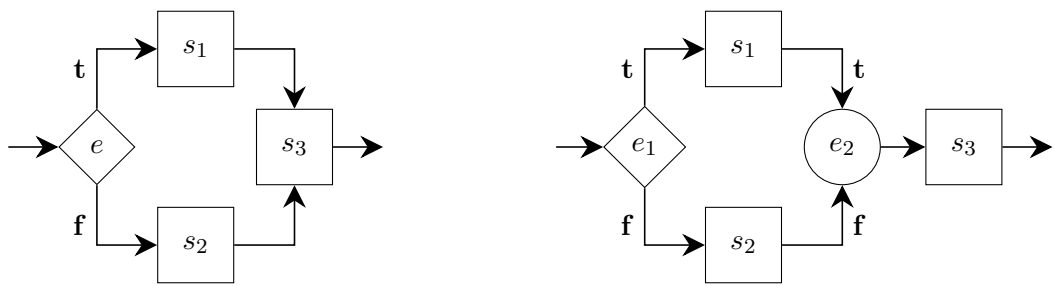
\begin{figure}[h]
\centering
\begin{subfigure}[b]{.495\columnwidth}
    \centering
    \begin{tikzpicture}[node distance = 1.3cm, auto]
        \node[none] (start) {};
        \node[conditional, right of = start] (e) {$e$};
        \node[block, right of = e, above of = e] (s1) {$s_1$};
        \node[block, right of = e, below of = e] (s2) {$s_2$};
        \node[block, right of = e, node distance = 2.6cm] (s3) {$s_3$};
        \node[none, right of = s3] (end) {};
        
        \draw[arrow] (start) -- (e);
        \draw[arrow] (e) |- node[anchor = north east] {\textbf{t}} (s1);
        \draw[arrow] (e) |- node[anchor = south east] {\textbf{f}} (s2);
        \draw[arrow] (s1) -| (s3);
        \draw[arrow] (s2) -| (s3);
        \draw[arrow] (s3) -- (end);
    \end{tikzpicture}
\end{subfigure}
\begin{subfigure}[b]{.495\columnwidth}
    \centering
    \begin{tikzpicture}[node distance = 1.3cm, auto]
        \node[none] (start) {};
        \node[conditional, right of = start] (e1) {$e_1$};
        \node[block, right of = e1, above of = e1] (s1) {$s_1$};
        \node[block, right of = e1, below of = e1] (s2) {$s_2$};
        \node[assertion, right of = e1, node distance = 2.6cm] (e2) {$e_2$};
        \node[block, right of = e2, node distance = 1.6cm] (s3) {$s_3$};
        \node[none, right of = s3] (end) {};
        
        \draw[arrow] (start) -- (e1);
        \draw[arrow] (e1) |- node[anchor = north east] {\textbf{t}} (s1);
        \draw[arrow] (e1) |- node[anchor = south east] {\textbf{f}} (s2);
        \draw[arrow] (s1) -| node[anchor = north west] {\textbf{t}} (e2);
        \draw[arrow] (s2) -| node[anchor = south west] {\textbf{f}} (e2);
        \draw[arrow] (e2) -- (s3);
        \draw[arrow] (s3) -- (end);
    \end{tikzpicture}
\end{subfigure}
\caption[Flowcharts of reversible and irreversible variants of a conditional statement]{Flowcharts of irreversible and reversible variants of a conditional statement followed by some other statement $s_3$. The reversible variant uses the assertion $e_2$ to join the two paths of computation reversibly - if the control flow reaches $e_2$ from the \textbf{true}-edge then $e_2$ must evaluate to \textbf{true} and vice versa, otherwise the statement is undefined.~\cite{ty:flowchart}}
\label{fig:conditional}
\end{figure}

In irreversible programming languages, this mapping can be a many-to-one (non-injective) function since we are then only concerned with forward determinism. The inverse of such a function is a one-to-many relation (sometimes called a \textit{multivalued function}) which means such languages are backwards non-deterministic, as it is impossible to uniquely determine the previous state of computation\footnote{Some languages are both forwards and backwards non-deterministic by design - the logic programming language Prolog is an example of such a language.}.

Every reversible program has exactly one corresponding inverse program in which every execution step is inverted and performed in reverse order of the original program. Since each execution step is locally invertible, as opposed to requiring a full-program analysis, the inversion can be achieved with straightforward recursive descent over the components of the program. Furthermore, given that each single execution step has a single-step inverse, the process of inverting a reversible program bears no additional cost in terms of program size.

Reversible programming languages may provide direct access to the inverse semantics of a code segment, in Janus this is exemplified by the \textit{uncall} statement which invokes the inverse computation of a given procedure~\cite{ty:janus}, while low-level reversible languages typically make use of a direction bit to invoke inverse semantics and reverse execution~\cite{mt:bob}. This direct access has given rise to some clever programming methodologies. One example is known as the \textit{Lecerf-Bennet reversal}\footnote{Also known as the \textit{local Bennett's method}~\cite{ty:janus}, or the \textit{compute-copy-uncompute paradigm}. It was first proposed by Lecerf~\cite{yl:reversibility} and later rediscovered by Bennett~\cite{cb:reversibility}.} which makes use of \textit{uncomputation} to reversibly purge the variable store of undesired intermediate values after a computation.

For some computations, having direct and inexpensive access to the exact inverse computation can be useful from a software development perspective. For example, implementing a compression algorithm in a reversible language\footnote{The futility of attempting to implement a \textit{lossy} compression algorithm within a language paradigm that forbids the erasure of information should not be lost on the reader at this point.} immediately yields the equivalent decompression routine by inversion of the program. Additionally, any effort that has gone into verifying the correctness of the compression algorithm, e.g. testing or perhaps even formal verification techniques such as model checking, can serve as an equally valid testament to the correctness of the inverse program (assuming the process of inversion is itself correct).

Besides the primary motivation of potentially improving the energy efficiency of computers beyond the von Neumann-Landauer limit, the field of reversible computing shows promise in a number of other areas:

\begin{description}[leftmargin = 0pt]
    \item[Quantum Computing] A quantum logic gate represents a transformation which can be applied to an isolated quantum system. For the resulting system to be consistent, the transformation matrix must be \textit{unitary}. Such transformations are inherently reversible, and indeed any reversible boolean function can be converted to a corresponding unitary transformation~\cite{ab:quantum}. As such, the field of quantum computing could stand to benefit from an increased understanding of reversible computing.
    \item[Program Debugging] Traditional program debugging involves stepping through code line by line, inspecting intermediate results and memory contents accordingly. Recently, vendors have added support for \textit{reverse debugging}, which involves stepping through code in reverse or restoring earlier program states from within a debugging session. This is usually implemented with a continous execution trace but on a reversible computing platform, such functionality is supported as a fundamental property of the system. A reversible extension to the Erlang programming language, for the purpose of supporting reverse debugging was suggested in~\cite{nn:erlang}.
    \item[Error Recovery] In parallel or pipeline-based systems, recovery from an unforeseen error condition often involves undoing recent related changes made to the state of the system. As an example, this is a primary features of most DBMS and it is implemented with special-purpose error recovery logic. On a reversible system however, this can be achieved by simple reverse execution back to the point where the error condition first arose. A reversible DSL for error recovery on robotic assembly lines was presented in~\cite{us:assembly}.
    \item[Discrete Event Simulation] The simulation of systems with asynchronous discrete update events lends itself well to concurrent execution. Suggested in~\cite{dj:timewarp}, the Dynamic Time Warp (DTW) algorithm is commonly used to synchronize event updates across execution threads. DTW uses update rollbacks to restore the simulation to a synchronized state, in case an event has been committed prematurely. Reversible computation can be used to realize event rollback while avoiding the high overhead of storing execution traces or simulation checkpoints~\cite{cc:des}.
\end{description}

\section{Object-Oriented Programming}
\label{sec:oop}

Like reversible computing, object-oriented programming (OOP) originated in the early 1960's, with the advent of the Simula language~\cite{gb:simula}. \textit{Unlike} reversible computing, OOP enjoys immense popularity in the software industry, as can be observed by the widespread use of object-oriented languages such as Java and C++. The OOP paradigm attempts to break a problem into many small manageable pieces of related state and behaviour called objects. An object may model an actual object in the problem domain, or it may represent a more abstract grouping of related entities within a program. A distinction is made between a particular kind or type of object, called a \textit{class}, and specific instances of these classes, known simply as objects.

OOP is based on the concept of encapsulation: Only the methods of an object has unrestricted access to the components of that object, thereby protecting the integrity of the internal state and reducing the overall system complexity. Encapsulation is closely related to the principle of information hiding, which holds that compartmentalization of design decisions made in one part of a program can be used to avoid extensive modification of other parts of that program if the design is altered~\cite[Chapter~1]{eg:patterns}.

A fundamental aspect of OOP is class inheritance, which allows one class to inherit the fields and methods of another class. Most OOP languages also use inheritance to establish an "is-a" relationship between two objects such that one may be substituted for the other by subtype-polymorphism. OOP lends itself well to code-reuse and maintainability of source code, and is often used in combination with imperative or procedural programming paradigms. In general, OOP is a set of techniques for intuitively structuring imperative code - it is a programming methodology rather than a model of computation.

\section{Motivation}
\label{sec:motivation}

After more than two dozen iterations of Moore's Law~\cite{gm:moore}, the semiconductor industry is fast approaching the von Neumann-Landauer limit. Reversible computing may be a viable solution, but it represents a significant paradigm shift from the currently prevailing irreversible models of computation.

The practicality of reversible computing hinges, inter alia, on the presence of high-level reversible programming languages that can be compiled to low-level reversible assembly code without significant overhead. Ideally, these languages should provide the same tools and features for producing abstract models and interfaces as are available for modern irreversible languages.

Object-oriented programming is immensely popular in the industry but the combination of OOP and reversible computing is entirely uncharted territory. The work presented in this thesis is motivated by the scarcity of high-level reversible programming languages and in particular, by the absence of any reversible object-oriented programming languages.

\section{Thesis Statement}
\label{sec:statement}

An effective implementation of a reversible object-oriented programming language is both possible and practical, provided the design of the language observes the limitations required for execution on reversible machines.

\section{Outline}
\label{sec:outline}

This thesis consists of 5 chapters, the first of which is this introductory chapter. The remaining 4 chapters are summarised as follows:
\begin{description}
\item[Chapter~\ref{chp:survey}] is a brief survey of existing reversible imperative programming languages and instruction sets.
\item[Chapter~\ref{chp:roopl}] presents the reversible object-oriented programming language ROOPL, along with a formalization of the language and a discussion of the most significant elements of its design.
\item[Chapter~\ref{chp:compilation}] presents the techniques required for a garbage-free and correct compilation from ROOPL source code to PISA instructions.
\item[Chapter~\ref{chp:conclusion}] contains conclusions and proposals for future work.
\end{description}
The appendix contains the source code listings for the ROOPL compiler, an example ROOPL program and the equivalent translated PISA program.
\newpage

\chapter{Reversible Programming Languages}
\label{chp:survey}

The following chapter contains a survey of reversible instruction sets and reversible imperative programming languages. Given that OOP is an approach for naturally organizing imperative code, it is clear that such languages are of special interest when designing a reversible OOP language. Indeed, the design of our reversible OOP language draws heavily from the design of the languages and instruction sets presented in this section.

\section{Janus}
\label{sec:janus}

The reversible programming language \textit{Janus} (named after the two-faced Greco-Roman god of beginnings and endings) was created by Cristopher Lutz and Howard Derby for a class at Caltech in 1982~\cite{cl:janus}. It was later rediscovered and formalized in~\cite{ty:janus} and some modifications were suggested in~\cite{ty:ejanus} - the following section deals with this modified version of the language.

\begin{figure}[h]
\centering

\textbf{Janus Grammar}

\begin{align}
    prog\quad&::=\quad p_{main}\ p^*\tag{program}\\
    t\quad&::=\quad \textbf{int}\ |\ \textbf{stack}\tag{data type}\\
    p_{main}\quad&::=\quad \textbf{procedure main}\ \textbf{\texttt{()}}\ (\textbf{int}\ x(\textbf{\texttt{[}}\overline{n}\textbf{\texttt{]}})^?\ |\ \textbf{stack}\ x)^*\ s\tag{main procedure}\\
    p\quad&::=\quad \textbf{procedure}\ q\textbf{\texttt{(}}t\ x,\ \dots,\ t\ x\textbf{\texttt{)}}\ s\tag{procedure definition}\\
    s\quad&::=\quad x\ \odot\textbf{\texttt{=}}\ e\ |\ x\textbf{\texttt{[}}e\textbf{\texttt{]}}\ \odot\textbf{\texttt{=}}\ e\tag{assignment}\\
    &\ |\ \qquad \textbf{if}\ e\ \textbf{then}\ s\ \textbf{else}\ s\ \textbf{fi}\ e\tag{conditional}\\
    &\ |\ \qquad \textbf{from}\ e\ \textbf{do}\ s\ \textbf{loop}\ s\ \textbf{until}\ e\tag{loop}\\
    &\ |\ \qquad \textbf{push}\textbf{\texttt{(}}x,\ x\textbf{\texttt{)}}\ |\ \textbf{pop}\textbf{\texttt{(}}x,\ x\textbf{\texttt{)}}\tag{stack modification}\\
    &\ |\ \qquad \textbf{local}\ t\ x\ \textbf{\texttt{=}}\ e\quad s\quad \textbf{delocal}\ t\ x\ \textbf{\texttt{=}}\ e\tag{local variable block}\\
    &\ |\ \qquad \textbf{call}\ q\textbf{\texttt{(}}x,\ \dots,\ x\textbf{\texttt{)}}\ |\ \textbf{uncall}\ q\textbf{\texttt{(}}x,\ \dots,\ x\textbf{\texttt{)}}\tag{procedure invocation}\\
    &\ |\ \qquad \textbf{skip}\ |\ s\ s\tag{statement sequence}\\
    e\quad&::=\quad \overline{n}\ |\ x\ |\ x\textbf{\texttt{[}}e\textbf{\texttt{]}}\ |\ e\ \otimes\ e\ |\ \textbf{empty}\textbf{\texttt{(}}x\textbf{\texttt{)}}\ |\ \textbf{top}\textbf{\texttt{(}}x\textbf{\texttt{)}}\ |\ \textbf{nil}\tag{expression}\\
    \odot\quad&::=\quad \textbf{\texttt{+}}\ |\ \textbf{\texttt{-}}\ |\ \textbf{\texttt{\textasciicircum}}\tag{operator}\\
    \otimes\quad&::=\quad \odot\ |\ \textbf{\texttt{*}}\ |\ \textbf{\texttt{/}}\ |\ \textbf{\texttt{\%}}\ |\ \textbf{\texttt{\&}}\ |\ \textbf{\texttt{|}}\ |\ \textbf{\texttt{\&\&}}\ |\ \textbf{\texttt{||}}\ |\ \textbf{\texttt{<}}\ |\ \textbf{\texttt{\textgreater}}\ |\ \textbf{\texttt{=}}\ |\ \textbf{\texttt{!=}}\ |\ \textbf{\texttt{<=}}\ |\ \textbf{\texttt{\textgreater=}}\tag{operator}
\end{align}

\caption[EBNF grammar for Janus]{EBNF grammar for Janus~\cite{ty:ejanus}}
\label{fig:janus-grammar}
\end{figure}

Janus is a procedural language with locally-invertible program statements and direct access to inverse semantics. There are 3 data types in Janus: plain integers, fixed-size integer arrays and dynamically-sized integer stacks. Integer variables and integer stacks may be declared locally or statically in the global scope, while integer arrays can only be declared statically.

A Janus program consists of a main procedure followed by any number of secondary procedures. The main procedure acts as the starting point of the program and is preceded by declarations of static variables, which serve as the program output upon termination. Secondary procedures may specify parameters which are passed to the callee by reference. Procedures can not return a value but may use output parameters to achieve similar effects. Procedure bodies are made up of one or more program statements, which may be one of several different forms.

A \textit{conditional statement} in Janus has both a branch condition and an exit assertion, both of which are expressions. The branch condition determines which branch of the conditional is executed, while the exit assertion is used to reversibly join the two paths of computation. If the branch condition evaluates to true, the \textbf{then}-branch is executed upon which the exit assertion should also evaluate to true. If the branch condition evaluates to false, the \textbf{else}-branch is executed after which the exit assertion should evaluate to false. If the exit assertion does not match the branch condition, the statement is undefined. See Figure~\ref{fig:conditional} in Chapter~\ref{chp:introduction} for a flowchart illustrating the mechanics of reversible conditionals.

A \textit{loop statement} has both an entry assertion and an exit condition, both of which are expressions. Initially, the entry assertion must evaluate to true after which the \textbf{do}-statement is executed. If the exit condition is then true, the loop terminates, otherwise the \textbf{loop}-statement is executed upon which the entry assertion must now evaluate to false. When executed in reverse, the exit condition serves as the entry assertion and vice versa. Figure~\ref{fig:loop} shows a flowchart illustrating the mechanics of reversible loops.

\begin{figure}[h]
\centering
\begin{subfigure}[b]{.495\columnwidth}
    \centering
    \begin{tikzpicture}[node distance = 1.3cm, auto]
        \node[none] (start) {};
        \node[assertion, right of = start] (e1) {$e_1$};
        \node[block, right of = e1, above of = e1] (s1) {$s_1$};
        \node[block, right of = e1, below of = e1] (s2) {$s_2$};
        \node[conditional, right of = e, node distance = 2.6cm] (e2) {$e_2$};
        \node[none, right of = e2] (end) {};
        
        \draw[arrow] (start) -- node[anchor = south west] {\textbf{t}} (e1);
        \draw[arrow] (e1) |- (s1);
        \draw[arrow] (s1) -| (e2);
        \draw[arrow] (e2) |- node[anchor = south west] {\textbf{f}} (s2);
        \draw[arrow] (s2) -| node[anchor = south west] {\textbf{f}} (e1);
        \draw[arrow] (e2) -- node[anchor = south east] {\textbf{t}} (end);
    \end{tikzpicture}
\end{subfigure}
\begin{subfigure}[b]{.495\columnwidth}
    \centering
    \begin{tikzpicture}[node distance = 1.3cm, auto]
        \node[none] (start) {};
        \node[conditional, right of = start] (e1) {$e_1$};
        \node[block, right of = e1, above of = e1] (s1) {$s_1$};
        \node[block, right of = e1, below of = e1] (s2) {$s_2$};
        \node[assertion, right of = e, node distance = 2.6cm] (e2) {$e_2$};
        \node[none, right of = e2] (end) {};
        
        \draw[arrow] (e1) -- node[anchor = south west] {\textbf{t}} (start);
        \draw[arrow] (s1) -| (e1);
        \draw[arrow] (e2) |- (s1);
        \draw[arrow] (s2) -| node[anchor = south west] {\textbf{f}} (e2);
        \draw[arrow] (e1) |- node[anchor = south west] {\textbf{f}} (s2);
        \draw[arrow] (end) -- node[anchor = south east] {\textbf{t}} (e2);
    \end{tikzpicture}
\end{subfigure}
\caption[Flowchart of a reversible loop statement]{Flowcharts of a reversible loop statement, in both directions of execution~\cite{ty:flowchart, ty:janus}}
\label{fig:loop}
\end{figure}
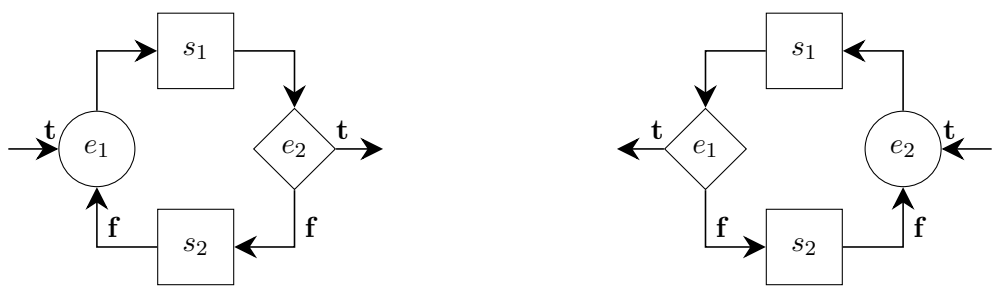

The \textit{stack modification statements}, \textbf{push} and \textbf{pop} are used to manipulate integer stacks in the usual fashion, the only difference being that pushing a variable onto a stack zero-clears the contents of the variable while popping a value into a variable presupposes that the variable is zero-cleared. This means that push and pop are inversions of each other.

A \textit{reversible variable update} in Janus works by updating a variable in the current scope in such a way that the original store remains reachable by subsequent uncomputation. Only updates that are injective in their first argument and have precisely defined inverses are allowed and it is a requirement that the expression being updated with does not in any way depend on the value of the variable being updated (to avoid loss of information). To ensure such an update cannot occur, it is not allowed for the variable identifier on the left side of the update to occur anywhere on the right-hand side. This also mandates a further restriction: no two identifiers may refer to the same location in memory in the same scope (a situation known as aliasing) as this would otherwise be a way to circumvent the aforementioned requirement.

The \textit{local variable block}, denoted by the \textbf{local}/\textbf{delocal} statement, defines a block scope wherein a new local variable is declared and initialized. After the block statement has executed with the new variable in scope, the variable is cleared by means of an exit expression which must evaluate to the value of the variable (otherwise the statement is undefined as it becomes impossible to reversibly clear the memory occupied by the variable).

The \textit{call} and \textit{uncall} statements are used to invoke procedures in the forwards and backwards direction. Arguments are passed by reference and it is a requirement that the same variable is not passed twice in the same procedure invocation to avoid aliasing of the arguments.

An \textit{expression} in Janus can either be a numeric literal, a variable identifier, an array element, a binary expression or a stack expression. Janus uses $0$ to represent the boolean value false, and non-zero to represent true.

\begin{figure}[h]
\centering
\begin{lstlisting}[style = basic, language = janus]
procedure main()
    int n
    int root
    n += 25
    call root(n, root)

procedure root(int n, int roo) //root := floor (sqrt(n))
    local int bit = 1
        from bit = 1
        do skip
        loop call doublebit(bit)
        until (bit * bit) > n
        from (bit * bit) > n
        do uncall doublebit(bit)
           if ((root + bit) * (root + bit)) <= n
           then root += bit
           fi (root / bit) % 2 != 0
        loop skip
        until bit=1
    delocal int bit = 1
    n -= root * root

procedure doublebit(int bit)
    local int z = bit
    bit += z
    delocal int z = bit / 2
\end{lstlisting}
\caption[Example Janus program for computing $\lfloor\sqrt{n}\rfloor$]{Example Janus program for computing $\lfloor\sqrt{n}\rfloor$ from~\cite{cl:janus}}
\label{fig:janus-program}
\end{figure}

Janus is known to be r-Turing complete as it is able to simulate any reversible Turing machine~\cite{ty:ejanus}. An efficient and clean translation from Janus to PISA (See Section~\ref{sec:pisa}) was presented in~\cite{ha:translation} and a partial evaluator for Janus was presented in~\cite{tm:partial}. The reversible control flow constructs used by Janus was explored in detail in~\cite{ty:flowchart}.

\section{Unstructured Janus}

An unstructured version of Janus was used in~\cite{tm:partial} as an intermediate language for polyvariant partial evaluation. Specialization of a program written in an imperative programming language is usually accomplished with polyvariant partial evaluation, which is most suitable for programs with unstructured control flow.

A precursor to the unstructured version of Janus was first presented in~\cite{ty:flowchart} as a reversible flowchart language. \citeauthor{tm:partial} suggests a simple transformation from Janus to a modified version of this flowchart language, before the partial evaluation is applied.

The language uses \textit{paired jumps} to organize the unstructured control flow in a reversible manner: Every jump statement must jump to a \textit{from}-statement which uniquely identifies the origin of the jump, thus reversibly joining the control flow. The language also supports conditional jumps which must then target a conditional from-statement, again for the purpose of reversibly joining the two paths of computation.

Unstructured Janus programs are arranged into a series of basic blocks, each consisting of a label, a from statement, a series of reversible assignments and finally a jump. The first block always starts with a \textit{start} statement and the end of the program is marked with a \textit{return} statement. The language is locally invertible, just like its structured counterpart.

The structured reversible program theorem, by \citeauthor{ty:flowchart} in~\cite{ty:flowchart} proves that such a language is computationally equivalent to its structured counterpart. Figure~\ref{fig:unstructured-program} shows a program for multiplying two odd integers using unstructured Janus.

\begin{figure}[h]
\centering

\begin{lstlisting}[style = basic, language = unstructured]
start:
goto f_2

if 0 = prod from f_2 a_2:
if odd(a) goto t1_3 e1_3

t1_3:
prod += b; t += a / 2; a -= t + 1; t -= a
goto t2_3

e1_3:
t += a / 2; a -= t; t -= a
goto e2_3

if !(prod < b) from t2_3 e2_3:
if a = 0 goto f_11 l_2

l_2:
v += b; b += v; v -= b/2
goto a_2

if prod < b + b from f_11 a_11:
v += b / 2; b -= v; v -= b
if odd(b) goto u_11 a_11

u_11:
return
\end{lstlisting}

\caption[Unstructured Janus program computing the product of two odd numbers]{Unstructured Janus program computing the product of two odd numbers, from~\cite{tm:partial}}
\label{fig:unstructured-program}
\end{figure}

\section{R}
\label{sec:r}

The reversible programming language \textit{R} (not to be confused with the statistical programming language of the same name) is an imperative reversible language developed at MIT in \citeyear{mf:r}~\cite{mf:r}. The syntax of R is a blend of LISP and C - with programs arranged as nested S-expressions but with support for C-like arrays and pointer arithmetics. R is a compiled language, with the only available compiler targeting the Pendulum reversible instruction set (see Section~\ref{sec:pisa}).

\begin{figure}[ht]
\centering

\textbf{R Grammar}

\begin{align}
    prog\quad&::=\quad s^*\tag{program}\\
    s\quad&::=\quad \textbf{\texttt{(}}\textbf{defmain}\ progname\ s^*\textbf{\texttt{)}}\tag{main routine}\\
    &\ |\ \qquad \textbf{\texttt{(}}\textbf{defsub}\ subname\ \textbf{\texttt{(}}name^*\textbf{\texttt{)}}\ s^*\textbf{\texttt{)}}\tag{subroutine}\\
    &\ |\ \qquad \textbf{\texttt{(}}\textbf{defword}\ name\ \overline{n}\textbf{\texttt{)}}\tag{global variable}\\
    &\ |\ \qquad \textbf{\texttt{(}}\textbf{defarray}\ name\ \overline{n}\,^*\textbf{\texttt{)}}\tag{global array}\\
    &\ |\ \qquad \textbf{\texttt{(}}\textbf{call}\ subname\ e^*\textbf{\texttt{)}}\tag{call subroutine}\\
    &\ |\ \qquad \textbf{\texttt{(}}\textbf{rcall}\ subname\ e^*\textbf{\texttt{)}}\tag{reverse-call subroutine}\\
    &\ |\ \qquad \textbf{\texttt{(}}\textbf{if}\ e\ \textbf{then}\ s^*\textbf{\texttt{)}}\tag{conditional}\\
    &\ |\ \qquad \textbf{\texttt{(}}\textbf{for}\ name\ \textbf{\texttt{=}}\ e\ \textbf{to}\ e\ s^*\textbf{\texttt{)}}\tag{loop}\\
    &\ |\ \qquad \textbf{\texttt{(}}\textbf{let}\ \textbf{\texttt{(}}name\ \textbf{\texttt{<-}}\ e\textbf{\texttt{)}}\ s^*\textbf{\texttt{)}}\tag{variable binding}\\
    &\ |\ \qquad \textbf{\texttt{(}}\textbf{printword}\ e\textbf{\texttt{)}}\ |\ \textbf{\texttt{(}}\textbf{println}\textbf{\texttt{)}}\tag{output}\\
    &\ |\ \qquad \textbf{\texttt{(}}loc\ \textbf{\texttt{++)}}\ |\ \textbf{\texttt{(-}}\ loc\textbf{\texttt{)}}\tag{increment/negate}\\
    &\ |\ \qquad \textbf{\texttt{(}}loc\ \textbf{\texttt{<->}}\ loc\textbf{\texttt{)}}\ |\ \textbf{\texttt{(}}loc\ \odot\ e\textbf{\texttt{)}}\tag{swap/update}\\
    loc\quad&::=\quad name\ |\ \textbf{\texttt{(*}}\ e\textbf{\texttt{)}}\ |\ \textbf{\texttt{(}}e\ \textbf{\texttt{\_}}\ e\textbf{\texttt{)}}\tag{location}\\
    e\quad&::=\quad loc\ |\ \textbf{\texttt{(}}e\ \otimes\ e\textbf{\texttt{)}}\ |\ \overline{n}\tag{expression}\\
    \odot\quad&::=\quad \textbf{\texttt{+=}}\ |\ \textbf{\texttt{-=}}\ |\ \textbf{\texttt{\textasciicircum=}}\ |\ \textbf{\texttt{<=<}}\ |\ \textbf{\texttt{>=>}}\tag{update operator}\\
    \otimes\quad&::=\quad \textbf{\texttt{+}}\ |\ \textbf{\texttt{-}}\ |\ \textbf{\texttt{\&}}\ |\ \textbf{\texttt{<{}<}}\ |\ \textbf{\texttt{>{}>}}\ |\ \textbf{\texttt{*/}}\tag{expression operator}\\
    &\ |\ \qquad \textbf{\texttt{=}}\ |\ \textbf{\texttt{<}}\ |\ \textbf{\texttt{>}}\ |\ \textbf{\texttt{<=}}\ |\ \textbf{\texttt{>=}}\ |\ \textbf{\texttt{!=}}\tag{relational operator}
\end{align}

\caption[EBNF grammar for R]{EBNF grammar for R, based on the rules presented in~\cite[Appdx.~C]{mf:reversibility}}
\label{fig:r-grammar}
\end{figure}

Figure~\ref{fig:r-grammar} shows a formal grammar describing the syntax rules of R. An R program consists of any number of statements, but should contain exactly one main routine, defined with the \textbf{defmain} statement. The main routine may invoke subroutines which are defined with the \textbf{defsub} statement. Also a program may make use of globally scoped variables and arrays, defined with the \textbf{defword} and \textbf{defarray} statement. These four types of statements may appear anywhere in a program, but only have an actual effect when appearing as top-level statements.

The \textbf{call} and \textbf{rcall} statements are used to invoke a subroutine in either direction of execution, and correspond to the \textbf{call} and \textbf{uncall} statements of Janus. Arguments are passed by reference, but only parameters bound to variables or memory references may be modified by the callee. Parameters bound to an expression or a constant should retain their value throughout the body of the subroutine to avoid undefined or irreversible behaviour.

The \textbf{if} statement is used for conditional execution. It is a requirement that the value of the conditional expression is the same before and after the conditional statement is executed, otherwise undefined or irreversible behaviour may occur. This limitation guarantees that the condition can be used to determine which branch of computation to follow in either direction of execution. It is equivalent to a Janus conditional with the same expression used as entry condition and exit assertion. A version with an else-branch was also proposed but never implemented in the compiler.

The \textbf{for} statement is used for definite iteration. The iteration variable is given an initial value matching the first expression and is then incremented upon each iteration until the termination value is reached. Both expressions must have the same value before and after the loop is executed to guarantee correct behaviour in both directions of execution. The for-loop may also be used for indefinite iteration by modifying the value of the iteration variable in the loop body - which allows the number of iterations to be determined dynamically as the loop proceeds.

A \textbf{let} statement creates a new local variable, limited in scope to the statements within the let-block. The local variable is initialized to the value of the let-expression and after the block statements have been executed the value of the let-expression should still match the value of the local variable (although they are not required to have the same value as they did initially). This is a requirement for the program to be able to reversibly zero-clear the local variable before it is reclaimed by the system - it is functionally equivalent to a Janus \textbf{local}/\textbf{delocal} block where the entry and exit expressions are the same.

The \textbf{printword} and \textbf{println} statements are used for program output. A \textbf{printword} statement will output the value of the given expression, while the \textbf{println} statement outputs a single line-break delimiter.

\begin{figure}[ht]
\centering
\begin{lstlisting}[style = basic, language = r]
(defsub fib (x1 x2 n)
    (if (n = 0) then
        (x1 += 1)
        (x2 += 1) )
    
    (if (n != 0) then
        (n -= 1)
        (call fib x1 x2 n)
        (x1 += x2)
        (x1 <-> x2)
        (n += 1) ) ) ; Restore value of n for conditional

(defword x1 0)
(defword x2 0)
(defword n 4)
(defmain fibprog (call fib x1 x2 n))
\end{lstlisting}
\caption[Example R program for computing the n\textsuperscript{th} Fibonacci pair]{Example R program for computing the n\textsuperscript{th} Fibonacci pair, adapted from example program in~\cite{ty:ejanus}}
\label{fig:r-program}
\end{figure}

Memory modification in R is done by the increment, negate, swap and update statements. These statements operate on \textit{memory locations} which may be represented either by variable identifiers, by expressions referring to memory addresses or by expressions referring to specific elements of an array (with an underscore representing array indexing). The update statements are subject to the same restrictions as in Janus, namely that the value of the expressions being updated with must not at the same time depend on the memory location being updated. This is necessary to ensure that the update does not erase information. The \textbf{\texttt{<=<}} and \textbf{\texttt{>=>}} operators represent \textit{arithmetic} left and right rotations.

Expressions in R can be either memory locations, numeric literals or binary operations. The supported operators are numerical addition, subtraction and bitwise conjunction (\textbf{\texttt{+}}, \textbf{\texttt{-}}, \textbf{\texttt{\&}}), logical left and right shifts (\textbf{\texttt{<{}<}}, \textbf{\texttt{>{}>}}), relational operators\footnote{As described in~\cite[Appdx.~C]{mf:reversibility}, the R compiler only supports the use of relational operators in conditional expressions but this can be considered a limitation of the implementation, not of the language.} (\textbf{\texttt{=}}, \textbf{\texttt{<}}, \textbf{\texttt{<=}}, \textbf{\texttt{!=}}, \textbf{\texttt{>}}, \textbf{\texttt{>=}}) and fractional product (\textbf{\texttt{*/}}), which is the product of a signed integer and a fixed-precision fraction between $-1$ and $1$.

\section{PISA}
\label{sec:pisa}

The Pendulum microprocessor and the Pendulum ISA (PISA) is a logically reversible computer architecture created at MIT by Carlin James Vieri~\cite{cv:pendulum, cv:engineering, mf:reversibility, cv:zero}. The Pendulum architecture resembles a mix of PDP-8 and RISC and it was the first reversible programmable processor and instruction set.

PISA is a MIPS-like assembly language that has gone through several incarnations. The version presented in this section is known as the \textit{PISA Assembly Language} (PAL) and it is compatible with the Pendulum virtual machine, PendVM~\cite{cr:pendvm}.

\begin{figure}[h]
\centering

\textbf{PISA Grammar}

\begin{align*}
    prog\quad&::=\quad ((l\ \textbf{\texttt{:}})^?\ i)^+\tag{program}\\
    i\quad&::=\quad \textbf{\texttt{ADD}}\ r\ r\ |\ \textbf{\texttt{ADDI}}\ r\ c\ |\ \textbf{\texttt{ANDX}}\ r\ r\ r\ |\ \textbf{\texttt{ANDIX}}\ r\ r\ c\tag{instruction}\\
    &\ |\ \qquad \textbf{\texttt{NORX}}\ r\ r\ r\ |\ \textbf{\texttt{NEG}}\ r\ |\ \textbf{\texttt{ORX}}\ r\ r\ r\ |\ \textbf{\texttt{ORIX}}\ r\ r\ r\ |\ \textbf{\texttt{RL}}\ r\ c\\
    &\ |\ \qquad \textbf{\texttt{RLV}}\ r\ r\ |\ \textbf{\texttt{RR}}\ r\ c\ |\ \textbf{\texttt{RRV}}\ r\ r\ |\ \textbf{\texttt{SLLX}}\ r\ r\ c\ |\ \textbf{\texttt{SLLVX}}\ r\ r\ r\\
    &\ |\ \qquad \textbf{\texttt{SRAX}}\ r\ r\ c\ |\ \textbf{\texttt{SRAVX}}\ r\ r\ r\ |\ \textbf{\texttt{SRLX}}\ r\ r\ c\ |\ \textbf{\texttt{SRLVX}}\ r\ r\ r\\
    &\ |\ \qquad \textbf{\texttt{SUB}}\ r\ r\ |\ \textbf{\texttt{XOR}}\ r\ r\ |\ \textbf{\texttt{XORI}}\ r\ c\ |\ \textbf{\texttt{BEQ}}\ r\ r\ l\ |\ \textbf{\texttt{BGEZ}}\ r\ l\\
    &\ |\ \qquad \textbf{\texttt{BGTZ}}\ r\ l\ |\ \textbf{\texttt{BLEZ}}\ r\ l\ |\ \textbf{\texttt{BLTZ}}\ r\ l\ |\ \textbf{\texttt{BNE}}\ r\ r\ l\ |\ \textbf{\texttt{BRA}}\ l\\
    &\ |\ \qquad \textbf{\texttt{EXCH}}\ r\ r\ |\ \textbf{\texttt{SWAPBR}}\ r\ |\ \textbf{\texttt{RBRA}}\ l\ |\ \textbf{\texttt{START}}\ |\ \textbf{\texttt{FINISH}}\\
    &\ |\ \qquad \textbf{\texttt{DATA}}\ c\\
    c\quad&::=\quad \cdots\ |\ \textbf{\texttt{-1}}\ |\ \textbf{\texttt{0}}\ |\ \textbf{\texttt{1}}\ |\ \cdots\tag{immediate}
\end{align*}

\vspace{2mm}
\textbf{Syntax Domains}

\begin{align*}
    prog &\in \text{Programs}  & i &\in \text{Instructions}\\
    r    &\in \text{Registers} & l &\in \text{Labels}
\end{align*}

\caption{Syntax domains and EBNF grammar for PISA}
\label{fig:pal-grammar}
\end{figure}

In a conventional processor, the rules governing control flow are quite simple: After each instruction, add $1$ to the program counter. In case of a jump, set the program counter to the address of the label being jumped to. In a reversible processor like Pendulum, these rules are much more involved since simply overwriting the contents of the program counter would constitute a loss of information which break reversibility.

The Pendulum processor uses three special-purpose registers for control flow logic:

\begin{enumerate}
    \item The \textit{program counter} (PC) for storing the address of the current instruction
    \item The \textit{branch register} (BR) for storing jump offsets
    \item The \textit{direction bit} (DIR) for keeping track of the execution direction
\end{enumerate}

After each instruction, if the branch register is zero, we simply add the direction bit to the program counter. The direction bit is either $1$ or $-1$ depending on the direction of execution so this corresponds to regular stepwise execution in either direction.

If the branch register is not zero, the product of the branch register and the direction bit is added to the program counter. When a PISA program is assembled to machine code, the target labels of each of the jump instructions are replaced with \textit{relative offsets}. When a jump instruction is then executed, the relative offset is placed in the branch register and when the PC is updated, control flow jumps to the target label. Using \textit{paired branches}, the PISA programmer can clear the branch register after a jump by always jumping only to jump instruction that points back to the original jump. This has the effect of adding the negation of the relative offset to the branch register, thereby zero-clearing it.

Aside from the usual conditional jump instructions (Branch-if-equal, branch-if-zero et cetera), PISA also contains the unconditional jump instruction \textbf{\texttt{BRA}} and the unconditional reverse-jump instruction \textbf{\texttt{RBRA}} which also flips the direction bit and can therefore be used to implement uncall or reverse-call functionality. When the direction bit is $-1$, the instructions are inverted so that addition becomes subtraction, left-rotation becomes right-rotation and so on. See Figure~\ref{fig:pisa-inversion} for a table illustrating how PISA instructions are inverted when the execution direction is flipped.

\begin{figure}[h]
\centering

\begin{tabular}{ll}
    \toprule
    \multicolumn{1}{c}{$i$} & \multicolumn{1}{c}{$i^{-1}$}\\
    \midrule
    \textbf{\texttt{ADD}} $r_1$ $r_2$ & \textbf{\texttt{SUB}} $r_1$ $r_2$ \\
    \textbf{\texttt{SUB}} $r_1$ $r_2$ & \textbf{\texttt{ADD}} $r_1$ $r_2$ \\
    \textbf{\texttt{ADDI}} $r$ $c$ & \textbf{\texttt{ADDI}} $r$ $-c$ \\
    \textbf{\texttt{RL}} $r$ $c$ & \textbf{\texttt{RR}} $r$ $c$ \\
    \textbf{\texttt{RR}} $r$ $c$ & \textbf{\texttt{RL}} $r$ $c$ \\
    \textbf{\texttt{RLV}} $r_1$ $r_2$ & \textbf{\texttt{RRV}} $r_1$ $r_2$ \\
    \textbf{\texttt{RRV}} $r_1$ $r_2$ & \textbf{\texttt{RLV}} $r_1$ $r_2$ \\
    \bottomrule
\end{tabular}

\caption[PISA inversion rules]{Inversion rules for PISA instructions, all other instructions are self-inverse}
\label{fig:pisa-inversion}
\end{figure}

PISA also has the \textbf{\texttt{SWAPBR}} instruction which affords direct control over the contents of the branch register (but crucially, not the PC directly) and makes it possible to implement dynamic jumps such as switch/case structures or function pointers. \textbf{\texttt{SWAPBR}} can also be used to allow incoming jumps from more than one location.

The special instructions \textbf{\texttt{START}} and \textbf{\texttt{FINISH}} are used to mark the beginning and end of a PISA program while the memory exchange instruction \textbf{\texttt{EXCH}} provides simultaneous reversible memory-read and memory-write functionality. The \textbf{\texttt{DATA}} instruction stores an immediate value in the corresponding memory cell and can be used to mark the static storage space of a program.

The remaining instructions are similar to those of other RISC processors and implement various register update functionality (bitwise-AND, bitwise-XOR and so on) albeit in a reversible manner. For example, bitwise-AND is performed with the \textbf{\texttt{ANDX}} instruction which XORs the resulting value into a third register to ensure reversibility.

Figure~\ref{fig:pisa-program} shows an example PISA program. The design of the Pendulum control flow logic is based in part on work by Cezzar~\cite{rc:reverse} and Hall~\cite{jh:reversible}. A complete formalization of the PISA language and the Pendulum machine was given in~\cite{ha:architecture} and a translation from Janus to PISA was presented in~\cite{ha:translation}. PISA is also the target language of the R compiler~\cite{mf:r, mf:reversibility} and in this thesis we use PISA as the target language for the translation presented in Chapter~\ref{chp:compilation}.

\begin{figure}[!ht]
\centering
\begin{lstlisting}[style = basic, language = pisa]
;; Call Fall:
10: ADDI $4 1000   ;h += 1000
11: ADDI $5 5      ;tend += 5
12: BRA 16         ;call

;; Uncall Fall:
18: ADDI $6 40     ;v += 40
19: ADDI $7 4
20: ADDI $5 4
21: RBRA 7         ;uncall

;; Subroutine Fall:
27: BRA 9
28: SWAPBR $8      ;br <=> rtn
29: NEG $8         ;rtn=-rtn
30: BGTZ $7 5      ;t > 0?
31: ADDI $6 10     ;v += 10
32: ADDI $4 5      ;h += 5
33: SUB $4 $6      ;h -= v
34: ADDI $7 1      ;t += 1
35: BNE $7 $5 -5   ;t != tend?
36: BRA -9
\end{lstlisting}
\caption[Example PISA program for simulating free-falling objects]{Example PISA program for simulating free-falling objects, from~\cite{ha:architecture}}
\label{fig:pisa-program}
\end{figure}

\section{BobISA}
\label{sec:bob}

The reversible computer architecture Bob and its instruction set BobISA were created at the University of Copenhagen by \citeauthor{mt:bob}~\cite{mt:bob, mkc:reversible}. Bob is a \textit{Harvard architecture} which is characterized by having separate storage for instructions and data\footnote{As opposed to a \textit{von Neumann architecture} which does not distinguish between program instructions and data.}.

BobISA was designed to be sufficiently expressive to serve as the target for high-level compilers while still being relatively straightforward to implement in hardware. BobISA consists of 17 instructions and is known to be \textit{r-Turing complete}~\cite{mt:bob}.

\vspace{0.5cm}

\begin{figure}[h]
\centering

\textbf{BobISA Grammar}

\begin{align*}
    prog\quad&::=\quad i^+\tag{program}\\
    i\quad&::=\quad \textbf{\texttt{ADD}}\ r\ r\ |\ \textbf{\texttt{SUB}}\ r\ r\ |\ \textbf{\texttt{ADD1}}\ r\ |\ \textbf{\texttt{SUB1}}\ r\tag{instruction}\\
    &\ |\ \qquad \textbf{\texttt{NEG}}\ r\ |\ \textbf{\texttt{XOR}}\ r\ r\ |\ \textbf{\texttt{XORI}}\ r\ c\ |\ \textbf{\texttt{MUL2}}\ r\\
    &\ |\ \qquad \textbf{\texttt{DIV2}}\ r\ |\ \textbf{\texttt{BGEZ}}\ r\ o\ |\ \textbf{\texttt{BLZ}}\ r\ o\ |\ \textbf{\texttt{BEVN}}\ r\ o\\
    &\ |\ \qquad \textbf{\texttt{BODD}}\ r\ o\ |\ \textbf{\texttt{BRA}}\ o\ |\ \textbf{\texttt{SWBR}}\ r\ |\ \textbf{\texttt{RSWB}}\ r\\
    &\ |\ \qquad \textbf{\texttt{EXCH}}\ r\ r\\
    c\quad&::=\quad \cdots\ |\ \textbf{\texttt{-1}}\ |\ \textbf{\texttt{0}}\ |\ \textbf{\texttt{1}}\ |\ \cdots\tag{immediate}\\
    o\quad&::=\quad \textbf{\texttt{-128}}\ |\ \cdots\ |\ \textbf{\texttt{0}}\ |\ \cdots\ |\ \textbf{\texttt{127}}\tag{offset}
\end{align*}

\caption{EBNF grammar for BobISA}
\label{fig:bob-grammar}
\end{figure}

The control flow logic of Bob is identical to that of PISA, with a few caveats:

\begin{itemize}
    \renewcommand\labelitemi{\normalfont\bfseries \textendash}
    \item There are only 8-bits to store jump offsets, so a plain jump cannot be of more than 127 lines.
    \item The \textbf{\texttt{SWBR}} instruction which is similar to the \textbf{\texttt{SWAPBR}} instruction of PISA, can be used for jump offsets longer than 127.
    \item BobISA also has the \textbf{\texttt{RSWB}} instruction which flips the direction bit in addition to swapping out the branch register.
\end{itemize}

While the jump targets in the BobISA grammar in Figure~\ref{fig:bob-grammar} are represented in terms of offsets, a construction similar to that of PISA could be used, where jumps are specified with instruction labels that are then converted to offsets during program assembly.

The remaining instructions are self-explanatory and most of them have PISA equivalents, with the exception of \textbf{\texttt{MUL2}} and \textbf{\texttt{DIV2}}. These instructions operate on 4-bit two's-complement numbers and will either double or halve the value of a given register. To avoid overflow and division of odd numbers, these instructions are only well-defined for a subset of the representable values as illustrated in Figure~\ref{fig:bob-muldivinv}. Input values outside of this subset are mapped to output in such a way that reversibility is preserved. Figure~\ref{fig:bob-muldivinv} also shows the inversion rules for those BobISA instructions that are not self-inverse. Like PISA, the inverse semantics of each instruction is used when the processor is running in reverse.

\begin{figure}[ht]
\centering

\begin{tabular}{ccc}
\begin{tabular}[t]{cc}
    \toprule
    $x$ & $\mathrm{MUL2}(x)$\\
    \midrule
    -4 & -8 \\
    -3 & -6 \\
    -2 & -4 \\
    -1 & -2 \\
    { }0 & { }0 \\
    { }1 & { }2 \\
    { }2 & { }4 \\
    { }3 & { }6 \\
    \bottomrule
\end{tabular}
\qquad
\begin{tabular}[t]{cc}
    \toprule
    $x$ & $\mathrm{DIV2}(x)$\\
    \midrule
    -8 & -4 \\
    -6 & -3 \\
    -4 & -2 \\
    -2 & -1 \\
    { }0 & { }0 \\
    { }2 & { }1 \\
    { }4 & { }2 \\
    { }6 & { }3 \\
    \bottomrule
\end{tabular}
\qquad
\begin{tabular}[t]{ll}
    \toprule
    \multicolumn{1}{c}{$i$} & \multicolumn{1}{c}{$i^{-1}$}\\
    \midrule
    \textbf{\texttt{ADD}} $r_1$ $r_2$ & \textbf{\texttt{SUB}} $r_1$ $r_2$ \\
    \textbf{\texttt{SUB}} $r_1$ $r_2$ & \textbf{\texttt{ADD}} $r_1$ $r_2$ \\
    \textbf{\texttt{ADD1}} $r$ & \textbf{\texttt{SUB1}} $r$ \\
    \textbf{\texttt{SUB1}} $r$ & \textbf{\texttt{ADD1}} $r$ \\
    \textbf{\texttt{MUL2}} $r$ & \textbf{\texttt{DIV2}} $r$ \\
    \textbf{\texttt{DIV2}} $r$ & \textbf{\texttt{MUL2}} $r$ \\
    \bottomrule
\end{tabular}
\end{tabular}

\caption[BobISA inversion rules]{Tables showing well-defined inputs and outputs for \textbf{\texttt{MUL2}} and \textbf{\texttt{DIV2}} instructions as well as the inversion rules for BobISA instructions}
\label{fig:bob-muldivinv}
\end{figure}

A complete low-level design with schematics and HDL programs was developed for the Bob architecture. Only 473 reversible gates are required to construct a Bob processor, totalling only 6328 transistors~\cite{mt:bob}. A translation from the reversible functional language RFUN to BobISA was presented in~\cite{jsk:translation}.
\newpage

\chapter{The \textsc{Roopl} Language}
\label{chp:roopl}

The \textit{Reversible Object-Oriented Programming Language} (ROOPL) is, to our knowledge, the first reversible programming language with built-in support for object-oriented programming and user-defined types. ROOPL is statically typed and supports inheritance, encapsulation and subtype-polymorphism via dynamic dispatch. ROOPL is purely reversible, in the sense that no computation history is required for backwards execution. Rather, each component of a ROOPL program is locally invertible at no extra cost to program size. The basic components of the language, such as control flow structures and variable updates draw heavy inspiration from the reversible imperative language Janus~\cite{ty:janus, ty:ejanus}, however the overall structure of a ROOPL program differs vastly from that of a Janus program.

\vspace{3mm}
\begin{figure}[ht]
\centering
\begin{lstlisting}[style = basic, language = roopl]
class Program
    int result
    int n
    
    method main()
        n ^= 4
        
        construct Fib f
            //Compute-copy-uncompute
            call f::fib(n)
            call f::get(result)
            uncall f::fib(n)
        destruct f

class Fib
    int x1
    int x2
    
    method fib(int n)
        if n = 0 then
            x1 ^= 1
            x2 ^= 1
        else
            n -= 1
            call fib(n)
            x1 += x2
            x1 <=> x2
        fi x1 = x2
    
    method get(int out)
        out ^= x2
\end{lstlisting}
\caption[Example ROOPL program computing the n\textsuperscript{th} Fibonacci pair]{Example ROOPL program computing the n\textsuperscript{th} Fibonacci pair, adapted from example program in~\cite{ty:ejanus}}
\label{fig:fibonacci-program}
\end{figure}

\section{Syntax}
\label{sec:syntax}

A ROOPL program consists of one or more class definitions, each of which may contain any number of member variables and one or more methods. Each program should contain exactly one class with a nullary method named \textit{main} which acts as the program entry point. This class will be instantiated when the program starts, and the fields of this object will act as the output of the program in much the same way that the variable store acts as the output of a Janus program.

\begin{figure}[ht]
\centering

\vspace{3mm}
\textbf{\textsc{Roopl} Grammar}

\begin{align}
    prog\quad&::=\quad cl^+\tag{program}\\
    cl\quad&::=\quad \textbf{class}\ c\ (\textbf{inherits}\ c)^?\ (t\ x)^*\ m^+\tag{class definition}\\
    t\quad&::=\quad \textbf{int}\ |\ c\tag{data type}\\
    m\quad&::=\quad \textbf{method}\ q\textbf{\texttt{(}}t\ x,\ \dots,\ t\ x\textbf{\texttt{)}}\ s\tag{method}\\
    s\quad&::=\quad x\ \odot\textbf{\texttt{=}}\ e\ |\ x\ \textbf{\texttt{<=>}}\ x\tag{assignment}\\
    &\ |\ \qquad \textbf{if}\ e\ \textbf{then}\ s\ \textbf{else}\ s\ \textbf{fi}\ e\tag{conditional}\\
    &\ |\ \qquad \textbf{from}\ e\ \textbf{do}\ s\ \textbf{loop}\ s\ \textbf{until}\ e\tag{loop}\\
    &\ |\ \qquad \textbf{construct}\ c\ x\quad s\quad\textbf{destruct}\ x\tag{object block}\\
    &\ |\ \qquad \textbf{call}\ q\textbf{\texttt{(}}x,\ \dots,\ x\textbf{\texttt{)}}\ |\ \textbf{uncall}\ q\textbf{\texttt{(}}x,\ \dots,\ x\textbf{\texttt{)}}\tag{local method invocation}\\
    &\ |\ \qquad \textbf{call}\ x\textbf{\texttt{::}}q\textbf{\texttt{(}}x,\ \dots,\ x\textbf{\texttt{)}}\ |\ \textbf{uncall}\ x\textbf{\texttt{::}}q\textbf{\texttt{(}}x,\ \dots,\ x\textbf{\texttt{)}}\tag{method invocation}\\
    &\ |\ \qquad \textbf{skip}\ |\ s\ s\tag{statement sequence}\\
    e\quad&::=\quad \overline{n}\ |\ x\ |\ \textbf{\texttt{nil}}\ |\ e\ \otimes\ e\tag{expression}\\
    \odot\quad&::=\quad \textbf{\texttt{+}}\ |\ \textbf{\texttt{-}}\ |\ \textbf{\texttt{\^}}\tag{operator}\\
    \otimes\quad&::=\quad \odot\ |\ \textbf{\texttt{*}}\ |\ \textbf{\texttt{/}}\ |\ \textbf{\texttt{\%}}\ |\ \textbf{\texttt{\&}}\ |\ \textbf{\texttt{|}}\ |\ \textbf{\texttt{\&\&}}\ |\ \textbf{\texttt{||}}\ |\ \textbf{\texttt{<}}\ |\ \textbf{\texttt{\textgreater}}\ |\ \textbf{\texttt{=}}\ |\ \textbf{\texttt{!=}}\ |\ \textbf{\texttt{<=}}\ |\ \textbf{\texttt{\textgreater=}}\tag{operator}
\end{align}

\vspace{2mm}
\textbf{Syntax Domains}

\begin{align*}
    prog &\in \text{Programs} & s &\in \text{Statements}      & n &\in \text{Constants} \\
    cl   &\in \text{Classes}  & e &\in \text{Expressions}     & x &\in \text{VarIDs}    \\
    t    &\in \text{Types}    & \odot &\in \text{ModOps}      & q &\in \text{MethodIDs} \\
    m    &\in \text{Methods}  & \otimes &\in \text{Operators} & c &\in \text{ClassIDs}
\end{align*}

\caption{Syntax domains and EBNF grammar for ROOPL}
\label{fig:roopl-grammar}
\end{figure}

A \textit{class definition} consists of the keyword \textbf{class} followed by the class name. If the class is a subclass of another, it is specified with the keyword \textbf{inherits} followed by the name of the base class. Next, any number of class fields are declared, each of which may be either integers or references to other objects (these are the only types in ROOPL). Finally, each class definition contains at least one method which is defined with the keyword \textbf{method} followed by the method name, a comma-separated list of parameters and the method body. A class must have at least one method, as method calls are the only mechanism of interfacing with an object.

A \textit{reversible assignment} in ROOPL uses the same C-like syntax as a reversible assignment in Janus. A variable can be updated either through addition (\textbf{\texttt{+=}}), subtraction (\textbf{\texttt{-=}}) or bitwise XOR (\textbf{\texttt{\textasciicircum=}}). It is only possible to reversibly update the value of some variable $x$ by some expression $e$ in this manner, if the value of $e$ does not depend, in any way, on the value of $x$. We can enforce this limitation by explicitly disallowing any occurrences of the identifier $x$ in the expression $e$, but this is only sufficient if we can also guarantee that no other identifiers refer to the same location in memory as $x$ (See Section~\ref{sec:argument-aliasing}).

A \textit{variable swap} denoted by the token \textbf{\texttt{<=>}} swaps the value of two integer variables or two object references. This was supported in Janus as syntactic sugar for the statement sequence:
\begin{equation*}
    x_1\ \textbf{\texttt{\textasciicircum=}}\ x_2 \quad x_2\ \textbf{\texttt{\textasciicircum=}}\ x_1 \quad x_1\ \textbf{\texttt{\textasciicircum=}}\ x_2
\end{equation*}
which achieves the same effect as $x_1\ \textbf{\texttt{<=>}}\ x_2$, given that $x_1$ and $x_2$ are both integers~\cite{ty:janus}. In ROOPL, we might wish to swap two object references, for which the XOR operation is undefined, so the swap statement has been made explicit in the language.

Loops and conditional statements are syntactically (and semantically) identical to Janus loops and Janus conditionals. The use of assertions at control flow join points ensure that we can execute these statements in reverse, in a deterministic manner.

An \textit{object block} denotes the instantiation and lifetime of a ROOPL object. The statement consist of the keyword \textbf{construct} followed by a class name and a variable identifier. Then follows the block statement $s$ within which the newly created object will be accessible, and finally the keyword \textbf{destruct} followed by the object identifier signifies the end of the object block.

A \textit{method invocation} may refer either to a local method or to a method in another object - both variants can be both called and uncalled. An \textit{expression} may be either a constant, a variable, the special value \textbf{nil} or a binary expression.

\section{Argument Aliasing}
\label{sec:argument-aliasing}

To avoid situations where multiple identifiers refer to the same memory location within the same scope, known as \textit{aliasing}, we must place some restrictions on method invocations. One source of aliasing occurs when the same identifier is passed to more than one parameter of a method:

\vspace{4mm}
\begin{lstlisting}[style = basic, language = roopl]
    method foo(int a)
        call bar(a, a)
    
    method bar(int x, int y)
        x -= y //Irreversible update!
\end{lstlisting}
\vspace{4mm}

Such situations are easily avoided by prohibiting method calls with the same identifier passed to more than one parameter, which is the same approach used in Janus. Another, similar source of aliasing is when a field of an object is passed to a parameter of a method of that same object:

\vspace{3mm}
\begin{lstlisting}[style = basic, language = roopl]
    class Object
        int a
        
        method main()
            a += 5
            call foo(a)
        
        method foo(int b)
            a -= b //Irreversible update!
\end{lstlisting}

In this case we can disallow object fields as arguments to local methods, and since the object field is already in scope in the callee, there is little point in also passing it as an argument. ROOPL uses two separate statements to distinguish between local and non-local method invocations, so it is a simple matter of prohibiting object fields as arguments to local call statements.

Finally, we must make sure that non-local method invocations are indeed non-local, which might not be the case if an object has obtained a reference to itself. We can avoid such a situation by disallowing non-local method calls to some object $x$ which also passes $x$ as an argument.

\section{Parameter Passing Schemes}
\label{sec:parameter-passing}

The most common parameter passing modes and their implications for reversible languages were briefly discussed in~\cite{ty:ejanus} while a more in-depth investigation was performed in~\cite{tm:parameters}.  The common call-by-value scheme is generally not suitable for reversible languages since the values accumulated in the function parameters after a function has executed, must be disposed of somehow when the function returns, which would result in a loss of information. It is also difficult to reconstruct multiple arguments given only a single return value, which is the main reason that Janus uses the call-by-reference strategy. With this approach, a function can simply store results in the parameter variables and sidestep traditional single return values altogether. The values in the parameters are handed back to the caller instead of being erased.

Another approach, which is likely simpler to implement in practice, is call-by-value-result presented in~\cite{tm:parameters}. Call-by-value-result involves swapping the function arguments into local variables in the called procedure, and copying them back after the body has been executed. This approach hinges upon the callee not being able to alter the argument variables other than through the local copies, which can only occur if more than one identifier, referring to the same argument, is in scope.

Call-by-reference and call-by-value-result are semantically equivalent parameter passing schemes in the absence of aliasing~\cite{tm:parameters}, and therefore either scheme can be used. The operational semantics of ROOPL (Section~\ref{sec:semantics}) uses call-by-reference.

\section{Object Model}
\label{sec:object-model}

ROOPL is a class-based programming language, it is based on the notion of \textit{classes} that serve as blueprints for specific objects or class instances. Alternatively, a language may allow objects to serve as blueprints for other objects - this is known as prototype-based programming. Prototype-based programming is dominated by dynamically-typed\footnote{For an example of a statically typed language with a prototype-based object model, see Omega~\cite{gb:omega}.}, interpreted languages (examples include JavaScript and Lua). While there is no immediate reason to believe that dynamic typing is not a feasible strategy for a reversible programming language, it is as of yet an unexplored notion.

Some OOP languages have very intricate object models - Java includes support for access modifiers, static methods and fields, final classes (that may not be subclassed), final methods (that cannot be overridden in a subclass) and both implementation inheritance and interface inheritance. \CC\ supports friend classes, virtual and non-virtual methods, abstract methods, private inheritance and multiple inheritance.

These features facilitate the creation of very rich models and interfaces but they are less interesting from our perspective: implementation on a reversible machine. The rules imposed by these features on the classes of a program are generally enforced at compile-time - wholly independently from the target architecture and its limitations (with the exception of dynamic dispatch which has to be handled at runtime).

The object model of ROOPL is therefore very simple compared to these languages - introducing access modifiers or static methods to ROOPL is possible but would be a meaningless venture as the implementation of such features would be identical for an irreversible language. The ROOPL object model is based on the following key points:

\begin{itemize}
    \renewcommand\labelitemi{\normalfont\bfseries \textendash}
    \item All class fields are protected, they may be accessed only from within class methods and subclass methods
    \item All class methods are public, they may be accessed from other objects
    \item All class methods are virtual and may be overridden in a subclass (but only by a method with the same type signature, there is no support for method overloading)
    \item A class may inherit only from a single base class (\textit{single inheritance})
    \item Any method that takes an object reference of some type $\tau$ also works when passed a reference of type $\tau'$ if $\tau'$ is a subclass of $\tau$ (\textit{subtype polymorphism})
    \item Local method calls are statically dispatched (\textit{closed recursion}), only method calls to other objects are dynamically dispatched
\end{itemize}

Note that the single inheritance object model of ROOPL still allows for inheritance hierarchies of arbitrary depth (known as \textit{multi-level} inheritance).

\section{Object Instantiation}
\label{sec:object-instantiation}

In irreversible OOP languages, object instantiation is typically accomplished in two or three general steps:

\begin{enumerate}
    \item A suitable amount of memory is reserved for the object
    \item All fields are initialized to some neutral value
    \item The class constructor is executed, establishing the class invariants of the object
\end{enumerate}

When the program (or the garbage collector) deallocates the object, the memory is (typically) simply marked as unused. Any leftover values from the internal state of the object will be irreversibly overwritten if/when another object is initialized in the same part of memory later on. In a reversible language we cannot clear leftover values in memory like this as that would constitute a loss of information.

Instead we require unused memory to already be zero-cleared at the time of object creation, so the fields of each new object have a known initial value. The only way to achieve this reversibly is to uncompute all the state accumulated inside an object before it is deallocated, returning all fields to the value zero. This cannot be done automatically so this responsibility lies with the program itself.

\begin{figure}[ht]
\centering

\begin{lstlisting}[style = basic, language = roopl]
class Object
    int data
    
    method add5()
        data += 5
    
    method get(int out)
        out ^= data

class Program
    int result
    
    method main()
        construct Object obj      //Allocate object
            call obj::add5()      //Perform computation
            call obj::get(result) //Fetch result
            uncall obj::add5()    //Uncompute internal state
        destruct obj              //Reversibly deallocate object
\end{lstlisting}

\caption{Simple example program illustrating the mechanics of an object block}
\label{fig:obj-program}
\end{figure}

A ROOPL object exists only within a \textbf{construct}/\textbf{destruct} block. Consider the statement:
\begin{equation*}
    \textbf{construct}\ c\ x\quad s\quad\textbf{destruct}\ x
\end{equation*}
the mechanics of such a statement are as follows:

\begin{enumerate}
    \item Memory for an object of class $c$ is allocated. All fields are automatically zero-initialized by virtue of residing in already zero-cleared memory.
    \item The block statement $s$ is executed, with the name $x$ representing a reference to the newly allocated object.
    \item The reference $x$ may be modified by swapping its value with that of other references of the same type, but it should be restored to its original value within the statement block $s$, otherwise the meaning of the object block is undefined.
    \item Any state that is accumulated within the object should be cleared or uncomputed before the end of the statement is reached, otherwise the meaning of the object block is undefined.
    \item The zero-cleared memory is reclaimed by the system.
\end{enumerate}

If the fields of the object are not zero-cleared after the block statement, it becomes impossible for the system to reversibly reclaim the memory occupied by the object. It is up to the program to maintain this invariant.

\section{Inheritance Semantics}
\label{sec:inheritance-semantics}

Before we can define the type system and formal semantics of the language, we need a precise definition of the object model as described in Section~\ref{sec:object-model} and Section~\ref{sec:object-instantiation}. Given the \textit{dynamic type} of some object, we wish to determine the class fields and class methods of the object such that inherited fields and methods are included, unless overridden by the derived class.

\begin{figure}[ht]
\fontsize{10pt}{12pt}\selectfont
\centering

\begin{equation*}
    \mathrm{gen}(\overbrace{cl_1,\ \dots,\ cl_n}^p) \quad = \quad \overbrace{\big[\ \alpha(cl_1)\ \mapsto\ \beta(cl_1),\ \dots,\ \alpha(cl_n)\ \mapsto\ \beta(cl_n)\ \big]}^{\Gamma}
\end{equation*}

\vspace{-0.2cm}

\begin{equation*}
    \alpha\big(\ \textbf{class}\ c\ \cdots\ \big)\ =\ c \qquad \qquad \beta\big(\ cl\ \big)\ =\ \big(\ \mathrm{fields}(cl),\ \mathrm{methods}(cl)\ \big)
\end{equation*}

\caption[Definition of function \textit{gen}]{Definition of function \textit{gen}, for constructing the class map of a given program}
\label{fig:class-map}
\end{figure}

To this end, we define the \textit{class map} $\Gamma$ of a program $p$ as a finite map from class identifiers (type names) to tuples of the method and field declarations of that class. The application of a class map $\Gamma$ to some class identifier $cl$ is denoted $\Gamma(cl)$. Figure~\ref{fig:class-map} shows the definition of function \textit{gen}, which is used to construct the class map of a program.

\begin{figure}[ht]
\fontsize{10pt}{12pt}\selectfont
\centering

\begin{equation*}
    \mathrm{fields}(cl)\ =\ \begin{cases} \eta(cl) &\mbox{if } cl\ \sim\ \big[\ \textbf{class}\ c\ \cdots\ \big]\\ 
\eta(cl)\ \cup\ \mathrm{fields}\big(\ \alpha^{-1}(c')\ \big) &\mbox{if } cl\ \sim\ \big[\ \textbf{class}\ c\ \textbf{inherits}\ c'\ \cdots\ \big]\end{cases}
\end{equation*}

\begin{equation*}
    \mathrm{methods}(cl)\ =\ \begin{cases} \delta(cl) &\mbox{if } cl\ \sim\ \big[\ \textbf{class}\ c\ \cdots\ \big]\\ 
\delta(cl)\ \uplus\ \mathrm{methods}\big(\ \alpha^{-1}(c')\ \big) &\mbox{if } cl\ \sim\ \big[\ \textbf{class}\ c\ \textbf{inherits}\ c'\ \cdots\ \big]\end{cases}
\end{equation*}

\begin{equation*}
    A\ \uplus\ B \quad \overset{\textbf{def}}{\scalebox{1.2}{=}} \quad A\ \cup\ \Big\{\ m\ \in\ B \quad \Big| \quad \nexists\ m'\ \big(\ \zeta(m')\ =\ \zeta(m)\ \wedge\ m'\ \in\ A\ \big)\ \Big\}
\end{equation*}

\vspace{-0.4cm}

\begin{equation*}
    \zeta\big(\ \textbf{method}\ q\ \textbf{(}\cdots\textbf{)}\ s\ \big)\ =\ q \qquad \qquad\eta\big(\ \textbf{class}\ c\ \cdots\ \overbrace{t_1\ f_1\ \cdots\ t_n\ f_n}^{\mathit{fs}}\ \cdots\ \big)\ =\ \mathit{fs}
\end{equation*}

\vspace{-0.2cm}

\begin{equation*}
    \delta\big(\ \textbf{class}\ c\ \cdots\ \overbrace{\textbf{method}\ q_1\ \textbf{(}\cdots\textbf{)}\ s_1\ \cdots\ \textbf{method}\ q_n\ \textbf{(}\cdots\textbf{)}\ s_n}^{\mathit{ms}}\ \cdots\ \big)\ =\ \mathit{ms}
\end{equation*}

\caption{Definition of functions for modelling class inheritance}
\label{fig:class-map-aux}
\end{figure}

Figure~\ref{fig:class-map-aux} shows the definition of the functions \textit{fields} and \textit{methods} which determines the class fields and class methods for a given class. The set operation $\uplus$ implements method overriding by dropping methods from the base class if a method with the same name exists in the derived class.

\section{Type System}
\label{sec:type-system}

The type system of ROOPL is specified by the syntax-directed typing rules shown in the following sections. There are three main type judgments covering expressions, statements and whole ROOPL programs. The inference rules are presented in the style of Winskell~\cite{gw:semantics} and are arranged in such a way that a complete type derivation can only be constructed for well-typed programs. The next section establishes the notation and presents auxiliary definitions.

\subsection{Preliminaries}
\label{subsec:type-prelims}

The set of types in ROOPL is given by the grammar:

\begin{equation*}
    \tau\ ::=\ \textbf{int}\ \big|\ c\ \in\ \text{ClassIDs}
\end{equation*}

A \textit{type environment} $\Pi$ is a finite map from variable identifiers to types. The application of a type environment $\Pi$ to some identifier $x$ is denoted by $\Pi(x)$. Update $\Pi'\ =\ \Pi[x\ \mapsto \tau]$ defines a type environment $\Pi'$ s.t. $\Pi'(x)\ =\ \tau$ and $\Pi'(y)\ =\ \Pi(y)$ if $y\ \neq\ x$. The empty type environment is written $[\ ]$. The function $\mathit{vars}\ :\ \mathit{Expressions}\ \rightarrow\ \mathit{VarIDs}$, is given by the following recursive definition:

\vspace{-4mm}
\begin{alignat*}{2}
    &\mathrm{vars}(\overline{n})\ &&=\ \varnothing\\
    &\mathrm{vars}(\textbf{nil})\ &&=\ \varnothing\\ 
    &\mathrm{vars}(x)\ &&=\ \{\ x\ \}\\
    &\mathrm{vars}(e_1\ \otimes\ e_2)\ &&=\ \mathrm{vars}(e_1)\ \cup\ \mathrm{vars}(e_2)
\end{alignat*}

\noindent To facilitate support for subtype polymorphism, we also define a binary subtype relation $c_1 \prec: c_2$ for classes:

\begin{enumerate}
    \item $c_1 \prec: c_2$ if $c_1$ inherits from $c_2$
    \item $c \prec: c$ (\textit{reflexivity})
    \item $c_1 \prec: c_3$ if $c_1 \prec: c_2$ and $c_2 \prec: c_3$ (\textit{transitivity})
\end{enumerate}

\subsection{Expressions}

The type judgment:
\begin{equation*}
    \dfrac{}{\Pi\ \vdash_{expr}\ e\ :\ \tau}   
\end{equation*}
defines the type of expressions. We say that under environment $\Pi$, expression $e$ has type $\tau$.

\begin{figure}[ht]
\fontsize{10pt}{12pt}\selectfont
\centering

\begin{equation*}
    \dfrac{}{\Pi\ \vdash_{expr}\ n\ :\ \textbf{int}}\ \textsc{T-Con} \qquad \dfrac{\Pi(x)\ =\ \tau}{\Pi\ \vdash_{expr}\ x\ :\ \tau}\ \textsc{T-Var} \qquad \dfrac{\tau\ \neq\ \textbf{int}}{\Pi\ \vdash_{expr}\ \textbf{nil}\ :\ \tau}\ \textsc{T-Nil}
\end{equation*}

\begin{equation*}
    \dfrac{\Pi\ \vdash_{expr}\ e_1\ :\ \textbf{int} \qquad \Pi\ \vdash_{expr}\ e_2\ :\ \textbf{int}}{\Pi\ \vdash_{expr}\ e_1\ \otimes\ e_2\ :\ \textbf{int}}\ \textsc{T-BinOpInt}
\end{equation*}

\begin{equation*}
    \dfrac{\Pi\ \vdash_{expr}\ e_1\ :\ \tau \qquad \Pi\ \vdash_{expr}\ e_2\ :\ \tau \qquad \ominus \in \{\textbf{\texttt{=}},\ \textbf{\texttt{!=}}\}}{\Pi\ \vdash_{expr}\ e_1\ \ominus\ e_2\ :\ \textbf{int}}\ \textsc{T-BinOpObj}
\end{equation*}

\caption{Typing rules for ROOPL expressions}
\label{fig:expression-typing}
\end{figure}

The type rules \textsc{T-Con}, \textsc{T-Var} and \textsc{T-Nil} defines the types of simple expressions. Numeric literals are always of type \textbf{int}, the type of some variable $x$ depends on its type in the type environment $\Pi$ and the \textbf{nil}-literal can have any non-integer type. All binary operations are defined for integers, while the equality and inequality comparisons are also defined for object references.

\subsection{Statements}

The type judgment:
\begin{equation*}
    \dfrac{}{\langle\Pi,\ c\rangle\ \vdash_{stmt}^\Gamma\ s}
\end{equation*}
\vspace{1pt}
defines the well-typed statements. We say that under type environment $\Pi$ within class $c$, the statement $s$ is well-typed with class map $\Gamma$.

\begin{figure}[!ht]
\fontsize{10pt}{12pt}\selectfont
\centering

\begin{equation*}
    \dfrac{x\ \notin\ \mathrm{vars}(e) \qquad \Pi\ \vdash_{expr}\ e\ :\ \textbf{int} \qquad \Pi(x)\ =\ \textbf{int}}{\langle\Pi,\ c\rangle\ \vdash_{stmt}^\Gamma\ x\ \odot\textbf{\texttt{=}}\ e}\ \textsc{T-AssVar}
\end{equation*}

\begin{equation*}
    \dfrac{\Pi\ \vdash_{expr}\ e_1\ :\ \textbf{int}\ \qquad \langle\Pi,\ c\rangle\ \vdash_{stmt}^\Gamma\ s_1\ \qquad \langle\Pi,\ c\rangle\ \vdash_{stmt}^\Gamma\ s_2\ \qquad \Pi\ \vdash_{expr}\ e_2\ :\ \textbf{int}}{\langle\Pi,\ c\rangle\ \vdash_{stmt}^\Gamma\ \textbf{if}\ e_1\ \textbf{then}\ s_1\ \textbf{else}\ s_2\ \textbf{fi}\ e_2}\ \textsc{T-If}
\end{equation*}

\begin{equation*}
    \dfrac{\Pi\ \vdash_{expr}\ e_1\ :\ \textbf{int}\ \qquad \langle\Pi,\ c\rangle\ \vdash_{stmt}^\Gamma\ s_1\ \qquad \langle\Pi,\ c\rangle\ \vdash_{stmt}^\Gamma\ s_2\ \qquad \Pi\ \vdash_{expr}\ e_2\ :\ \textbf{int}}{\langle\Pi,\ c\rangle\ \vdash_{stmt}^\Gamma\ \textbf{from}\ e_1\ \textbf{do}\ s_1\ \textbf{loop}\ s_2\ \textbf{until}\ e_2}\ \textsc{T-Loop}
\end{equation*}

\begin{equation*}
    \dfrac{\langle\Pi[x\ \mapsto\ c'],\ c\rangle\ \vdash_{stmt}^\Gamma\ s}{\langle\Pi,\ c\rangle\ \vdash_{stmt}^\Gamma\ \textbf{construct}\ c'\ x \quad s \quad \textbf{destruct}\ x}\ \textsc{T-ObjBlock} \qquad \dfrac{}{\langle\Pi,\ c\rangle\ \vdash_{stmt}^\Gamma\ \textbf{skip}}\ \textsc{T-Skip}
\end{equation*}

\begin{equation*}
    \dfrac{\langle\Pi,\ c\rangle\ \vdash_{stmt}^\Gamma\ s_1 \qquad \langle\Pi,\ c\rangle\ \vdash_{stmt}^\Gamma\ s_2}{\langle\Pi,\ c\rangle\ \vdash_{stmt}^\Gamma\ s_1\ s_2}\ \textsc{T-Seq} \qquad \dfrac{\Pi(x_1)\ =\ \Pi(x_2)}{\langle\Pi,\ c\rangle\ \vdash_{stmt}^\Gamma\ x_1\ \textbf{\texttt{<=>}}\ x_2}\ \textsc{T-SwpVar}
\end{equation*}

\begin{equation*}
    \dfrac{
        \begin{gathered}
            \Gamma(c) = \Big(\ fields,\ methods\ \Big) \qquad \Big(\ \textbf{method}\ q\textbf{\texttt{(}}t_1\ y_1,\ \dots,\ t_n\ y_n\textbf{\texttt{)}}\ s\ \Big)\ \in\ methods\\
            \big\{\ x_1,\ \dots,\ x_n\ \big\}\ \cap\ fields\ =\ \emptyset \qquad i\ \neq\ j \implies x_i\ \neq\ x_j \qquad \Pi(x_1) \prec: t_1\ \cdots\ \Pi(x_n) \prec: t_n
        \end{gathered}
    }{
        \langle\Pi,\ c\rangle\ \vdash_{stmt}^\Gamma\ \textbf{call}\ q\textbf{\texttt{(}}x_1,\ \dots,\ x_n\textbf{\texttt{)}}
    }\ \textsc{T-Call}
\end{equation*}

\begin{equation*}
    \dfrac{
        \begin{gathered}
            \Gamma(\Pi(x_0)) = \Big(\ fields,\ methods\ \Big) \qquad \Big(\ \textbf{method}\ q\textbf{\texttt{(}}t_1\ y_1,\ \dots,\ t_n\ y_n\textbf{\texttt{)}}\ s\ \Big)\ \in\ methods\\
            i\ \neq\ j \implies x_i\ \neq\ x_j \qquad \Pi(x_1) \prec: t_1\ \cdots\ \Pi(x_n) \prec: t_n
        \end{gathered}
    }{
        \langle\Pi,\ c\rangle\ \vdash_{stmt}^\Gamma\ \textbf{call}\ x_0\textbf{\texttt{::}}q\textbf{\texttt{(}}x_1,\ \dots,\ x_n\textbf{\texttt{)}}
    }\ \textsc{T-CallO}
\end{equation*}

\begin{equation*}
    \dfrac{\langle\Pi,\ c\rangle\ \vdash_{stmt}^\Gamma\ \textbf{call}\ q\textbf{\texttt{(}}x_1,\ \dots,\ x_n\textbf{\texttt{)}}}{\langle\Pi,\ c\rangle\ \vdash_{stmt}^\Gamma\ \textbf{uncall}\ q\textbf{\texttt{(}}x_1,\ \dots,\ x_n\textbf{\texttt{)}}}\ \textsc{T-UC} \qquad \dfrac{\langle\Pi,\ c\rangle\ \vdash_{stmt}^\Gamma\ \textbf{call}\ x_0\textbf{\texttt{::}}q\textbf{\texttt{(}}x_1,\ \dots,\ x_n\textbf{\texttt{)}}}{\langle\Pi,\ c\rangle\ \vdash_{stmt}^\Gamma\ \textbf{uncall}\ x_0\textbf{\texttt{::}}q\textbf{\texttt{(}}x_1,\ \dots,\ x_n\textbf{\texttt{)}}}\ \textsc{T-UCO}
\end{equation*}

\caption{Typing rules for ROOPL statements}
\label{fig:statement-typing}
\end{figure}

The type rule \textsc{T-AssVar} defines well-typed variable assignments as only those where both sides of the assignment are of type \textbf{int} and the assignee identifier $x$ does not occur in the expression $e$. Rules \textsc{T-If} and \textsc{T-Loop} define the set of well-typed conditionals and loop statements - the entry and exit conditions must be integers, while the branch and loop statements should be well-typed themselves. An object block is well-typed if the block statement is, with the new object $x$ bound in the type environment. The skip statement is always well-typed while a statement sequence is well-typed provided each of its constituent statements are as well. A variable swap statement is well-typed only if both of its operands have the same type.

A local method invocation is well-typed, in accordance with type rule \textsc{T-Call}, only if:

\begin{itemize}
    \renewcommand\labelitemi{\normalfont\bfseries \textendash}
    \item The number of arguments matches the arity of the method
    \item No class fields are passed as arguments to the method (See Section~\ref{sec:argument-aliasing})
    \item There are no duplicate arguments (See Section~\ref{sec:argument-aliasing})
    \item Each argument is a subtype of the type of the equivalent formal parameter
\end{itemize}

The type rule \textsc{T-CallO} establishes similar conditions for foreign method invocations, for which there is no restriction on class fields being used as arguments. There is however, the condition that the callee object $x_0$ is not also passed as an argument. The type rules \textsc{T-UC} and \textsc{T-UCO} describe the conditions for uncalling methods and they are both defined in terms of their inverse counterparts.

\subsection{Programs}

\begin{figure}[ht]
\fontsize{10pt}{12pt}\selectfont
\centering

\begin{equation*}
    \dfrac{\langle\Pi[x_1\ \mapsto\ t_1,\ \dots,\ x_n\ \mapsto\ t_n],\ c\rangle\ \vdash_{stmt}^\Gamma\ s}{\langle\Pi,\ c\rangle\ \vdash_{meth}^\Gamma\ \textbf{method}\ q\textbf{\texttt{(}}t_1\ x_1,\ \dots,\ t_n\ x_n\textbf{\texttt{)}}\ s}\ \textsc{T-Method}
\end{equation*}

\begin{equation*}
    \dfrac{
        \begin{gathered}
            \Gamma(c) = \Big(\ \overbrace{\{\langle t_1,\ f_1\rangle,\ \dots,\ \langle t_i,\ f_i\rangle\}}^{fields},\ \overbrace{\{m_1,\ \dots,\ m_n\}}^{methods}\ \Big)\\
            \Pi = [f_1\ \mapsto\ t_1,\ \dots,\ f_i\ \mapsto t_i] \qquad \langle\Pi,\ c\rangle\ \vdash_{meth}^\Gamma m_1 \quad \cdots \quad \langle\Pi,\ c\rangle\ \vdash_{meth}^\Gamma m_n
        \end{gathered}
    }{
        \vdash_{class}^\Gamma c
    }\ \textsc{T-Class}
\end{equation*}

\begin{equation*}
    \dfrac{
        \begin{gathered}
            \Big(\ \textbf{method main}\textbf{\texttt{()}}\ s\ \Big)\ \in\ \bigcup\limits_{i = 1}^n\ \mathrm{methods}(c_i)\\
            \Gamma = \mathrm{gen}(c_1,\ \dots,\ c_n) \qquad \vdash_{class}^\Gamma c_1 \quad \cdots \quad \vdash_{class}^\Gamma c_n
        \end{gathered}
    }{
        \vdash_{prog} c_1\ \cdots\ c_n
    }\ \textsc{T-Prog}
\end{equation*}

\caption{Typing rules for ROOPL methods, classes and programs}
\label{fig:program-typing}
\end{figure}
The type rules \textsc{T-Prog}, \textsc{T-Class} and \textsc{T-Method} defines the set of well-typed programs, classes and methods respectively.

A class is well-typed iff each of its methods are well-typed with all class fields bound to their respective types in the type environment. A method is well-typed iff its body is well-typed with all parameters bound to their respective types in the type environment. A ROOPL program is well-typed iff all of its classes are well-typed and there exists a nullary method named \textbf{main}. See Figure~\ref{fig:class-map-aux} for the definition of function \textit{methods}.

\section{Language Semantics}
\label{sec:semantics}

The operational semantics of ROOPL are specified by the syntax-directed inference rules shown in the following sections. There are three main judgments: the evaluation of ROOPL expressions, the execution of ROOPL statements and the execution of ROOPL programs. The next section establishes the notation and presents some auxiliary definitions.

\subsection{Preliminaries}

Let $\mathbb{N}_0$ be the set of non-negative integers. A \textit{memory location} $l \in \mathbb{N}_0$ refers to a single location in program memory. An \textit{environment} $\gamma$ is a partial function mapping variable identifiers to memory locations. A \textit{store} $\mu$ is a partial function mapping memory locations to values. An \textit{object} is a tuple consisting of the class name of the object and an environment mapping the object fields to memory locations. A \textit{value} $v$ is either an integer, an object or a memory location.

The application of an environment $\gamma$ to some variable identifier $x$ is denoted by $\gamma(x)$. Update $\gamma' = \gamma[x\ \mapsto\ l]$ defines an environment $\gamma'$ such that $\gamma'(x) = l$ and $\gamma'(y) = \gamma(y)$ if $y \neq x$. The empty environment is written $[\ ]$. The same notation is used for stores.

\begin{figure}[ht]
\centering

\begin{alignat*}{2}
l \in\ &\text{Locs}\ &&=\quad \mathbb{N}_0\\
\gamma \in\ &\text{Envs}\ &&=\quad \text{VarIDs}\ \rightharpoonup\ \text{Locs}\\
\mu \in\ &\text{Stores}\ &&=\quad \text{Locs}\ \rightharpoonup\ \text{Values}\\
&\text{Objects}\ &&=\quad \{\ \langle c_f,\ \gamma_f \rangle\ |\ c_f\ \in\ \text{ClassIDs}\ \wedge\ \gamma_f\ \in\ \text{Envs}\ \}\\
v \in\ &\text{Values}\ &&=\quad \mathbb{Z}\ \cup\ \text{Objects}\ \cup\ \text{Locs}
\end{alignat*}

\caption{Semantic values}
\label{fig:semantic-values}
\end{figure}

\subsection{Expressions}

The judgment:
\begin{equation*}
    \langle\gamma,\ \mu\rangle\ \vdash_{expr}\ e\ \Rightarrow\ v
\end{equation*}
defines the meaning of expressions. We say that under environment $\gamma$ and store $\mu$, expression $e$ evaluates to the value $v$.

\begin{figure}[ht]
\fontsize{10pt}{12pt}\selectfont
\centering

\begin{equation*}
\dfrac{}{\langle\gamma,\ \mu\rangle\ \vdash_{expr}\ n\ \Rightarrow\ \overline{n}}\ \textsc{Con} \qquad \dfrac{}{\langle\gamma,\ \mu\rangle\ \vdash_{expr}\ x\ \Rightarrow\ \mu(\gamma(x))}\ \textsc{Var} \qquad \dfrac{}{\langle\gamma,\ \mu\rangle\ \vdash_{expr}\ \textbf{\texttt{nil}}\ \Rightarrow\ 0}\ \textsc{Nil}
\end{equation*}

\begin{equation*}
    \dfrac{\langle\gamma,\ \mu\rangle\ \vdash_{expr}\ e_1\ \Rightarrow\ v_1 \qquad \langle\gamma,\ \mu\rangle\ \vdash_{expr}\ e_2\ \Rightarrow\ v_2 \qquad \llbracket \otimes \rrbracket(v_1,\ v_2)\ =\ v}{\langle\gamma,\ \mu\rangle\ \vdash_{expr}\ e_1 \otimes e_2\ \Rightarrow\ v}\ \textsc{BinOp}
\end{equation*}

\caption{Semantic inference rules for evaluation of ROOPL expressions}
\label{fig:expression-semantics}
\end{figure}

There are no side effects on the store when evaluating a ROOPL expression. Like in Janus, the logic value \textit{true} is represented by any non-zero integer, while \textit{false} is represented by zero. For the sake of simplicity, \textbf{nil} evaluates to $0$, which can never be the value of a non-nil reference, thereby ensuring that the equality and inequality operators behave as expected.

\begin{figure}[ht]
\fontsize{10pt}{12pt}\selectfont
\centering

\begin{align*}
    \llbracket \textbf{\texttt{+}} \rrbracket(v_1,\ v_2)\ &=\ v_1\ +\ v_2 & \llbracket \textbf{\texttt{\%}} \rrbracket(v_1,\ v_2)\ &=\ v_1\ mod\ v_2\\[0.8ex]
    \llbracket \textbf{\texttt{-}} \rrbracket(v_1,\ v_2)\ &=\ v_1\ -\ v_2 & \llbracket \textbf{\texttt{\&}} \rrbracket(v_1,\ v_2)\ &=\ v_1\ and\ v_2\\[0.8ex]
    \llbracket \textbf{\texttt{*}} \rrbracket(v_1,\ v_2)\ &=\ v_1\ \times\ v_2 & \llbracket \textbf{\texttt{|}} \rrbracket(v_1,\ v_2)\ &=\ v_1\ or\ v_2\\[0.8ex]
    \llbracket \textbf{\texttt{/}} \rrbracket(v_1,\ v_2)\ &=\ \frac{v_1}{v_2} & \llbracket \textbf{\texttt{\^}} \rrbracket(v_1,\ v_2)\ &=\ v_1\ xor\ v_2\\[0.8ex]
    \llbracket \textbf{\texttt{\&\&}} \rrbracket(v_1,\ v_2)\ &=\ \begin{cases}0 &\mbox{if  } v_1\ =\ 0\ \vee\ v_2\ =\ 0\\1 &\mbox{otherwise}\end{cases} & \llbracket \textbf{\texttt{<=}} \rrbracket(v_1,\ v_2)\ &=\ \begin{cases}1 &\mbox{if  } v_1\ \leq\ v_2\\0 &\mbox{otherwise}\end{cases}\\[0.8ex]
    \llbracket \textbf{\texttt{||}} \rrbracket(v_1,\ v_2)\ &=\ \begin{cases}0 &\mbox{if  } v_1\ =\ v_2\ =\ 0\\1 &\mbox{otherwise}\end{cases} & \llbracket \textbf{\texttt{>=}} \rrbracket(v_1,\ v_2)\ &=\ \begin{cases}1 &\mbox{if  } v_1\ \geq\ v_2\\0 &\mbox{otherwise}\end{cases}\\[0.8ex]
    \llbracket \textbf{\texttt{<}} \rrbracket(v_1,\ v_2)\ &=\ \begin{cases}1 &\mbox{if  } v_1\ <\ v_2\\0 &\mbox{otherwise}\end{cases} & \llbracket \textbf{\texttt{=}} \rrbracket(v_1,\ v_2)\ &=\ \begin{cases}1 &\mbox{if  } v_1\ =\ v_2\\0 &\mbox{otherwise}\end{cases}\\[0.8ex]
    \llbracket \textbf{\texttt{>}} \rrbracket(v_1,\ v_2)\ &=\ \begin{cases}1 &\mbox{if  } v_1\ >\ v_2\\0 &\mbox{otherwise}\end{cases} & \llbracket \textbf{\texttt{!=}} \rrbracket(v_1,\ v_2)\ &=\ \begin{cases}1 &\mbox{if  } v_1\ \neq\ v_2\\0 &\mbox{otherwise}\end{cases}
\end{align*}

\caption[Definition of the functions $\llbracket \otimes \rrbracket$]{Definition of the functions $\llbracket \otimes \rrbracket$, where $\otimes$ represents any of the binary expression operators}
\label{fig:expression-semantics-auxiliary}
\end{figure}

The inference rules \textsc{Con}, \textsc{Var} and \textsc{Nil} defines the meaning of expressions containing simple values or variables, while \textsc{BinOp} defines the meaning of expressions containing any of the arithmetic operators $\{\textbf{\texttt{+}},\ \textbf{\texttt{-}},\ \textbf{\texttt{*}},\ \textbf{\texttt{/}},\ \textbf{\texttt{\%}}\}$, bitwise operators $\{\textbf{\texttt{\&}},\ \textbf{\texttt{|}},\ \textbf{\texttt{\textasciicircum}}\}$, logical operators $\{\textbf{\texttt{\&\&}},\ \textbf{\texttt{||}}\}$ or relational operators $\{\textbf{\texttt{<}},\ \textbf{\texttt{>}},\ \textbf{\texttt{=}},\ \textbf{\texttt{!=}},\ \textbf{\texttt{<=}},\ \textbf{\texttt{>=}}\}$, all of which are defined in Figure~\ref{fig:expression-semantics-auxiliary}.

\subsection{Statements}

The judgment:
\begin{equation*}
    \langle l,\ \gamma\rangle\ \vdash_{stmt}^\Gamma\ s\ :\ \mu\ \rightleftharpoons\ \mu'
\end{equation*}
defines the meaning of statements. We say that under environment $\gamma$ and object $l$, statement $s$ with class map $\Gamma$ reversibly transforms store $\mu$ to store $\mu'$. The location $l$ is simply the location in the store $\mu$ of the \textit{current object}. It is equivalent to the value of the \textit{this} or \textit{self} keywords of other OOP languages but cannot be referred to explicitly in ROOPL. Figure~\ref{fig:statement-semantics-a} on page~\pageref{fig:statement-semantics-a} and Figure~\ref{fig:statement-semantics-b} on page~\pageref{fig:statement-semantics-b} shows the operational semantics of ROOPL statements.

\begin{subfigures}
\begin{figure}[!ht]
\fontsize{10pt}{12pt}\selectfont
\centering

\begin{equation*}
    \dfrac{}{\langle l,\ \gamma\rangle\ \vdash_{stmt}^\Gamma\ \textbf{skip}\ :\ \mu\ \rightleftharpoons\ \mu}\ \textsc{Skip}
\end{equation*}

\begin{equation*}
    \dfrac{\langle l,\ \gamma\rangle\ \vdash_{stmt}^\Gamma\ s_1\ :\ \mu\ \rightleftharpoons\ \mu' \qquad \langle l,\ \gamma\rangle\ \vdash_{stmt}^\Gamma\ s_2\ :\ \mu'\ \rightleftharpoons\ \mu''}{\langle l,\ \gamma\rangle\ \vdash_{stmt}^\Gamma\ s_1\ s_2\ :\ \mu\ \rightleftharpoons\ \mu''}\ \textsc{Seq}
\end{equation*}

\begin{equation*}
    \dfrac{\langle\gamma,\ \mu\rangle\ \vdash_{expr}\ e\ \Rightarrow\ v \qquad \llbracket \odot \rrbracket(\mu(\gamma(x)),\ v)\ =\ v'}{\langle l,\ \gamma\rangle\ \vdash_{stmt}^\Gamma\ x \odot \textbf{\texttt{=}}\ e\ :\ \mu\ \rightleftharpoons\ \mu[\gamma(x)\ \mapsto\ v']}\ \textsc{AssVar}
\end{equation*}

\begin{equation*}
    \dfrac{\mu(\gamma(x_1))\ =\ v_1 \qquad \mu(\gamma(x_2))\ =\ v_2}{\langle l,\ \gamma\rangle\ \vdash_{stmt}^\Gamma\ x_1\ \textbf{\texttt{<=>}}\ x_2\ :\ \mu\ \rightleftharpoons\ \mu[\gamma(x_1)\ \mapsto\ v_2,\ \gamma(x_2)\ \mapsto\ v_1]}\ \textsc{SwpVar}
\end{equation*}

\begin{equation*}
    \dfrac{
        \begin{gathered}
            \langle\gamma,\ \mu\rangle\ \vdash_{expr}\ e_1\ \centernot\Rightarrow\ 0 \qquad \langle l,\ \gamma\rangle\ \vdash_{stmt}^\Gamma\ s_1\ :\ \mu\ \rightleftharpoons\ \mu'\\
            \langle l,\ \gamma\rangle\ \vdash_{loop}^\Gamma\ (e_1,\ s_1,\ s_2,\ e_2)\ :\ \mu'\ \rightleftharpoons\ \mu''
        \end{gathered}
    }{
        \langle l,\ \gamma\rangle\ \vdash_{stmt}^\Gamma\ \textbf{from}\ e_1\ \textbf{do}\ s_1\ \textbf{loop}\ s_2\ \textbf{until}\ e_2\ :\ \mu\ \rightleftharpoons\ \mu''
    }\ \textsc{LoopMain}
\end{equation*}

\begin{equation*}
    \dfrac{\langle\gamma,\ \mu\rangle\ \vdash_{expr}\ e_2\ \centernot\Rightarrow\ 0}{\langle l,\ \gamma\rangle\ \vdash_{loop}^\Gamma\ (e_1,\ s_1,\ s_2,\ e_2)\ :\ \mu\ \rightleftharpoons\ \mu}\ \textsc{LoopBase}
\end{equation*}

\begin{equation*}
    \dfrac{
        \begin{gathered}
            \langle\gamma,\ \mu\rangle\ \vdash_{expr}\ e_2\ \Rightarrow\ 0 \qquad \langle l,\ \gamma\rangle\ \vdash_{stmt}^\Gamma\ s_2\ :\ \mu\ \rightleftharpoons\ \mu'\\
            \langle\gamma,\ \mu'\rangle\ \vdash_{expr}\ e_1\ \Rightarrow\ 0 \qquad \langle l,\ \gamma\rangle\ \vdash_{stmt}^\Gamma\ s_1\ :\ \mu'\ \rightleftharpoons\ \mu''\\
            \langle l,\ \gamma\rangle\ \vdash_{loop}^\Gamma\ (e_1,\ s_1,\ s_2,\ e_2)\ :\ \mu''\ \rightleftharpoons\ \mu'''
        \end{gathered}
    }{
        \langle l,\ \gamma\rangle\ \vdash_{loop}^\Gamma\ (e_1,\ s_1,\ s_2,\ e_2)\ :\ \mu\ \rightleftharpoons\ \mu'''
    }\ \textsc{LoopRec}
\end{equation*}

\begin{equation*}
    \dfrac{\langle\gamma,\ \mu\rangle\ \vdash_{expr}\ e_1\ \centernot\Rightarrow\ 0 \qquad \langle l,\ \gamma\rangle\ \vdash_{stmt}^\Gamma\ s_1\ :\ \mu\ \rightleftharpoons\ \mu' \qquad \langle\gamma,\ \mu'\rangle\ \vdash_{expr}\ e_2\ \centernot\Rightarrow\ 0}{\langle l,\ \gamma\rangle\ \vdash_{stmt}^\Gamma\ \textbf{if}\ e_1\ \textbf{then}\ s_1\ \textbf{else}\ s_2\ \textbf{fi}\ e_2\ :\ \mu\ \rightleftharpoons\ \mu'}\ \textsc{IfTrue}
\end{equation*}

\begin{equation*}
    \dfrac{\langle\gamma,\ \mu\rangle\ \vdash_{expr}\ e_1\ \Rightarrow\ 0 \qquad \langle l,\ \gamma\rangle\ \vdash_{stmt}^\Gamma\ s_2\ :\ \mu\ \rightleftharpoons\ \mu' \qquad \langle\gamma,\ \mu'\rangle\ \vdash_{expr}\ e_2\ \Rightarrow\ 0}{\langle l,\ \gamma\rangle\ \vdash_{stmt}^\Gamma\ \textbf{if}\ e_1\ \textbf{then}\ s_1\ \textbf{else}\ s_2\ \textbf{fi}\ e_2\ :\ \mu\ \rightleftharpoons\ \mu'}\ \textsc{IfFalse}
\end{equation*}

\caption{Semantic inference rules for execution of ROOPL statements}
\label{fig:statement-semantics-a}
\end{figure}

\begin{figure}[!ht]
\fontsize{10pt}{12pt}\selectfont
\centering

\begin{equation*}
    \dfrac{
        \begin{gathered}
            \mu(l) = \langle c,\ \gamma'\rangle \qquad \Gamma(c) = \Big(\ fields,\ methods\ \Big)\\
            \Big(\ \textbf{method}\ q\textbf{\texttt{(}}t_1\ y_1,\ \dots,\ t_n\ y_n\textbf{\texttt{)}}\ s\ \Big)\ \in\ methods\\
            \langle l,\ \gamma'[y_1\ \mapsto\ \gamma(x_1),\ \dots,\ y_n\ \mapsto\ \gamma(x_n)]\rangle\ \vdash_{stmt}^\Gamma\ s\ :\ \mu\ \rightleftharpoons\ \mu'
        \end{gathered}
    }{
        \langle l,\ \gamma\rangle\ \vdash_{stmt}^\Gamma\ \textbf{call}\ q\textbf{\texttt{(}}x_1,\ \dots,\ x_n\textbf{\texttt{)}}\ :\ \mu\ \rightleftharpoons\ \mu'
    }\ \textsc{Call}
\end{equation*}

\begin{equation*}
    \dfrac{\langle l,\ \gamma\rangle\ \vdash_{stmt}^\Gamma\ \textbf{call}\ q\textbf{\texttt{(}}x_1,\ \dots,\ x_n\textbf{\texttt{)}}\ :\ \mu'\ \rightleftharpoons\ \mu}{\langle l,\ \gamma\rangle\ \vdash_{stmt}^\Gamma\ \textbf{uncall}\ q\textbf{\texttt{(}}x_1,\ \dots,\ x_n\textbf{\texttt{)}}\ :\ \mu\ \rightleftharpoons\ \mu'}\ \textsc{Uncall}
\end{equation*}

\begin{equation*}
    \dfrac{
        \begin{gathered}
            l' = \mu(\gamma(x_0)) \qquad \mu(l') = \langle c,\ \gamma'\rangle \qquad \Gamma(c) = \Big(\ fields,\ methods\ \Big)\\
            \Big(\ \textbf{method}\ q\textbf{\texttt{(}}t_1\ y_1,\ \dots,\ t_n\ y_n\textbf{\texttt{)}}\ s\ \Big)\ \in\ methods\\
            \langle l',\ \gamma'[y_1\ \mapsto\ \gamma(x_1),\ \dots,\ y_n\ \mapsto\ \gamma(x_n)]\rangle\ \vdash_{stmt}^\Gamma\ s\ :\ \mu\ \rightleftharpoons\ \mu'
        \end{gathered}
    }{
        \langle l,\ \gamma\rangle\ \vdash_{stmt}^\Gamma\ \textbf{call}\ x_0\textbf{\texttt{::}}q\textbf{\texttt{(}}x_1,\ \dots,\ x_n\textbf{\texttt{)}}\ :\ \mu\ \rightleftharpoons\ \mu'
    }\ \textsc{CallOBj}
\end{equation*}

\begin{equation*}
    \dfrac{\langle l,\ \gamma\rangle\ \vdash_{stmt}^\Gamma\ \textbf{call}\ x_0\textbf{\texttt{::}}q\textbf{\texttt{(}}x_1,\ \dots,\ x_n\textbf{\texttt{)}}\ :\ \mu'\ \rightleftharpoons\ \mu}{\langle l,\ \gamma\rangle\ \vdash_{stmt}^\Gamma\ \textbf{uncall}\ x_0\textbf{\texttt{::}}q\textbf{\texttt{(}}x_1,\ \dots,\ x_n\textbf{\texttt{)}}\ :\ \mu\ \rightleftharpoons\ \mu'}\ \textsc{UncallObj}
\end{equation*}

\begin{equation*}
    \dfrac{
        \begin{gathered}
            \Gamma(c) = \Big(\ \overbrace{\{\langle t_1,\ f_1\rangle,\ \dots,\ \langle t_n,\ f_n\rangle\}}^{fields},\ methods\ \Big) \quad \gamma' = [f_1\ \mapsto\ a_1,\ \dots,\ f_n\ \mapsto\ a_n]\\
            \big\{\ l',\ r,\ a_1,\ \dots,\ a_n\ \big\}\ \cap\ \mathrm{dom}(\mu)\ =\ \emptyset \qquad \big|\big\{\ l',\ r,\ a_1,\ \dots,\ a_n\ \big\}\big|\ =\ n\ +\ 2\\
            \mu' = \mu\Bigg[\begin{gathered}a_1\ \mapsto\ 0,\ \dots,\ a_n\ \mapsto\ 0\\l'\ \mapsto\ \langle c,\ \gamma'\rangle,\ r\ \mapsto\ l'\end{gathered}\Bigg] \quad \begin{gathered}\langle l,\ \gamma[x\ \mapsto\ r]\rangle\ \vdash_{stmt}^\Gamma\ s\ :\ \mu'\ \rightleftharpoons\ \mu''\\\mu''(a_1) = 0 \quad \cdots \quad \mu''(a_n) = 0\end{gathered}\\[0.8ex]
        \end{gathered}
    }{
        \langle l,\ \gamma\rangle\ \vdash_{stmt}^\Gamma\ \textbf{construct}\ c\ x \quad s \quad \textbf{destruct}\ x\ :\ \mu\ \rightleftharpoons\ \mu''{\restriction_{\mathrm{dom}(\mu)}}
    }\ \textsc{ObjBlock}
\end{equation*}

\caption{Semantic inference rules for execution of ROOPL statements (\textit{cont.})}
\label{fig:statement-semantics-b}
\end{figure}
\end{subfigures}

Rule \textsc{Skip} defines the meaning of the skip statement which has no effect on the store $\mu$. Rule \textsc{Seq} defines the meaning of statement sequences and rule \textsc{AssVar} defines reversible assignments.

The rules \textsc{LoopMain}, \textsc{LoopBase} and \textsc{LoopRec} defines the meaning of loops. If assertion $e_1$ holds, the loop is entered by rule \textsc{LoopMain}. Then the loop iterates by rule \textsc{LoopRec} until $e_2$ does not hold, terminating the loop by rule \textsc{LoopBase}. Since conditionals and loops in ROOPL are comparable to those in Janus, these rules are similar to those presented in~\cite{ty:ejanus}.

The semantics of conditional statements are given by rules \textsc{IfTrue} and \textsc{IfFalse}. If the entry condition evaluates to \textit{true} (non-zero), then the \textbf{then}-branch is executed and the exit assertion should also evaluate to \textit{true}. If the entry condition evaluates to \textit{false}, the \textbf{else}-branch is executed and the exit assertion should evaluate to \textit{false}.

Rule \textsc{Call} defines the meaning of invoking a method local to the current object. The method $q$ in the current class $c$ should have exactly $n$ formal parameters $y_1,\ \dots,\ y_n$, matching the $n$ arguments $x_1,\ \dots,\ x_n$. The resulting store $\mu'$ is the store obtained from executing the method body $s$ in the object environment $\gamma'$ with the arguments bound to the formal parameters.

Rule \textsc{Uncall} essentially reverses the direction of execution by requiring the input store of a \textbf{call} statement to serve as the output store of the inverse \textbf{uncall} statement. A similar technique was used in~\cite{ty:janus, ty:ejanus}.

Rule \textsc{CallObj} governs invocation of methods not local to the current object. The resulting store $\mu'$ is the store obtained from executing the method body $s$ in the environment $\gamma'$ of the object $x_0$, with the arguments bound to the formal parameters. The inverse rule \textsc{UncallObj} is defined using the same approach used for rule \textsc{Uncall}.

Even if $x_0$ has been upcast to a base class (as allowed by the type system, see Section~\ref{sec:type-system}) earlier in the program, the class name $c$ refers to the \textit{dynamic} type of $x_0$. As a result, the method lookup will correctly yield the appropriate method from the derived class - in accordance with the concept of subtype-polymorphism (the actual mechanism used to achieve dynamic dispatch, virtual lookup tables, are considered an implementation detail at this point). Method dispatch in ROOPL depends only on the name of the method and the type of the callee object, not on the number of arguments nor their individual types (\textit{single dispatch}).

Rule \textsc{ObjBlock} defines the meaning of a \textbf{construct}/\textbf{destruct} block and the semantics of object construction and destruction. The \textbf{construct}/\textbf{destruct} blocks of ROOPL are similar to the \textbf{local}/\textbf{delocal} blocks of Janus. In both cases, it is the program itself that is responsible for reversibly returning the memory to a state where it can be reclaimed by the system and in the presence of recursion, there is no upper bound on the size the store can grow to. Like in Janus, if $x$ is already in scope when a block scope is entered, that variable is shadowed by the new object $x$ within the statement block (\textit{static lexical scoping}).

The new memory locations $l'$, $r$ and $a_1,\ \dots,\ a_n$ should be unused in the store $\mu$ and they should all represent distinct memory locations. The identifiers $f_1,\ \dots,\ f_n$ representing the fields of the new object are bound to the unused memory locations $a_1,\ \dots,\ a_n$ in the new object environment $\gamma'$. Next, we let $\mu'$ be the updated store containing:

\begin{itemize}
    \renewcommand\labelitemi{\normalfont\bfseries \textendash}
    \item The location $l$ mapped to the object tuple $\langle c,\ \gamma'\rangle$
    \item The object reference $r$ mapped to the location $l$
    \item The $n$ new object fields mapped to $0$
\end{itemize}

The result store $\mu''$ (restricted to the domain of $\mu$) is the store obtained from executing the block statement $s$ in store $\mu'$ under environment $\gamma$ mapping $x$ to the object reference $r$, provided all object fields are zero-cleared in $\mu''$ afterwards (otherwise the statement is undefined).

\subsection{Programs}

The judgment:
\begin{equation*}
    \vdash_{prog}\ p\ \Rightarrow\ \sigma
\end{equation*}
defines the meaning of ROOPL programs. Whichever class in $p$ contains the main method is instantiated and the main method body is executed. The result is a partial function $\sigma$ mapping identifiers to values, corresponding to the class fields of the main class.

\begin{figure}[ht]
\centering

\begin{equation*}
\dfrac{
    \begin{gathered}
        \Gamma = \mathrm{gen}(c_1,\ \dots,\ c_n) \qquad \Gamma(c) = \Big(\ \overbrace{\{\langle t_1,\ f_1\rangle,\ \dots,\ \langle t_i,\ f_i\rangle\}}^{fields},\ methods\ \Big)\\
        \Big(\ \textbf{method main}\textbf{\texttt{()}}\ s\ \Big)\ \in\ methods \qquad \gamma = [f_1\ \mapsto\ 1,\ \dots,\ f_i\ \mapsto\ i]\\
        \mu = [1\ \mapsto\ 0,\ \dots,\ i\ \mapsto\ 0,\ i + 1\ \mapsto\ \langle c,\ \gamma\rangle] \qquad \langle i + 1,\ \gamma\rangle\ \vdash_{stmt}^\Gamma\ s\ :\ \mu\ \rightleftharpoons\ \mu'
    \end{gathered}
}{
    \vdash_{prog}\ c_1\ \cdots\ c_n\ \Rightarrow\ (\mu'\ \circ\ \gamma)
}\ \textsc{Main}
\end{equation*}

\caption{Semantic inference rule for execution of ROOPL programs}
\label{fig:program-semantics}
\end{figure}

Rule \textsc{Main} defines the meaning of a ROOPL program. The fields $f_1,\ \dots,\ f_i$ of the class $c$ containing the main method are bound in a new environment $\gamma$ to the first $i$ memory addresses (excluding address $0$ which is reserved for \textbf{nil}). The first $i$ memory addresses are then initialized to $0$ in a new environment $\mu$ as well as the address $i + 1$ which maps to the new instance of the main object. The modified store $\mu'$ is obtained from executing the body $s$ of the main method. The composite function $(\mu'\ \circ\ \gamma)$, which maps each class field to its final value, serves as the output of executing $p$.

\section{Program Inversion}
\label{sec:inversion}

A common formulation of the Church-Turing thesis states that a function $f$ is computable iff there exists some Turing Machine that computes it. By extension, if some program $p$, written in a Turing-equivalent programming language\footnote{or indeed any algorithm specified in a Turing-equivalent model of computation}, computes a function $f$ then $f$ is computable.

Program inversion is the process of determining an inverse program of $p$, computing the function $f^{-1}$. Given a computable function $f\ :\ X\ \rightarrow\ Y$, we wish to find a program computing the function $f'\ :\ Y\ \rightarrow\ X$ such that:
\begin{equation*}
    f(x) = y\quad \Leftrightarrow\quad f'(y) = x
\end{equation*}
Since $f$ is computable, we can compute $f'(y)$ by simulating $f$ on all inputs $x\ \in\ X$ until the result is $y$. This is a variation of McCarthy's generate-and-test technique~\cite{jm:gat}, which implies that we can always find the inverse program if $f$ is computable. Unfortunately, this is a completely impractical approach to program inversion. McCarthy himself described his approach in the following terms:
\begin{displayquote}[jm:gat][.]
\textelp{} this procedure is extremely inefficient. It corresponds to looking for a proof of a conjecture by checking in some order all possible English essays
\end{displayquote}
Recently, more practical methods for automatic program inversion of irreversible programs have superseded the generate-and-test algorithm~\cite{rg:global}. In the context of reversible programming languages, program inversion is both simple and efficient. Reversible languages like Janus and ROOPL support \textit{local inversion} of program statements - no contextual information or whole-program analysis is needed~\cite{rg:local}. This is a property of reversible languages that follows from the nature of their design and the constraints they impose on the programmer. The statement inverter $\mathcal{I}$ in Figure~\ref{fig:roopl-stmt-inverter} maps ROOPL statements to their inverse counterparts.

\begin{figure}[ht]
\fontsize{10pt}{12pt}\selectfont
\centering

\begin{alignat*}{2}
    &\mathcal{I}\ \llbracket \textbf{skip} \rrbracket\ =\ \textbf{skip} \qquad &&\mathcal{I}\ \llbracket s_1\ s_2 \rrbracket\ =\ \mathcal{I}\llbracket s_2 \rrbracket \enspace \mathcal{I}\llbracket s_1 \rrbracket\\[1.2ex]
    &\mathcal{I}\ \llbracket x\ \textbf{\texttt{+=}}\ e \rrbracket\ =\ x\ \textbf{\texttt{-=}}\ e \qquad &&\mathcal{I}\ \llbracket x\ \textbf{\texttt{-=}}\ e \rrbracket\ =\ x\ \textbf{\texttt{+=}}\ e\\[1.2ex]
    &\mathcal{I}\ \llbracket x\ \textbf{\texttt{\textasciicircum=}}\ e \rrbracket\ =\ x\ \textbf{\texttt{\textasciicircum=}}\ e \qquad &&\mathcal{I}\ \llbracket x_1\ \textbf{\texttt{<=>}}\ x_2 \rrbracket\ =\ x_1\ \textbf{\texttt{<=>}}\ x_2\\[1.2ex]
    &\mathcal{I}\ \llbracket \textbf{call}\ q\textbf{\texttt{(}}\dots\textbf{\texttt{)}} \rrbracket\ =\ \textbf{uncall}\ q\textbf{\texttt{(}}\dots\textbf{\texttt{)}} \qquad &&\mathcal{I}\ \llbracket \textbf{call}\ x\textbf{\texttt{::}}q\textbf{\texttt{(}}\dots\textbf{\texttt{)}} \rrbracket\ =\ \textbf{uncall}\ x\textbf{\texttt{::}}q\textbf{\texttt{(}}\dots\textbf{\texttt{)}}\\[1.2ex]
    &\mathcal{I}\ \llbracket \textbf{uncall}\ q\textbf{\texttt{(}}\dots\textbf{\texttt{)}} \rrbracket\ =\ \textbf{call}\ q\textbf{\texttt{(}}\dots\textbf{\texttt{)}} \qquad &&\mathcal{I}\ \llbracket \textbf{uncall}\ x\textbf{\texttt{::}}q\textbf{\texttt{(}}\dots\textbf{\texttt{)}} \rrbracket\ =\ \textbf{call}\ x\textbf{\texttt{::}}q\textbf{\texttt{(}}\dots\textbf{\texttt{)}}\\[1.2ex]
    &\mathcal{I}\ \llbracket \textbf{if}\ e_1\ \textbf{then}\ s_1\ \textbf{else}\ s_2\ \textbf{fi}\ e_2 \rrbracket\ &&=\ \textbf{if}\ e_1\ \textbf{then}\ \mathcal{I}\llbracket s_1 \rrbracket\ \textbf{else}\ \mathcal{I}\llbracket s_2 \rrbracket\ \textbf{fi}\ e_2\\[1.2ex]
    &\mathcal{I}\ \llbracket \textbf{from}\ e_1\ \textbf{do}\ s_1\ \textbf{loop}\ s_2\ \textbf{until}\ e_2 \rrbracket\ &&=\ \textbf{from}\ e_1\ \textbf{do}\ \mathcal{I}\llbracket s_1 \rrbracket\ \textbf{loop}\ \mathcal{I}\llbracket s_2 \rrbracket\ \textbf{until}\ e_2\\[1.2ex]
    &\mathcal{I}\ \left\llbracket\textbf{construct}\ c\ x\quad s\quad \textbf{destruct}\ x\right\rrbracket\ &&=\ \textbf{construct}\ c\ x\quad\mathcal{I}\llbracket s \rrbracket\quad\textbf{destruct}\ x
\end{alignat*}

\caption{Statement inverter for ROOPL statements}
\label{fig:roopl-stmt-inverter}
\end{figure}

In ROOPL, statement inversion does not change the size of statements and as a consequence, a ROOPL program is exactly the same size as its own inverse. Furthermore, provided that every statement has the same computational complexity as its inverse, it follows that ROOPL programs have the same computational complexity as their inverted counterparts.

\begin{figure}[ht]
\fontsize{10pt}{12pt}\selectfont
\centering

\begin{equation*}
    \mathcal{I}_c\ \left\llbracket
        \begin{aligned}
            \textbf{class}\ &c\ \cdots\\
            &\textbf{method}\ q_1\ \textbf{\texttt{(}}\dots\textbf{\texttt{)}}\ s_1\\
            &\qquad\vdots\\
            &\textbf{method}\ q_n\ \textbf{\texttt{(}}\dots\textbf{\texttt{)}}\ s_n\\
        \end{aligned}
    \right\rrbracket\ =
        \begin{aligned}
            \textbf{class}\ &c\ \cdots\\
            &\textbf{method}\ q_1\ \textbf{\texttt{(}}\dots\textbf{\texttt{)}}\ \mathcal{I}'\llbracket s_1 \rrbracket\\
            &\qquad\vdots\\
            &\textbf{method}\ q_n\ \textbf{\texttt{(}}\dots\textbf{\texttt{)}}\ \mathcal{I}'\llbracket s_n \rrbracket\\
        \end{aligned}
\end{equation*}

\begin{equation*}
    \mathcal{I}_{prog}\ \llbracket cl_1\ \cdots\ cl_n \rrbracket\ =\ \mathcal{I}_c\llbracket cl_1 \rrbracket\ \cdots\ \mathcal{I}_c\llbracket cl_n \rrbracket
\end{equation*}

\caption{Program and class inverters for ROOPL}
\label{fig:roopl-prog-inverter}
\end{figure}

Whole-program inversion is accomplished by straightforward recursive descent over the components and statements of the program. Figure~\ref{fig:roopl-prog-inverter} shows the definition of the ROOPL program inverter $\mathcal{I}_{prog}$, which inverts each method in each class to produce the inverse program. The program inverter $\mathcal{I}_{prog}$ is an \textit{involution}, so inverting a program twice will yield the original program.

\begin{figure}[ht]
\fontsize{10pt}{12pt}\selectfont
\centering

\begin{alignat*}{2}
    &\mathcal{I}'\ \llbracket \textbf{call}\ q\textbf{\texttt{(}}\dots\textbf{\texttt{)}} \rrbracket\ =\ \textbf{call}\ q\textbf{\texttt{(}}\dots\textbf{\texttt{)}} \qquad &&\mathcal{I}'\ \llbracket \textbf{call}\ x\textbf{\texttt{::}}q\textbf{\texttt{(}}\dots\textbf{\texttt{)}} \rrbracket\ =\ \textbf{call}\ x\textbf{\texttt{::}}q\textbf{\texttt{(}}\dots\textbf{\texttt{)}}\\[1.2ex]
    &\mathcal{I}'\ \llbracket \textbf{uncall}\ q\textbf{\texttt{(}}\dots\textbf{\texttt{)}} \rrbracket\ =\ \textbf{uncall}\ q\textbf{\texttt{(}}\dots\textbf{\texttt{)}} \qquad &&\mathcal{I}'\ \llbracket \textbf{uncall}\ x\textbf{\texttt{::}}q\textbf{\texttt{(}}\dots\textbf{\texttt{)}} \rrbracket\ =\ \textbf{uncall}\ x\textbf{\texttt{::}}q\textbf{\texttt{(}}\dots\textbf{\texttt{)}}\\[1.2ex]
    &\omit\hfil $\mathcal{I}'\ \llbracket s \rrbracket\ $ &&\omit $=\ \mathcal{I}\llbracket s \rrbracket$ \hfil
\end{alignat*}

\caption{Modified statement inverter for ROOPL statements}
\label{fig:roopl-modified-stmt-inverter}
\end{figure}

Because calling a method is equivalent to uncalling the same method inverted, if we change call-statements into uncall-statements and vice-versa, the inversion of the method body is cancelled out.

To fix this issue, we use a modified version of the statement inverter for the whole-program inversion, that does not invert calls and uncalls. Figure~\ref{fig:roopl-modified-stmt-inverter} shows the modified statement inverter $\mathcal{I}'$.

\subsection{Invertibility of Statements}

Theorem~\ref{thm:statement-invertibility} shows that $\mathcal{I}$ is in fact a statement inverter. If executing statement $s$ in store $\mu$ yields $\mu'$, then executing statement $\mathcal{I}\llbracket s\rrbracket$ in store $\mu'$ should yield $\mu$.

\begin{theorem}\label{thm:statement-invertibility}(Invertibility of statements)
\begin{equation*}
    \overbrace{\langle l,\ \gamma\rangle\ \vdash_{stmt}^\Gamma\ s\ :\ \mu\ \rightleftharpoons\ \mu'}^{\mathcal{S}}\ \iff\ \overbrace{\langle l,\ \gamma\rangle\ \vdash_{stmt}^\Gamma\ \mathcal{I}\llbracket s\rrbracket\ :\ \mu'\ \rightleftharpoons\ \mu}^{\mathcal{S}'}
\end{equation*}
\end{theorem}

\noindent \textit{Proof.} The proof is by structural induction on the semantic derivation of $\mathcal{S}$ but is omitted. It suffices to show that $\mathcal{S}$ implies $\mathcal{S}'$ - since this can also serve as proof that $\mathcal{S}'$ implies $\mathcal{S}$ because $\mathcal{I}$ is an involution.

\subsection{Type-Safe Statement Inversion}

When given a well-typed statement, the statement inverter $\mathcal{I}$ should always produce a well-typed (inverse) statement. This is an important property of the language as it prevents situations where some method can be \textit{called} successfully, but \textit{uncalling} the same method produces an error or undefined behaviour. The following theorem expresses this property:

\begin{theorem}\label{thm:type-inversion}(Inversion of well-typed statements)
\begin{equation*}
    \overbrace{\langle \Pi,\ c\rangle\ \vdash_{stmt}^\Gamma\ s}^{\mathcal{T}}\ \implies\ \overbrace{\langle \Pi,\ c\rangle\ \vdash_{stmt}^\Gamma\ \mathcal{I}\llbracket s\rrbracket}^{\mathcal{T}'}
\end{equation*}
\end{theorem}

\begin{proof}
By structural induction on $\mathcal{T}$:

\vspace{5mm}

\noindent \textbf{Case} {\fontsize{10pt}{12pt}\selectfont $\mathcal{T}\ =\ \dfrac{\overbrace{x\ \notin\ \mathrm{vars}(e)}^{\mathcal{C}_1} \qquad \overbrace{\Pi\ \vdash_{expr}\ e\ :\ \textbf{int}}^{\mathcal{E}} \qquad \overbrace{\Pi(x)\ =\ \textbf{int}}^{\mathcal{C}_2}}{\langle\Pi,\ c\rangle\ \vdash_{stmt}^\Gamma\ x\ \odot\textbf{\texttt{=}}\ e}\ \textsc{AssVar}$}

\vspace{3mm}

\begin{case}
In this case, $\mathcal{I}\llbracket x\ \odot\textbf{\texttt{=}}\ e \rrbracket\ =\ x\ \odot'\textbf{\texttt{=}}\ e$ for some $\odot'$, so $\mathcal{T}'$ will also be a derivation of rule $\textsc{AssVar}$. Therefore we can just reuse the expression derivation $\mathcal{E}$ and the conditions $\mathcal{C}_1$ and $\mathcal{C}_2$ to construct $\mathcal{T}'$:

\vspace{3mm}

{\fontsize{10pt}{12pt}\selectfont $\mathcal{T}'\ =\ \dfrac{\overbrace{x\ \notin\ \mathrm{vars}(e)}^{\mathcal{C}_1} \qquad \overbrace{\Pi\ \vdash_{expr}\ e\ :\ \textbf{int}}^{\mathcal{E}} \qquad \overbrace{\Pi(x)\ =\ \textbf{int}}^{\mathcal{C}_2}}{\langle\Pi,\ c\rangle\ \vdash_{stmt}^\Gamma\ x\ \odot'\textbf{\texttt{=}}\ e}$}
\end{case}

\vspace{5mm}

\noindent \textbf{Case} {\fontsize{10pt}{12pt}\selectfont $\mathcal{T}\ =\ \dfrac{\Pi(x_1)\ =\ \Pi(x_2)}{\langle\Pi,\ c\rangle\ \vdash_{stmt}^\Gamma\ x_1\ \textbf{\texttt{<=>}}\ x_2}\ \textsc{T-SwpVar}$}

\vspace{3mm}

\begin{case}
Since $\mathcal{I}\llbracket x_1\ \textbf{\texttt{<=>}}\ x_2 \rrbracket\ =\ x_1\ \textbf{\texttt{<=>}}\ x_2$, we can just use the derivation of $\mathcal{T}$ for $\mathcal{T}'$:

\vspace{3mm}

{\fontsize{10pt}{12pt}\selectfont $\mathcal{T}'\ =\ \dfrac{\Pi(x_1)\ =\ \Pi(x_2)}{\langle\Pi,\ c\rangle\ \vdash_{stmt}^\Gamma\ x_1\ \textbf{\texttt{<=>}}\ x_2}$}
\end{case}

\vspace{5mm}

\noindent \textbf{Case} {\fontsize{10pt}{12pt}\selectfont $\mathcal{T}\ =\ \dfrac{\overbrace{\Pi\ \vdash_{expr}\ e_1\ :\ \textbf{int}}^{\mathcal{E}_1}\ \quad \overbrace{\langle\Pi,\ c\rangle\ \vdash_{stmt}^\Gamma\ s_1}^{\mathcal{S}_1}\ \quad \overbrace{\langle\Pi,\ c\rangle\ \vdash_{stmt}^\Gamma\ s_2}^{\mathcal{S}_2}\ \quad \overbrace{\Pi\ \vdash_{expr}\ e_2\ :\ \textbf{int}}^{\mathcal{E}_2}}{\langle\Pi,\ c\rangle\ \vdash_{stmt}^\Gamma\ \textbf{if}\ e_1\ \textbf{then}\ s_1\ \textbf{else}\ s_2\ \textbf{fi}\ e_2}\ \textsc{T-If}$}

\vspace{3mm}

\begin{case}
We have: \quad $\mathcal{I}\ \llbracket \textbf{if}\ e_1\ \textbf{then}\ s_1\ \textbf{else}\ s_2\ \textbf{fi}\ e_2 \rrbracket\ =\ \textbf{if}\ e_1\ \textbf{then}\ \mathcal{I}\llbracket s_1 \rrbracket\ \textbf{else}\ \mathcal{I}\llbracket s_2 \rrbracket\ \textbf{fi}\ e_2$

\vspace{3mm}

By the induction hypothesis on $\mathcal{S}_1$ we get: \quad $\mathcal{S}'_1\ =\ \langle\Pi,\ c\rangle\ \vdash_{stmt}^\Gamma\ \mathcal{I}\llbracket s_1\rrbracket$

\vspace{3mm}

By the induction hypothesis on $\mathcal{S}_2$ we get: \quad $\mathcal{S}'_2\ =\ \langle\Pi,\ c\rangle\ \vdash_{stmt}^\Gamma\ \mathcal{I}\llbracket s_2\rrbracket$

\vspace{3mm}

Using $\mathcal{E}_1$, $\mathcal{S}'_1$, $\mathcal{S}'_2$ and $\mathcal{E}_2$ we can construct $\mathcal{T}'$:

\vspace{3mm}

{\fontsize{10pt}{12pt}\selectfont $\mathcal{T}'\ =\ \dfrac{\overbrace{\Pi\ \vdash_{expr}\ e_1\ :\ \textbf{int}}^{\mathcal{E}_1}\ \quad \overbrace{\langle\Pi,\ c\rangle\ \vdash_{stmt}^\Gamma\ \mathcal{I}\llbracket s_1 \rrbracket}^{\mathcal{S}'_1}\ \quad \overbrace{\langle\Pi,\ c\rangle\ \vdash_{stmt}^\Gamma\ \mathcal{I}\llbracket s_2 \rrbracket}^{\mathcal{S}'_2}\ \quad \overbrace{\Pi\ \vdash_{expr}\ e_2\ :\ \textbf{int}}^{\mathcal{E}_2}}{\langle\Pi,\ c\rangle\ \vdash_{stmt}^\Gamma\ \textbf{if}\ e_1\ \textbf{then}\ \mathcal{I}\llbracket s_1 \rrbracket\ \textbf{else}\ \mathcal{I}\llbracket s_2 \rrbracket\ \textbf{fi}\ e_2}$}
\end{case}

\vspace{5mm}

\noindent \textbf{Case} {\fontsize{10pt}{12pt}\selectfont $\mathcal{T}\ =\ \dfrac{\overbrace{\Pi\ \vdash_{expr}\ e_1\ :\ \textbf{int}}^{\mathcal{E}_1}\ \quad \overbrace{\langle\Pi,\ c\rangle\ \vdash_{stmt}^\Gamma\ s_1}^{\mathcal{S}_1}\ \quad \overbrace{\langle\Pi,\ c\rangle\ \vdash_{stmt}^\Gamma\ s_2}^{\mathcal{S}_2}\ \quad \overbrace{\Pi\ \vdash_{expr}\ e_2\ :\ \textbf{int}}^{\mathcal{E}_2}}{\langle\Pi,\ c\rangle\ \vdash_{stmt}^\Gamma\ \textbf{from}\ e_1\ \textbf{do}\ s_1\ \textbf{loop}\ s_2\ \textbf{until}\ e_2}\ \textsc{T-Loop}$}

\vspace{3mm}

\begin{case}
We have: \quad $\mathcal{I}\ \llbracket \textbf{from}\ e_1\ \textbf{do}\ s_1\ \textbf{loop}\ s_2\ \textbf{until}\ e_2 \rrbracket\ =\ \textbf{from}\ e_1\ \textbf{do}\ \mathcal{I}\llbracket s_1 \rrbracket\ \textbf{loop}\ \mathcal{I}\llbracket s_2 \rrbracket\ \textbf{until}\ e_2$

\vspace{3mm}

By the induction hypothesis on $\mathcal{S}_1$ we get: \quad $\mathcal{S}'_1\ =\ \langle\Pi,\ c\rangle\ \vdash_{stmt}^\Gamma\ \mathcal{I}\llbracket s_1\rrbracket$

\vspace{3mm}

By the induction hypothesis on $\mathcal{S}_2$ we get: \quad $\mathcal{S}'_2\ =\ \langle\Pi,\ c\rangle\ \vdash_{stmt}^\Gamma\ \mathcal{I}\llbracket s_2\rrbracket$

\vspace{3mm}

Using $\mathcal{E}_1$, $\mathcal{S}'_1$, $\mathcal{S}'_2$ and $\mathcal{E}_2$ we can construct $\mathcal{T}'$:

\vspace{3mm}

{\fontsize{10pt}{12pt}\selectfont $\mathcal{T}'\ =\ \dfrac{\overbrace{\Pi\ \vdash_{expr}\ e_1\ :\ \textbf{int}}^{\mathcal{E}_1}\ \quad \overbrace{\langle\Pi,\ c\rangle\ \vdash_{stmt}^\Gamma\ \mathcal{I}\llbracket s_1 \rrbracket}^{\mathcal{S}'_1}\ \quad \overbrace{\langle\Pi,\ c\rangle\ \vdash_{stmt}^\Gamma\ \mathcal{I}\llbracket s_2 \rrbracket}^{\mathcal{S}'_2}\ \quad \overbrace{\Pi\ \vdash_{expr}\ e_2\ :\ \textbf{int}}^{\mathcal{E}_2}}{\langle\Pi,\ c\rangle\ \vdash_{stmt}^\Gamma\ \textbf{from}\ e_1\ \textbf{do}\ \mathcal{I}\llbracket s_1 \rrbracket\ \textbf{loop}\ \mathcal{I}\llbracket s_2 \rrbracket\ \textbf{until}\ e_2}$}
\end{case}

\vspace{5mm}

\noindent \textbf{Case} {\fontsize{10pt}{12pt}\selectfont $\mathcal{T}\ =\ \dfrac{\overbrace{\langle\Pi,\ c\rangle\ \vdash_{stmt}^\Gamma\ s_1}^{\mathcal{S}_1} \qquad \overbrace{\langle\Pi,\ c\rangle\ \vdash_{stmt}^\Gamma\ s_2}^{\mathcal{S}_2}}{\langle\Pi,\ c\rangle\ \vdash_{stmt}^\Gamma\ s_1\ s_2}\ \textsc{T-Seq}$}

\vspace{3mm}

\begin{case}
We have: \quad $\mathcal{I}\ \llbracket s_1\ s_2 \rrbracket\ =\ \mathcal{I}\llbracket s_2 \rrbracket \enspace \mathcal{I}\llbracket s_1 \rrbracket$

\vspace{3mm}

By the induction hypothesis on $\mathcal{S}_1$ we get: \quad $\mathcal{S}'_1\ =\ \langle\Pi,\ c\rangle\ \vdash_{stmt}^\Gamma\ \mathcal{I}\llbracket s_1\rrbracket$

\vspace{3mm}

By the induction hypothesis on $\mathcal{S}_2$ we get: \quad $\mathcal{S}'_2\ =\ \langle\Pi,\ c\rangle\ \vdash_{stmt}^\Gamma\ \mathcal{I}\llbracket s_2\rrbracket$

\vspace{3mm}

Using $\mathcal{S}'_1$ and $\mathcal{S}'_2$ we can construct $\mathcal{T}'$:

\vspace{3mm}

{\fontsize{10pt}{12pt}\selectfont $\mathcal{T}'\ =\ \dfrac{\overbrace{\langle\Pi,\ c\rangle\ \vdash_{stmt}^\Gamma\ \mathcal{I}\llbracket s_2\rrbracket}^{\mathcal{S}'_2} \qquad \overbrace{\langle\Pi,\ c\rangle\ \vdash_{stmt}^\Gamma\ \mathcal{I}\llbracket s_1\rrbracket}^{\mathcal{S}'_1}}{\langle\Pi,\ c\rangle\ \vdash_{stmt}^\Gamma\ \mathcal{I}\llbracket s_2 \rrbracket \enspace \mathcal{I}\llbracket s_1 \rrbracket}$}
\end{case}

\vspace{5mm}

\noindent \textbf{Case} {\fontsize{10pt}{12pt}\selectfont $\mathcal{T}\ =\ \dfrac{}{\langle\Pi,\ c\rangle\ \vdash_{stmt}^\Gamma\ \textbf{skip}}\ \textsc{T-Skip}$}

\vspace{3mm}

\begin{case}
Since $\mathcal{I}\ \llbracket \textbf{skip} \rrbracket\ =\ \textbf{skip}$, and $\textsc{T-Skip}$ is axiomatic, we can choose $\mathcal{T}$ as:

\vspace{3mm}

{\fontsize{10pt}{12pt}\selectfont $\mathcal{T}'\ =\ \dfrac{}{\langle\Pi,\ c\rangle\ \vdash_{stmt}^\Gamma\ \textbf{skip}}$}
\end{case}

\vspace{5mm}

\noindent \textbf{Case} {\fontsize{10pt}{12pt}\selectfont $\mathcal{T}\ =\ \dfrac{\overbrace{\langle\Pi[x\ \mapsto\ c'],\ c\rangle\ \vdash_{stmt}^\Gamma\ s}^{\mathcal{S}}}{\langle\Pi,\ c\rangle\ \vdash_{stmt}^\Gamma\ \textbf{construct}\ c'\ x \quad s \quad \textbf{destruct}\ x}\ \textsc{T-ObjBlock}$}

\vspace{3mm}

\begin{case}
We have: \quad $\mathcal{I}\ \left\llbracket\textbf{construct}\ c\ x\quad s\quad \textbf{destruct}\ x\right\rrbracket\ =\ \textbf{construct}\ c\ x\quad\mathcal{I}\llbracket s \rrbracket\quad\textbf{destruct}\ x$

\vspace{3mm}

By the induction hypothesis on $\mathcal{S}$ we get: \quad $\mathcal{S}'\ =\ \langle\Pi[x\ \mapsto\ c'],\ c\rangle\ \vdash_{stmt}^\Gamma\ \mathcal{I}\llbracket s\rrbracket$

\vspace{3mm}

Which we can use to construct $\mathcal{T}'$:

\vspace{3mm}

{\fontsize{10pt}{12pt}\selectfont $\mathcal{T}'\ =\ \dfrac{\overbrace{\langle\Pi[x\ \mapsto\ c'],\ c\rangle\ \vdash_{stmt}^\Gamma\ \mathcal{I}\llbracket s\rrbracket}^{\mathcal{S}'}}{\langle\Pi,\ c\rangle\ \vdash_{stmt}^\Gamma\ \textbf{construct}\ c'\ x \quad \mathcal{I}\llbracket s\rrbracket \quad \textbf{destruct}\ x}$}
\end{case}

\vspace{5mm}

\noindent \textbf{Case} {\fontsize{10pt}{12pt}\selectfont $\mathcal{T}\ =\ \dfrac{\cdots}{\langle\Pi,\ c\rangle\ \vdash_{stmt}^\Gamma\ \textbf{call}\ q\textbf{\texttt{(}}x_1,\ \dots,\ x_n\textbf{\texttt{)}}}\ \textsc{T-Call}$}

\vspace{3mm}

\begin{case}
We have: \quad $\mathcal{I}\ \llbracket \textbf{call}\ q\textbf{\texttt{(}}\dots\textbf{\texttt{)}} \rrbracket\ =\ \textbf{uncall}\ q\textbf{\texttt{(}}\dots\textbf{\texttt{)}}$

\vspace{3mm}

Which means $\mathcal{T}'$ must be of the form:

\vspace{3mm}

{\fontsize{10pt}{12pt}\selectfont $\mathcal{T}'\ =\ \dfrac{\overbrace{\langle\Pi,\ c\rangle\ \vdash_{stmt}^\Gamma\ \textbf{call}\ q\textbf{\texttt{(}}x_1,\ \dots,\ x_n\textbf{\texttt{)}}}^{\mathcal{S}}}{\langle\Pi,\ c\rangle\ \vdash_{stmt}^\Gamma\ \textbf{uncall}\ q\textbf{\texttt{(}}x_1,\ \dots,\ x_n\textbf{\texttt{)}}}$}

\vspace{3mm}

Where we can simply use the derivation of $\mathcal{T}$ in place of $\mathcal{S}$.
\end{case}

\vspace{5mm}

\noindent \textbf{Case} {\fontsize{10pt}{12pt}\selectfont $\mathcal{T}\ =\ \dfrac{\cdots}{\langle\Pi,\ c\rangle\ \vdash_{stmt}^\Gamma\ \textbf{call}\ x_0\textbf{\texttt{::}}q\textbf{\texttt{(}}x_1,\ \dots,\ x_n\textbf{\texttt{)}}}\ \textsc{T-CallO}$}

\vspace{3mm}

\begin{case}
We have: \quad $\mathcal{I}\ \llbracket \textbf{call}\ x\textbf{\texttt{::}}q\textbf{\texttt{(}}\dots\textbf{\texttt{)}} \rrbracket\ =\ \textbf{uncall}\ x\textbf{\texttt{::}}q\textbf{\texttt{(}}\dots\textbf{\texttt{)}}$

\vspace{3mm}

Which means $\mathcal{T}'$ must be of the form:

\vspace{3mm}

{\fontsize{10pt}{12pt}\selectfont $\mathcal{T}'\ =\ \dfrac{\overbrace{\langle\Pi,\ c\rangle\ \vdash_{stmt}^\Gamma\ \textbf{call}\ x_0\textbf{\texttt{::}}q\textbf{\texttt{(}}x_1,\ \dots,\ x_n\textbf{\texttt{)}}}^{\mathcal{S}}}{\langle\Pi,\ c\rangle\ \vdash_{stmt}^\Gamma\ \textbf{uncall}\ x_0\textbf{\texttt{::}}q\textbf{\texttt{(}}x_1,\ \dots,\ x_n\textbf{\texttt{)}}}$}

\vspace{3mm}

Where we can simply use the derivation of $\mathcal{T}$ in place of $\mathcal{S}$.
\end{case}

\vspace{5mm}

\noindent \textbf{Case} {\fontsize{10pt}{12pt}\selectfont $\mathcal{T}\ =\ \dfrac{\overbrace{\langle\Pi,\ c\rangle\ \vdash_{stmt}^\Gamma\ \textbf{call}\ q\textbf{\texttt{(}}x_1,\ \dots,\ x_n\textbf{\texttt{)}}}^{\mathcal{S}}}{\langle\Pi,\ c\rangle\ \vdash_{stmt}^\Gamma\ \textbf{uncall}\ q\textbf{\texttt{(}}x_1,\ \dots,\ x_n\textbf{\texttt{)}}}\ \textsc{T-UC}$}

\vspace{3mm}

\begin{case}
We have: \quad $\mathcal{I}\ \llbracket \textbf{uncall}\ q\textbf{\texttt{(}}\dots\textbf{\texttt{)}} \rrbracket\ =\ \textbf{call}\ q\textbf{\texttt{(}}\dots\textbf{\texttt{)}}$

\vspace{3mm}

Which means we can just use the derivation $\mathcal{S}$ as $\mathcal{T}'$.
\end{case}

\vspace{3mm}

\noindent \textbf{Case} {\fontsize{10pt}{12pt}\selectfont $\mathcal{T}\ =\ \dfrac{\overbrace{\langle\Pi,\ c\rangle\ \vdash_{stmt}^\Gamma\ \textbf{call}\ x_0\textbf{\texttt{::}}q\textbf{\texttt{(}}x_1,\ \dots,\ x_n\textbf{\texttt{)}}}^{\mathcal{S}}}{\langle\Pi,\ c\rangle\ \vdash_{stmt}^\Gamma\ \textbf{uncall}\ x_0\textbf{\texttt{::}}q\textbf{\texttt{(}}x_1,\ \dots,\ x_n\textbf{\texttt{)}}}\ \textsc{T-UCO}$}

\vspace{3mm}

\begin{case}
We have: \quad $\mathcal{I}\ \llbracket \textbf{uncall}\ x\textbf{\texttt{::}}q\textbf{\texttt{(}}\dots\textbf{\texttt{)}} \rrbracket\ =\ \textbf{call}\ x\textbf{\texttt{::}}q\textbf{\texttt{(}}\dots\textbf{\texttt{)}}$

\vspace{3mm}

Which means we can just use the derivation $\mathcal{S}$ as $\mathcal{T}'$.
\end{case}
\end{proof}

\noindent Using Theorem~\ref{thm:type-inversion}, we can show that well-typedness is also preserved over inversion of methods. By type rule \textsc{T-Method} (See Figure~\ref{fig:program-typing}, page~\pageref{fig:program-typing}), we see that a method is well-typed iff its body is well-typed.

The class inverter $\mathcal{I}_c$ (See Figure~\ref{fig:roopl-prog-inverter}) defines the inverse of a method $q$ with body $s$, as the same method with the body $\mathcal{I}\llbracket s \rrbracket$. By Theorem~\ref{thm:type-inversion}, we know that if $s$ is well-typed, then so is $\mathcal{I}\llbracket s \rrbracket$ - by extension, if $q$ is well-typed then so is the inverse of $q$.

By the definition of the class inverter and the program inverter, it is clear that this result also extends to inversion of classes and inversion of programs.

\section{Language Extensions}
\label{sec:sugar}

The language extensions introduced in this section are not part of the core language, but are used in the ROOPL programs we present in subsequent sections and chapters.

\subsection{Local Variables}
\label{subsec:local-blocks}

Due to the restriction prohibiting member variables being passed to methods of the same object, it is sometimes necessary to create proxy objects or needlessly complicated structures to achieve relatively simple tasks. The restriction only serves to avoid aliasing situations, so we can make the life of a ROOPL programmer easier by adding the \textbf{local}/\textbf{delocal} blocks from Janus to ROOPL:

\begin{equation*}
    \textbf{local int}\ x\ =\ e_1\quad s\quad\textbf{delocal}\ x\ =\ e_2
\end{equation*}

Unlike in Janus, only integers can be allocated this way. If $x$ is already in scope at the time this statement occurs, the new $x$ shadows the definition of the existing $x$, just as is the case for object blocks. The semantics of this statement were already covered in~\cite{ty:ejanus} and do not differ in any noticeable way in ROOPL.

\subsection{Class Constructors and Deconstructors}

In OOP, a \textit{class invariant} is a constraint placed on the internal state of an object. Consider a \textit{Date} class representing a specific day of the year, with member variables denoting the day of the month and the month of the year as integers. An obvious invariant for this class is that the day of the month should always be between $1$ and $31$ inclusively and the month should always be between $1$ and $12$ inclusively. Class invariants are an instance of \textit{contract programming}\footnote{Popularized by languages such as \textit{Eiffel} and \textit{D}, which both include support for automatically verifying class invariants at runtime.} that is especially relevant for OOP, where we wish to hide the internal constraints of a class behind the public interface.

In ROOPL, all newly created objects are always zero-initialized, which is directly at odds with the notion of class invariants. In our example, this means that all \textit{Date} objects start out representing day $0$ of month $0$ which is outside of our established invariant and inconsistent with the rules of the system we are modelling. If we, instead, allow the programmer to specify how an object should be initialized, we can make sure that class invariants are enforced throughout an objects' lifetime.

\begin{figure}[ht]
\centering

\begin{equation*}
    \begin{aligned}
        &\textbf{construct}\ c\ x\textbf{\texttt{(}}x_1,\ \dots,\ x_n\textbf{\texttt{)}}\\
        &s\\
        &\textbf{destruct}\ x\textbf{\texttt{(}}z_1,\ \dots,\ z_n\textbf{\texttt{)}}
    \end{aligned}
    \qquad \overset{\textbf{def}}{\scalebox{1.8}{=}} \qquad
    \begin{aligned}
        &\textbf{construct}\ c\ x\\
        &\textbf{call}\ x\textbf{\texttt{::}}\text{constructor}\textbf{\texttt{(}}x_1,\ \dots,\ x_n\textbf{\texttt{)}}\\
        &s\\
        &\textbf{uncall}\ x\textbf{\texttt{::}}\text{constructor}\textbf{\texttt{(}}z_1,\ \dots,\ z_n\textbf{\texttt{)}}\\
        &\textbf{destruct}\ x
    \end{aligned}
\end{equation*}

\caption{Class constructor/deconstructor extension}
\label{fig:constructor-sugar}
\end{figure}

Figure~\ref{fig:constructor-sugar} shows a new form of the \textbf{construct}/\textbf{destruct} statement, which automatically invokes the special method \textit{constructor} when a new object is created, establishing the class invariants of the object. After the block statement is executed, the constructor is then automatically uncalled (we call this the \textit{deconstructor} call) before the object is then finally deallocated. The purpose of the deconstructor is to uncompute the state accumulated within the object by the constructor (and possibly by other method invocations within $s$).

Ideally the compiler should be able to enforce that the default constructor (which zero-initializes the object) is only ever invoked when the class in question does not specify its own constructor. The proposed implementation only shows how to implement class constructors/deconstructors in terms of the core language.

\begin{figure}[ht]
\centering
\begin{lstlisting}[style = basic, language = roopl]
class Date
    int day     //Invariant:   1 <= day <= 31
    int month   //Invariant:   1 <= month <= 12
    
    method constructor(int d, int m)
        if d <= 0 then
            day += 1
        else
            day += d % 31
        fi d <= 0
        
        if m <= 0 then
            month += 1
        else
            month += m % 12
        fi m <= 0
    
    method nextDay()
        if day = 31 then
            day -= 30
            call nextMonth()
        else
            day += 1
        fi day = 1
    
    method nextMonth()
        if month = 12 then
            month -= 11
        else
            month += 1
        fi month = 1

    method getDay(int out)
        out ^= day

    method getMonth(int out)
        out ^= month
\end{lstlisting}
\caption{ROOPL class representing a calendar date}
\label{fig:date-program}
\end{figure}

Note that there is no requirement that the constructor and deconstructor are given the same arguments. The only requirements are that the class invariants are established after the constructor call and that the internal state of the object is zero-cleared after the deconstructor call. Figure~\ref{fig:date-program} shows how an implementation of a simplified \textit{Date} class might look in ROOPL, with accessors and constructor/deconstructor method included.

\subsection{Expression Arguments}

Like in both Janus and R, we permit expressions to be used as arguments to a method provided the method does not directly alter the value of the parameter in any way. If the value of the expression parameter is altered by the callee, the meaning of the call is undefined.

\begin{figure}[ht]
\centering

\begin{equation*}
    \begin{aligned}
        &\textbf{call}\ q\textbf{\texttt{(}}\ \dots,\ e,\ \dots\ \textbf{\texttt{)}}\\
    \end{aligned}
    \qquad \overset{\textbf{def}}{\scalebox{1.8}{=}} \qquad
    \begin{aligned}
        &\textbf{local int}\ x'\ = e\\
        &\textbf{call}\ q\textbf{\texttt{(}}\ \dots,\ x',\ \dots\ \textbf{\texttt{)}}\\
        &\textbf{delocal}\ x'\ = e\\
    \end{aligned}
\end{equation*}

\caption{Language extension for expressions as method arguments}
\label{fig:expression-sugar}
\end{figure}

\subsection{Method Reversal}

Because arguments are passed by reference, a method invocation can bring about changes to many or all of the argument variables in the caller. On top of this, ROOPL methods are impure and can result in alterations being made to the internal state of one or more objects.

\begin{figure}[ht]
\centering

\begin{equation*}
    \begin{aligned}
        &\textbf{reversal}\ q\textbf{\texttt{(}}x_1,\ x_2\textbf{\texttt{)}}\ s\\
    \end{aligned}
    \qquad \overset{\textbf{def}}{\scalebox{1.8}{=}} \qquad
    \begin{aligned}
        &\textbf{call}\ q\textbf{\texttt{(}}x_1,\ x_2\textbf{\texttt{)}}\\
        &s\\
        &\textbf{uncall}\ q\textbf{\texttt{(}}x_1,\ x_2\textbf{\texttt{)}}
    \end{aligned}
\end{equation*}

\caption{Language extension for single-statement method reversals}
\label{fig:method-reversal}
\end{figure}

A common pattern for reversibly dealing with side effects and extra data is to sandwich the statement block handling the result between a call and an uncall of the method in question. This allows the programmer to copy the result or utilize it in some computation without worrying about the subsequent clean up. Figure~\ref{fig:method-reversal} shows a language extension that conveniently reduces this pattern to a single statement.

\subsection{Short Form Control Flow}

For the sake of convenience, we introduce short forms for conditionals and loops.

\begin{figure}[ht]
\centering

\begin{align*}
    \textbf{if}\ e_1\ \textbf{then}\ s\ \textbf{fi}\ e_2 \qquad &\overset{\textbf{def}}{\scalebox{1.8}{=}} \qquad \textbf{if}\ e_1\ \textbf{then}\ s\ \textbf{else skip fi}\ e_2\\
    \textbf{from}\ e_1\ \textbf{do}\ s\ \textbf{until}\ e_2 \qquad &\overset{\textbf{def}}{\scalebox{1.8}{=}} \qquad \textbf{from}\ e_1\ \textbf{do}\ s\ \textbf{loop skip until}\ e_2\\
    \textbf{from}\ e_1\ \textbf{loop}\ s\ \textbf{until}\ e_2 \qquad &\overset{\textbf{def}}{\scalebox{1.8}{=}} \qquad \textbf{from}\ e_1\ \textbf{do skip}\ \textbf{loop}\ s\ \textbf{until}\ e_2
\end{align*}

\caption{Syntactic sugar for short form conditionals and loops}
\label{fig:short-sugar}
\end{figure}

\section{Language Idioms}
\label{sec:idioms}

Like in conventional programming languages, specific program patterns are used, in ROOPL, to express recurring tasks or constructs that are not built-in features of the language. Such programming idioms are discussed in the following sections.

\subsection{Zero-Cleared Copying}

Care must be taken when copying and clearing values in a reversible language. Copying the value of one variable to another can only be done reversibly if the destination variable is zero-cleared, otherwise the value of the destination variable must be overwritten, resulting in a loss of information. Likewise, clearing the value of some variable is only possible if the same value is stored elsewhere at the same point in time, also to prevent loss of information. In ROOPL, both copying and clearing can be achieved with an XOR-assignment:

\begin{equation*}
    x\ \textbf{\texttt{\textasciicircum=}}\ y
\end{equation*}

If $x = y$ before the above statement, then $x$ is zero-cleared. If $x = 0$ before the assignment, then the value of $y$ is copied into $x$. This technique was first described in~\cite{ty:janus}.

\subsection{Mutators and Accessors}

\begin{figure}[ht]
\centering

\begin{lstlisting}[style = basic, language = roopl]
class Object
    int data
    
    method get(int out)
        out ^= data
    
    method swap(int in)
        data <=> in
        
    method sub(int val)
        data -= val
    
    method add(int val)
        data += val
    
    method xor(int val)
        data ^= val
\end{lstlisting}

\caption{Basic mutator and accessor methods in ROOPL}
\label{fig:accessors}
\end{figure}

In accordance with the principle of encapsulation, the member variables of a ROOPL object are not directly accessible from outside the methods of that object. To facilitate access, we can implement special accessor and mutator methods (colloquially known as \textit{getters} and \textit{setters}).

The semantics of accessors and mutators are slightly different in a reversible language. In conventional OOP languages, a mutator will simply assign a new value to the member variable, overwriting the existing value. In ROOPL we are limited to \textit{reversible mutators}, exemplified by the methods \textit{swap}, \textit{sub}, \textit{add} and \textit{xor} in Figure~\ref{fig:accessors}.

The swap mutator works mostly like a conventional mutator, but rather than irreversibly overwriting the existing value, it places that value in the parameter, leaving the caller responsible for uncomputing or clearing it.

Since ROOPL does not support return values, we must supply the accessor method \textit{get} with an output parameter. Provided the argument variable is zero-cleared before invocation, the value of the member variable is copied into the argument and thereby made accessible to the caller, outside of the object.

\subsection{Abstract Methods}

An \textit{abstract method} is a method with only a method signature but no method body. If a class contains an abstract method, it cannot be instantiated. Instead a subclass can override the abstract method and provide a method body, in which case the subclass can be instantiated. Abstract methods are used as a way to define interfaces - the base class contains a number of abstract methods that all subclasses must implement.

\begin{figure}[ht]
\centering

\begin{lstlisting}[style = basic, language = roopl]
//Shape interface
class Shape
    method resize(int scale)
        skip //Abstract method

    method translate(int x, int y)
        skip //Abstract method

    method draw()
        skip //Abstract method
    
    method getArea(int out)
        skip //Abstract method
\end{lstlisting}

\caption{Example of an interface in ROOPL}
\label{fig:interface-program}
\end{figure}

ROOPL does not have any special facilities for supporting abstract methods (See Section~\ref{sec:object-model}) but we can simulate abstract methods and class interfaces by using the \textbf{skip} statement as a method body for the abstract methods of an interface. Figure~\ref{fig:interface-program} shows an example of a class interface defined in this manner.

\subsection{Call-Uncall}

A core tenet of modern software development is the DRY-principle~\cite{ah:pragmatic}, short for Don't Repeat Yourself. It holds that duplication in logic should be eliminated via abstraction, which usually entails using methods and procedures to facilitate code reuse in a program\footnote{In fact the DRY-principle also holds that duplication in process and testing should be eliminated by automation. In the absence of DRY, a software project is said to become WET (Write Everything Twice), which is generally considered a very error-prone approach to software development.}.

In a reversible language like ROOPL, however, every statement has two distinct meanings depending on the direction of execution and therefore twice as many possible applications for the programmer to consider. As such, the potential for code reuse in ROOPL programs is considerable - many common programming tasks have an equally common inverse (the canonical examples are the \textit{push} and \textit{pop} operations of a stack), but in ROOPL such inversions are free in terms of programming effort and code size.

Another idiomatic use of the uncall mechanism is the compute-copy-uncompute technique, which reversibly uncomputes intermediate values left over after a computation, retaining only the desired results.

\subsection{Linked Lists}
\label{sec:linked-lists}

While Janus included built-in support for arrays~\cite{ty:janus} and stacks~\cite{ty:ejanus}, ROOPL does not support any data structures or collections as language primitives\footnote{There is no inherent reason such language constructs could not be added to ROOPL, and they would likely improve the expressiveness of the language. However, they are not especially noteworthy nor interesting from an OOP perspective and were therefore not included.}. Using recursion and recursively defined data types, we can define a linked list in ROOPL even without built-in support for arrays or other types of collections.

\begin{subfigures}
\begin{figure}[!ht]
\centering

\begin{lstlisting}[style = basic, language = roopl]
class Node //Represents a single node in the list
    int data
    Node next //Reference to next node in the list

    //Constructor method
    method constructor(int d, Node n)
        data ^= d
        next <=> n

    //Accessor & mutator methods
    method add(int out)
        out += data

    method sub(int out)
        out -= data

    method xor(int out)
        out ^= data

    method swap(int out)
        out <=> data

    method swapNext(Node out)
        out <=> next
    
    method length(int out) //Finds the length of the list
        out += 1
        if next != nil then
            call next::length(out)
        fi next != nil

    method insert(int n, Node new) //Inserts a (single) new node in the list
        if n = 0 then
            next <=> new
        else
            if n = 1 then
                next <=> new
            fi n = 1
            
            if next != nil then
                n -= 1
                call next::insert(n, new)
                n += 1
            fi next != nil
        fi n = 0
\end{lstlisting}

\caption{Example of recursively defined linked lists in ROOPL}
\label{fig:list-program-a}
\end{figure}

Figure~\ref{fig:list-program-a} shows the definition of a \textit{Node} class which contains a single integer and a reference to the next node in the list, which is always \textbf{nil} for the last node in a list. The node provides a constructor and a variety of accessors to both the data and the next node.

The \textit{Node} class also implements a method \textit{length} for recursively computing the length of the list. The method \textit{insert} is used to insert a single node into the list at a given index, or alternatively, extracting a node from the list when uncalled.

\begin{figure}[h]
\centering

\begin{lstlisting}[style = basic, language = roopl]
class Iterator //Iterator interface
    int result

    //Abstract method
    method run(Node head, Node next)
        skip

    //Accessor
    method get(int out)
        out <=> result

class ListBuilder
    int n //The length of the list to build
    Iterator it //The iterator instance to run
    Node empty //Helper node

    //Constructor method
    method constructor(int len, Iterator i)
        n += len
        it <=> i

    method build(Node head)
        if n = 0 then
            if head != nil
                //List is done, run the iterator
                call it::run(head, empty)
            fi head != nil
        else
            //Not yet done, construct next node
            construct Node next(n, head)
                n -= 1
                call build(next)
                n += 1
            destruct next(n, head)
        fi n = 0
\end{lstlisting}

\caption{Example of recursively defined linked lists in ROOPL (\textit{cont.})}
\label{fig:list-program-b}
\end{figure}

The \textit{ListBuilder} class defined in Figure~\ref{fig:list-program-b} is used to recursively construct lists of arbitrary length from back to front. As a \textit{Node} is constructed, it is passed its own (1-based) index in the list and a reference to the next node in the list. When the list has been built, an iterator is invoked on the head of the list (working front-to-back). When the iterator finally returns, the list is deconstructed.

\begin{figure}[ht]
\centering

\begin{lstlisting}[style = basic, language = roopl]
class Sum inherits Iterator
    int sum

    method run(Node head, Node next)
        call head::add(sum)
        call head::swapNext(next)
        if next = nil then
            result += sum //Finished
        else
            call run(next, head) //More work to do
        fi next = nil
        uncall head::swapNext(next) //Return list to original state
        uncall head::add(sum)

class Program
    int result //Final result
    Node empty //Helper node
    
    method main()
        local int n = 5 //List length
        construct Sum it //Construct iterator
            construct ListBuilder lb(n, it) //Construct list builder
                call lb::build(empty) //Build & iterate
            destruct lb(n, it)
            call it::get(result) //Fetch result
        destruct it
        delocal n = 5
\end{lstlisting}

\caption{Example of recursively defined linked lists in ROOPL (\textit{cont.})}
\label{fig:list-program-c}
\end{figure}
\end{subfigures}

The class \textit{Sum} in Figure~\ref{fig:list-program-c} on page~\pageref{fig:list-program-c}, is an example of a class that implements the \textit{Iterator} interface. It iterates over the nodes in a list, summing up the value of their contents. The class \textit{Program} illustrates how to use \textit{ListBuilder} and \textit{Sum} to build a linked-list and iterate over it. By using the \textit{Iterator} interface we make the list builder more generic - it doesn't care what kind of operation we want to perform on the list, it only cares that the iterator object it is given conforms to the interface that it knows about.

The list is created by recursively entering a \textbf{construct}/\textbf{destruct} block. When the desired length is reached, the recursion halts, the iterator is invoked and then the list is deconstructed simply by unwinding the call stack, one call (and one corresponding list node) at a time.

This style of programming is similar to continuation-passing style (CPS) - the iterator acts as a continuation that the builder can pass the list on to after it has been constructed. There is no way for the builder to return the list back to the initial caller, as that would involve unwinding the call stack and thus deconstructing the list in the process. The main difference between this approach and CPS is that CPS is usually accomplished by passing the continuation directly as a function, but since ROOPL does not support higher-order functions we are limited to using objects.

\section{Computational Strength}
\label{sec:compstrength}

A programming language is said to be \textit{computationally universal} or \textit{Turing complete} if it is capable of simulating any single-taped Turing Machine, which in turn means it is capable of computing any of the computable functions. Reversible programming languages like Janus and ROOPL are not Turing complete since they are only capable of computing exactly those computable functions that are also injective.

\citeauthor{ty:ejanus} suggests simulation of the \textit{reversible} Turing machines as the computational benchmark for reversible programming languages~\cite{ty:ejanus}. A reversible Turing machine (RTM) is any Turing machine computing an injective function~\cite{cb:reversibility, ty:flowchart}. If a reversible programming language is able to cleanly simulate any RTM, then we say that it is \textit{reversibly universal} or \textit{r-Turing complete}.

The original versions of Janus~\cite{cl:janus, ty:janus} were not r-Turing complete since they only supported static fixed-size storage. The latest version of the language adds support for dynamic storage and was proven to be r-Turing complete by construction of an RTM interpreter~\cite{ty:ejanus}. In the following sections, we present techniques for constructing a similar RTM interpreter using ROOPL. The intepreter serves as a proof that ROOPL is also reversibly universal.

\subsection{RTM Representation}

We use the same Turing machine formalism as used in~\cite{ty:ejanus}, with state transitions represented by quadruples:

\begin{definition}\label{def:quadruple-tm}(Quadruple Turing Machine)\vspace{4mm}\\
\noindent A TM T is a tuple $(Q,\ \Gamma,\ b,\ \delta,\ q_s,\ q_f)$ where
\begin{itemize}[label = {}, itemsep = 1pt]
    \item $Q$ is the finite, non-empty set of states
    \item $\Gamma$ is the finite, non-empty set of tape alphabet symbols
    \item $b\ \in\ \Gamma$ is the blank symbol
    \item $\delta\ :\ (Q\ \times\ \Gamma\ \times\ \Gamma\ \times\ Q)\ \cup\ (Q\ \times\ \{/\}\ \times\ \{L,\ R\}\ \times\ Q)$ is the partial function representing the transitions
    \item $q_s\ \in\ Q$ is the starting state
    \item $q_f\ \in\ Q$ is the final state
\end{itemize}
The symbols $L$ and $R$ represent the tape head shift-directions left and right. A quadruple is either a symbol rule of the form $(q_1,\ s_1,\ s_2,\ q_2)$ or a shift rule of the form $(q_1,\ /,\ d,\ q_2)$ where $q_1 \in Q$, $q_2 \in Q$, $s_1 \in \Gamma$, $s_2 \in \Gamma$ and $d$ being either $L$ or $R$.

A symbol rule $(q_1,\ s_1,\ s_2,\ q_2)$ means that in state $q_1$, when reading $s_1$ from the tape, write $s_2$ to the tape and change to state $q_2$. A shift rule $(q_1,\ /,\ d,\ q_2)$ means that in state $q_1$, move the tape head in direction $d$ and change to state $q_2$.
\end{definition}

\begin{definition}\label{def:reversible-tm}(Reversible Turing Machine)\vspace{4mm}\\
\noindent A TM T is a reversible TM iff, for any distinct pair of quadruples $(q_1,\ s_1,\ s_2,\ q_2)\ \in\ \delta_T$ and $(q'_1,\ s'_1,\ s'_2,\ q'_2)\ \in\ \delta_T$, we have
\begin{itemize}[label = {}, itemsep = 1pt]
    \item $q_1\ =\ q'_1\ \implies\ (t_1\ \neq\ / \quad \wedge \quad t'_1\ \neq\ / \quad \wedge \quad t_1\ \neq\ t'_1)$ (forward determinism)
    \item $q_2\ =\ q'_2\ \implies\ (t_1\ \neq\ / \quad \wedge \quad t'_1\ \neq\ / \quad \wedge \quad t_2\ \neq\ t'_2)$ (backward determinism)
\end{itemize}
\end{definition}

\vspace{1mm}

\noindent In ROOPL we can represent the set of states $\{q_1,\ \dots,\ q_n\}$ and the tape alphabet $\Gamma$ as integers. The shift rule symbol $/$ and the direction symbols $L$ and $R$ are then represented by the integer variables \textbf{\texttt{SLASH}}, \textbf{\texttt{LEFT}} and \textbf{\texttt{RIGHT}} respectively.

With this representation, we can model a transition rule as an object containing four integers \textbf{\texttt{q1}}, \textbf{\texttt{s1}}, \textbf{\texttt{s2}} and \textbf{\texttt{q2}} where \textbf{\texttt{s1}} equals \textbf{\texttt{SLASH}} for shift rules. A linked list of such transition rules serves as the full transition table $\delta$. Using the techniques described in Section~\ref{sec:linked-lists} we can look up the appropriate transition rule at each step of the simulation, with an index variable that rolls around to $0$ whenever it exceeds the length of the transition table.

Since states are numbers in our simulation, we can use a single integer variable which is updated as the simulation runs, to keep track of the current state of the RTM. After each iteration of the RTM simulation - the current state is compared to the final state \textbf{\texttt{Qf}}, if they are the same the simulation stops.

\subsection{Tape Representation}

The tape of an RTM has to be able to grow unboundedly in both directions\footnote{The term \textit{linear bounded automaton} is used to denote TM-like automatons with an upper bound on the size of the tape.}. With the tape alphabet being represented by integers, we can use a simple object containing just an integer to model a tape cell. The full tape is represented by a linked list of such cells.

The position of the tape head of the RTM determines which tape cell is currently being inspected or modified. In our simulation we can use an integer variable to store the position of the tape head as an index into the list of tape cells. Initially, the tape should contain just the input and the tape head should be at index $0$. After each simulated step of the RTM we:
\begin{enumerate}
    \item Calculate the current length of the tape.
    \item If the position of the tape head is less than zero: The tape head has moved off the left end of the tape. We allocate a new cell, prepend it to the list and zero-clear the tape head position.
    \item If the position of the tape head exceeds the current length of the tape: The tape head has moved off the right end of the tape. We allocate a new cell and append it to the tape list.
\end{enumerate}
Our model of the tape can now also grow unboundedly in both directions.

\subsection{RTM Simulation}

Figure~\ref{fig:rtm-instruction-method} shows the method \textit{inst} which executes a single instruction given a reference to the head of the tape, the position of the tape head, the current state of the RTM and four integers representing the transition rule to be executed.

\begin{figure}[!h]
\centering

\begin{lstlisting}[style = basic, language = roopl]
method inst(Cell tape, int pos, int state, int q1, int s1, int s2, int q2)
    local int symbol = 0
    call tape::lookup(pos, symbol) //Fetch current symbol
    
    if state = q1 && s1 = symbol then //SYMBOL RULE
        state += q2 - q1 //Update state to q2
        symbol += s2 - s1 //Update symbol to s2
        call tape::add(pos, s2 - s1) //Update tape cell to s2
    fi state = q2 && s2 = symbol
    
    uncall tape::lookup(pos, symbol) //Zero-clear symbol
    delocal symbol = 0

    if state = q1 && s1 = SLASH then //SHIFT RULE
        state += q2 - q1 //Update state to q2
        
        if s2 = RIGHT then
            pos += 1 //Move tape head right
        fi s2 = RIGHT
        
        if s2 = LEFT then
            pos -= 1 //Move tape head left
        fi s2 = LEFT
    fi state = q2 && s1 = SLASH
\end{lstlisting}

\caption{Method for executing a single TM transition}
\label{fig:rtm-instruction-method}
\end{figure}

Figure~\ref{fig:rtm-simulation-method} shows the recursively defined \textit{simulate} method which is the main method responsible for running the RTM simulation. It extends the tape in either direction when necessary, fetches the transition quadruple, updates the program counter and copies the result when the RTM halts.

\begin{figure}[ht]
\centering

\begin{lstlisting}[style = basic, language = roopl]
method simulate(Cell tape, int pos, int state, int pc)
    local int len = 0
    call tape::length(len) //Calculate length of tape

    if pos > len then //Append new tape cell
        construct Cell new(BLANK, empty)
        call tape::insert(pos, len)
        call simulate(tape, pos, state, pc) //Continue simulation
        uncall tape::insert(pos, len)
        destruct new(BLANK, empty)
    else
        if pos < 0 then //Prepend new tape cell
            construct Cell new(BLANK, tape)
            tape <=> new
            pos += 1
            call simulate(tape, pos, state, pc) //Continue simulation
            pos -= 1
            tape <=> new
            destruct new(BLANK, tape)
        else
            local int q1 = 0, s1 = 0, s2 = 0, q2 = 0
            call incPc(pc, PC_MAX) //Increment pc
            call RTM::get(pc, q1, s1, s2, q2) //Fetch transition quadruple
            
            call inst(tape, pos, state, q1, s1, s2, q2)

            if state = Qf then //If RTM simulation is finished
                call tape::get(result) //Copy result of simulation
            else
                call simulate(tape, pos, state, pc) //Continue simulation
            fi state = Qf
            
            uncall inst(tape, pos, state, q1, s1, s2, q2)
            
            uncall RTM::get(pc, q1, s1, s2, q2) //Clear transition quadruple
            uncall incPc(pc, PC_MAX) //Decrement pc
            delocal q1 = 0, s1 = 0, s2 = 0, q2 = 0
        fi pos < 0
    fi pos > len

    uncall tape::length(len) //Clear length of tape
    delocal len = 0
\end{lstlisting}

\caption{Main RTM simulation method}
\label{fig:rtm-simulation-method}
\end{figure}

Unlike the RTM simulator created with Janus, which uses a pair of stack primitives to represent the RTM tape, the ROOPL RTM simulator cannot finish with the TM tape as the program output. Whenever a tape cell is created, the simulator invokes the next operation recursively - but when the TM halts, the call stack of the simulation must unwind before the main method and the program can finally terminate, which results in the tape cells being deallocated one by one. The program must even ensure that the tape cells are zero-cleared before they are deallocated which can only be done reversibly by uncomputing the simulation. When the TM halts, the entire simulation therefore runs again in reverse to return the tape cells to their original state as the simulator proceeds down the call stack.
\newpage

\newcommand{\inst}[1]{\textbf{\texttt{\MakeTextUppercase{#1}}}}

\chapter{Compilation}
\label{chp:compilation}

This chapter presents the code generation schemes used to translate ROOPL source code to \textit{PISA Assembly Language} (PAL). The translated programs are semantically equivalent to the source programs and generate no additional garbage data. Due to the syntactic and semantic similarities between Janus and ROOPL, some of the techniques presented here are similar to those presented in~\cite{ha:translation} which describes the translation from Janus to PAL.

\section{Preliminaries}
\label{sec:compilation-preliminaries}

See Section~\ref{sec:pisa} in Chapter~\ref{chp:survey} for a brief description of the PISA instruction set that we target in this chapter. A more in-depth presentation of PISA and the Pendulum architecture can be found in~\cite{cv:pendulum}. For presentation purposes, we will make use of the three pseudoinstructions defined in Figure~\ref{fig:pseudoinstructions}.

\begin{figure}[h]
\centering

\begin{alignat*}{2}
    \inst{subi}\quad r\quad i \qquad &\overset{\textbf{def}}{\scalebox{1.8}{=}} \qquad&&\ \inst{addi}\quad r\quad -i\\[1.6ex]
    \inst{push}\quad r \qquad &\overset{\textbf{def}}{\scalebox{1.8}{=}} \qquad&&\Big[\inst{exch}\quad r\quad r_{sp}\ , \quad \inst{addi}\quad r_{sp}\quad 1\Big]\\[1.6ex]
    \inst{pop}\quad r \qquad &\overset{\textbf{def}}{\scalebox{1.8}{=}} \qquad&&\Big[\inst{subi}\quad r_{sp}\quad 1\ , \quad \inst{exch}\quad r\quad r_{sp}\Big]
\end{alignat*}

\caption{Definition of pseudoinstructions \inst{subi}, \inst{push} and \inst{pop}}
\label{fig:pseudoinstructions}
\end{figure}

Our translation uses \textit{virtual function tables} and \textit{object layout prefixing} to implement subtype polymorphism. Every class method of the source program is translated to a series of PISA instructions. The translated methods accept an extra hidden parameter for the object pointer, which points to the object that the method is associated with and is used to access the instance variables of that object.

\section{Memory Layout}
\label{sec:memory-layout}

We use a series of labelled load-time \inst{data} instructions at the beginning of each translated program to initialize a portion of memory with virtual function tables and other static data that the translated program needs. We refer to this portion of program memory as \textit{static storage} because it is statically sized and initialized.

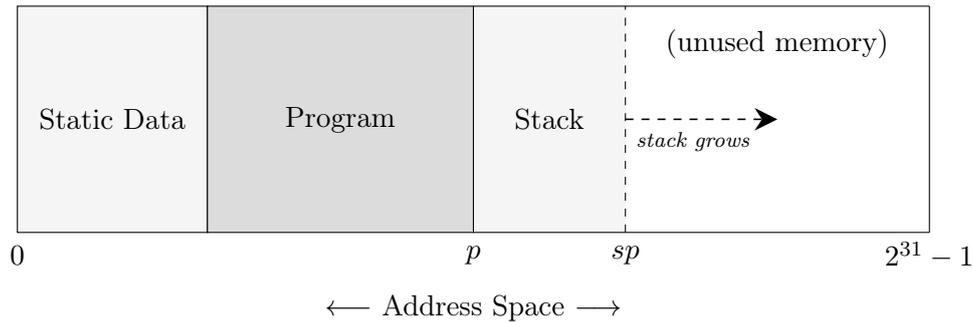
\begin{figure}[ht]
\centering

\begin{tikzpicture}
    \fill[fill = grey] (6, 0) rectangle (8, 3) node[midway] {Stack};
    \draw (6, 0) -- (12, 0);
    \draw (6, 3) -- (12, 3);
    \draw (12, 0) -- (12, 3);
    \draw[dashed] (8, 0) -- (8, 3);

    \filldraw[fill = grey, draw = black] (0, 0) rectangle (2.5, 3) node[midway] {Static Data};
    \filldraw[fill = darkgrey, draw = black] (2.5, 0) rectangle (6, 3) node[midway] {Program};
    
    \node at (10, 2.5) {(unused memory)};
    \node at (0, -.3) {$0$};
    \node at (6, -.3) {$p$};
    \node at (8, -.3) {$sp$};
    \node at (12, -.3) {$2^{31} - 1$};
    \node at (6, -1) {$\longleftarrow$ Address Space $\longrightarrow$};
    \draw[arrow, dashed] (8, 1.5) -- (10, 1.5);
    \node at (8.9, 1.2) {\scriptsize{\textit{stack grows}}};
\end{tikzpicture}

\caption{Memory layout of a ROOPL program}
\label{fig:memory-layout}
\end{figure}

Figure~\ref{fig:memory-layout} shows the full layout of a ROOPL program in memory:

\begin{enumerate}
    \item The static storage segment begins at address $0$ and contains static data initialised with \inst{data} instructions.
    \item The program segment is placed just after the static storage segment and contains the actual program instructions which consists mainly of translated class methods.
    \item The program stack is placed after the program segment at address $p$. The stack is a LIFO structure which grows and shrinks as the program executes.
\end{enumerate}

The program stack is used to store activation records, objects and local variables. The stack is accessed with the stack pointer $sp$ and initially $sp = p$.

\section{Dynamic Dispatch}
\label{sec:dynamic-dispatch}

Dynamic dispatch is a mechanism for selecting which implementation of a method to invoke, based on the type of the associated object at run time.

\begin{figure}[h]
\centering

\begin{subfigure}[t]{0.4\textwidth}
\vskip 0pt
\centering
\begin{lstlisting}[style = basic, numbers = none, frame = none, language = roopl]
class Shape
    int x
    int y

    method getArea(int out)
    method resize(int scale)
    method translate(int x, int y)
    method draw()

class Rectangle inherits Shape
    int a
    int b

    method getArea(int out)

class Circle inherits Shape
    int radius

    method getArea(int out)
    method getRadius(int out)
\end{lstlisting}
\end{subfigure}\hspace{2cm}
\begin{subfigure}[t]{0.475\textwidth}
\vskip 0pt
\centering
\resizebox{\textwidth}{!}{
\begin{minipage}{1.325\textwidth}
\begin{alignat*}{4}
&l\_Shape\_vt\ \texttt{:}\quad&&\inst{data}\quad&&90\quad&&\text{; Shape\textbf{\texttt{::}}getArea}\\
& &&\inst{data}\quad&&106\quad&&\text{; Shape\textbf{\texttt{::}}resize}\\
& &&\inst{data}\quad&&124\quad&&\text{; Shape\textbf{\texttt{::}}translate}\\
& &&\inst{data}\quad&&140\quad&&\text{; Shape\textbf{\texttt{::}}draw}\\
&l\_Rectangle\_vt\ \texttt{:}\quad&&\inst{data}\quad&&74\quad&&\text{; Rectangle\textbf{\texttt{::}}getArea}\\
& &&\inst{data}\quad&&106\quad&&\text{; Shape\textbf{\texttt{::}}resize}\\
& &&\inst{data}\quad&&124\quad&&\text{; Shape\textbf{\texttt{::}}translate}\\
& &&\inst{data}\quad&&140\quad&&\text{; Shape\textbf{\texttt{::}}draw}\\
&l\_Circle\_vt\ \texttt{:}\quad&&\inst{data}\quad&&26\quad&&\text{; Circle\textbf{\texttt{::}}getArea}\\
& &&\inst{data}\quad&&106\quad&&\text{; Shape\textbf{\texttt{::}}resize}\\
& &&\inst{data}\quad&&124\quad&&\text{; Shape\textbf{\texttt{::}}translate}\\
& &&\inst{data}\quad&&140\quad&&\text{; Shape\textbf{\texttt{::}}draw}\\
& &&\inst{data}\quad&&42\quad&&\text{; Circle\textbf{\texttt{::}}getRadius}\\
\end{alignat*}
\end{minipage}
}
\end{subfigure}

\caption[Virtual function table layout]{Virtual function table layout for a simple class hierarchy with overridden methods}
\label{fig:vtable-layout}
\end{figure}

Since ROOPL allows an object of type $\tau$ to be passed to a method expecting an object of type $\tau'$ if $\tau \prec: \tau'$, any method calls invoked on the object must be dispatched to the correct implementation in case $\tau$ overrides a method in $\tau'$. This can only be done at run time since it is impossible to determine the actual type of an object at compile time.

There are several ways to implement dynamic dispatch but the most common implementation uses virtual function tables (\textit{vtables}) to determine which implementation to dispatch to. Every class in a translated ROOPL program has a vtable which is used to map method names to the memory addresses of the method implementation for that class. Figure~\ref{fig:vtable-layout} shows how vtables in ROOPL are arranged for a simple class hierarchy:

\begin{itemize}
    \renewcommand\labelitemi{\normalfont\bfseries \textendash}
    \item The \textit{Shape} class has no base class and therefore the vtable entries all point to the original (non-overriden) method implementations.
    \item The \textit{Rectangle} class inherits from \textit{Shape} and overrides the \textit{getArea} method but does not override any other methods. Correspondingly, the vtable points to the overriding implementation of \textit{getArea} but points to the original implementations for the other methods \textit{resize}, \textit{translate} and \textit{draw}.
    \item The \textit{Circle} class is similar to \textit{Rectangle} but also adds a method \textit{getRadius} which is added to the vtable after the entries for the methods inherited from \textit{Shape}.
\end{itemize}

When a method is invoked on an object, the vtable is inspected at some statically determined offset. In our example, offset $0$ is used for invocations of method \textit{getArea}, offset $1$ is used for method \textit{resize}, offset $2$ for \textit{translate} and offset $3$ for \textit{draw}.

Placing the vtable entry for \textit{getRadius} after the entries for the inherited methods ensures that the inherited methods are placed at the same offsets in the vtable for all subclasses of \textit{Shape}. Therefore if a method is invoked on an object of type \textit{Shape}, the same offset is used to look up the address in the vtable regardless of the actual, dynamic type of the callee object. This technique is known as \textit{prefixing} and it greatly simplifies the translation of polymorphic behaviour. We also utilize prefixing in the memory layout of ROOPL objects for similar benefits.

\section{Object Layout}
\label{sec:object-layout}

Each ROOPL object consists of a pointer to the class vtable followed by a number of memory cells corresponding to the number of instance variables.

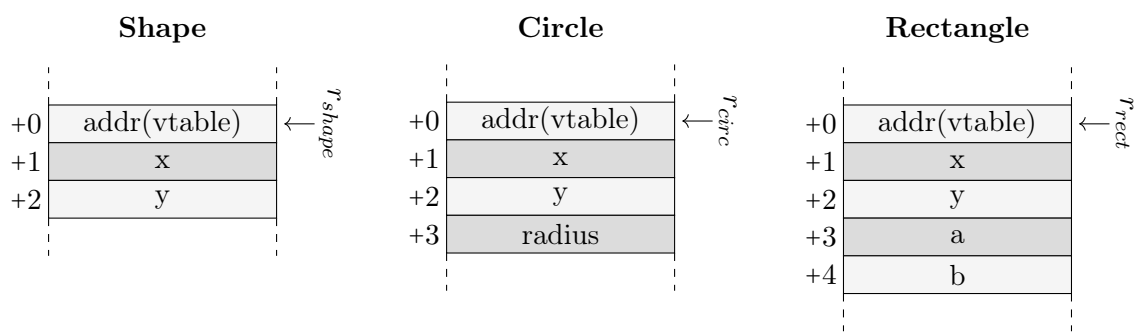
\begin{figure}[h]
\centering

\begin{subfigure}[t]{.32\textwidth}
    \vskip 0pt
    \centering
    \begin{tikzpicture}
        \draw[dashed] (0, 1.5) -- (0, 2);
        \draw[dashed] (3, 1.5) -- (3, 2);
        \filldraw[fill = grey, draw = black] (0, 1) rectangle (3, 1.5) node[midway] {addr(vtable)};
        \filldraw[fill = darkgrey, draw = black] (0, .5) rectangle (3, 1) node[midway] {x};
        \filldraw[fill = grey, draw = black] (0, 0) rectangle (3, .5) node[midway] {y};
        \draw[dashed] (0, 0) -- (0, -.5);
        \draw[dashed] (3, 0) -- (3, -.5);

        \node at (-.3, 1.25) {\texttt{+}$0$};
        \node at (-.3, .75) {\texttt{+}$1$};
        \node at (-.3, .25) {\texttt{+}$2$};
        \draw[->] (3.5, 1.25) -- (3.1, 1.25);
        \node[rotate = 270] at (3.7, 1.25) {$r_{shape}$};
        
        \node at (1.5, 2.5) {\textbf{Shape}};
    \end{tikzpicture}
\end{subfigure}
\begin{subfigure}[t]{.32\textwidth}
    \vskip 0pt
    \centering
    \begin{tikzpicture}
        \draw[dashed] (0, 1.5) -- (0, 2);
        \draw[dashed] (3, 1.5) -- (3, 2);
        \filldraw[fill = grey, draw = black] (0, 1) rectangle (3, 1.5) node[midway] {addr(vtable)};
        \filldraw[fill = darkgrey, draw = black] (0, .5) rectangle (3, 1) node[midway] {x};
        \filldraw[fill = grey, draw = black] (0, 0) rectangle (3, .5) node[midway] {y};
        \filldraw[fill = darkgrey, draw = black] (0, -.5) rectangle (3, 0) node[midway] {radius};
        \draw[dashed] (0, -.5) -- (0, -1);
        \draw[dashed] (3, -.5) -- (3, -1);

        \node at (-.3, 1.25) {\texttt{+}$0$};
        \node at (-.3, .75) {\texttt{+}$1$};
        \node at (-.3, .25) {\texttt{+}$2$};
        \node at (-.3, -.25) {\texttt{+}$3$};
        \draw[->] (3.5, 1.25) -- (3.1, 1.25);
        \node[rotate = 270] at (3.7, 1.25) {$r_{circ}$};
        
        \node at (1.5, 2.5) {\textbf{Circle}};
    \end{tikzpicture}
\end{subfigure}
\begin{subfigure}[t]{.32\textwidth}
    \vskip 0pt
    \centering
    \begin{tikzpicture}
        \draw[dashed] (0, 1.5) -- (0, 2);
        \draw[dashed] (3, 1.5) -- (3, 2);
        \filldraw[fill = grey, draw = black] (0, 1) rectangle (3, 1.5) node[midway] {addr(vtable)};
        \filldraw[fill = darkgrey, draw = black] (0, .5) rectangle (3, 1) node[midway] {x};
        \filldraw[fill = grey, draw = black] (0, 0) rectangle (3, .5) node[midway] {y};
        \filldraw[fill = darkgrey, draw = black] (0, -.5) rectangle (3, 0) node[midway] {a};
        \filldraw[fill = grey, draw = black] (0, -1) rectangle (3, -.5) node[midway] {b};
        \draw[dashed] (0, -1) -- (0, -1.5);
        \draw[dashed] (3, -1) -- (3, -1.5);

        \node at (-.3, 1.25) {\texttt{+}$0$};
        \node at (-.3, .75) {\texttt{+}$1$};
        \node at (-.3, .25) {\texttt{+}$2$};
        \node at (-.3, -.25) {\texttt{+}$3$};
        \node at (-.3, -.75) {\texttt{+}$4$};
        \draw[->] (3.5, 1.25) -- (3.1, 1.25);
        \node[rotate = 270] at (3.7, 1.25) {$r_{rect}$};
        
        \node at (1.5, 2.5) {\textbf{Rectangle}};
    \end{tikzpicture}
\end{subfigure}

\caption[Illustration of object memory layout]{Illustration of prefixing in the memory layout of 3 ROOPL objects}
\label{fig:object-layout}
\end{figure}

Figure~\ref{fig:object-layout} illustrates the layout of 3 objects based on the class hierarchy from Figure~\ref{fig:vtable-layout}. When a statement or expression refers to an instance variable, the variable offset is added to the hidden object pointer which is then dereferenced (using \inst{exch}) to fetch the value of the instance variable. Again we utilize prefixing to ensure the variable offsets are identical across subclasses of the same type.

Because the class vtable pointer is always stored at offset $0$, a vtable lookup is accomplished simply by dereferencing the pointer to the callee object, adding the method offset and then dereferencing the resulting address which yields the memory address of the method implementation.

\section{Program Structure}
\label{sec:program-structure}

The overall structure of a translated ROOPL program is illustrated in Figure~\ref{fig:pisa-program-layout}. After the static storage segment follows a series of translated class methods in turn followed by a section of code which acts as the starting point of the program.

\begin{figure}[h]
\centering

\resizebox{.8\linewidth}{!}{
\begin{minipage}{\linewidth}
\begin{alignat*}{6}
&\textbf{(1)}\quad&& &&\cdots\cdots && && &&\text{; Static data declarations}\\
&\textbf{(2)}\quad&& &&\cdots\cdots && && &&\text{; Code for program class methods}\\
&\textbf{(3)}\quad&&start\ \texttt{:}\quad&&\inst{start}\quad&& && &&\text{; Program starting point}\\
&\textbf{(4)}\quad&& &&\inst{addi}\quad &&r_{sp}\quad &&p&&\text{; Initialize stack pointer}\\
&\textbf{(5)}\quad&& &&\inst{xor}\quad &&r_m\quad &&r_{sp}\qquad &&\text{; Store address of main object in $r_m$}\\
&\textbf{(6)}\quad&& &&\inst{xori}\quad &&r_v\quad &&label_{vt}\qquad &&\text{; Store address of vtable in $r_v$}\\
&\textbf{(7)}\quad&& &&\inst{exch}\quad &&r_v\quad &&r_{sp}\qquad &&\text{; Push address of vtable onto stack}\\
&\textbf{(8)}\quad&& &&\inst{addi}\quad &&r_{sp}\quad &&size_m\qquad &&\text{; Allocate space for main object}\\
&\textbf{(9)}\quad&& &&\inst{push}\quad &&r_m\quad && &&\text{; Push '\textit{this}' onto stack}\\
&\textbf{(10)}\quad&& &&\inst{bra}\quad &&label_m \span\omit\span \qquad&&\text{; Call main procedure}\\
&\textbf{(11)}\quad&& &&\inst{pop}\quad &&r_m\quad && &&\text{; Pop '\textit{this}' from stack}\\
&\textbf{(12)}\quad&& &&\inst{subi}\quad &&r_{sp}\quad &&size_m\qquad &&\text{; Deallocate space of main object}\\
&\textbf{(13)}\quad&& &&\inst{exch}\quad &&r_v\quad &&r_{sp}\qquad &&\text{; Pop vtable address into $r_v$}\\
&\textbf{(14)}\quad&& &&\inst{xori}\quad &&r_v\quad &&label_{vt}\qquad &&\text{; Clear $r_v$}\\
&\textbf{(15)}\quad&& &&\inst{xor}\quad &&r_m\quad &&r_{sp}\qquad &&\text{; Clear $r_m$}\\
&\textbf{(16)}\quad&& &&\inst{subi}\quad &&r_{sp}\quad &&p &&\text{; Clear stack pointer}\\
&\textbf{(17)}\quad&&finish\ \texttt{:}\quad&&\inst{finish}\quad && && &&\text{; Program exit point}
\end{alignat*}
\end{minipage}
}

\caption{Overall layout of a translated ROOPL program}
\label{fig:pisa-program-layout}
\end{figure}

This section is responsible for initializing the stack pointer, allocating an instance of the object containing the main method, calling the main method, deallocating the main object and finally clearing the stack pointer:

The stack pointer is initialized simply by adding the base address of the stack to whichever register $r_{sp}$ should contain the stack pointer. The base address of the stack varies with the size of the translated program but is always known at compile-time - in Figure~\ref{fig:pisa-program-layout} the base address of the stack is simply denoted $p$. After the stack is in place, we allocate an instance of the main object on the stack by pushing the address of the vtable (denoted $label_{vt}$) onto the stack and adding the size of the object to the stack pointer (denoted $size_m$). We then push the address of this object onto the stack and unconditionally branch to the main method at $label_m$. The address of the main object is popped off the stack by the callee and serves as the object pointer.

After the main method returns, we pop the address of the main object from the stack, deallocate the object and clear the stack pointer. This is done by inverting the steps taken to initialize the stack and the object. After the program terminates, the values of the main object member variables will be left in memory where the stack used to be. This is clearly not an ideal location for the program output to reside, we address this concern in Section~\ref{sec:implementation}.

\section{Class Methods}
\label{sec:class-methods}

The calling convention described in~\cite{ha:translation} is a generalized version of the PISA calling convention presented in~\cite{mf:reversibility}, modified to support recursion. The ROOPL translation uses a similar approach with added support for method parameters (including the hidden object pointer) with pass-by-reference semantics.

\begin{figure}[h]
\centering

\resizebox{.8\linewidth}{!}{
\begin{minipage}{\linewidth}
\begin{alignat*}{6}
&\textbf{(1)}\quad&&q_{top}\ \texttt{:}\quad&&\inst{bra}\quad&&m_{bot}\quad&& &&\\
&\textbf{(2)}\quad&& &&\inst{pop}\quad&&r_{ro}\quad&& &&\text{; Load return offset}\\
&\textbf{(3)}\quad&& &&\inst{push}\quad&&r_{x_2}\quad&& &&\text{; Restore argument } x_2\\
&\textbf{(4)}\quad&& &&\inst{push}\quad&&r_{x_1}\quad&& &&\text{; Restore argument } x_1\\
&\textbf{(5)}\quad&& &&\inst{push}\quad&&r_{this}\quad&& &&\text{; Restor this-pointer}\\
&\textbf{(6)}\quad&&label_q\ \texttt{:}\quad&&\inst{swapbr}\quad&&r_{ro}\quad&& &&\text{; Method entry and exit point}\\
&\textbf{(7)}\quad&& &&\inst{neg}\quad&&r_{ro}\quad&& &&\text{; Negate return offset}\\
&\textbf{(8)}\quad&& &&\inst{pop}\quad&&r_{this}\quad&& &&\text{; Load this-pointer}\\
&\textbf{(9)}\quad&& &&\inst{pop}\quad&&r_{x_1}\quad&& &&\text{; Load argument } x_1\\
&\textbf{(10)}\quad&& &&\inst{pop}\quad&&r_{x_2}\quad&& &&\text{; Load argument } x_2\\
&\textbf{(11)}\quad&& &&\inst{push}\quad&&r_{ro}\quad&& &&\text{; Store return offset}\\
&\textbf{(12)}\quad&& &&\cdots\cdots && && &&\text{; Code for method body } q_{body}\\
&\textbf{(13)}\quad&&q_{bot}\ \texttt{:}\quad&&\inst{bra}\quad&&m_{top}\quad&& &&\\
\end{alignat*}
\end{minipage}
}

\caption{PISA translation of a ROOPL method}
\label{fig:pisa-method}
\end{figure}

Figure~\ref{fig:pisa-method} shows the PISA translation of a ROOPL method taking two parameters $x_1$ and $x_2$, with method body $q_{body}$. The caller transfers control to instruction \textbf{(6)} after which the object-pointer and method arguments are popped off the stack, the return offset is stored and the body is executed. The method prologue works identically for both directions of execution and it works with local method calls (which are simple static branch instructions) and with method calls invoked on other objects (which are dynamically dispatched). This avoids the need for multiple translations of the same method to support reverse execution, which would greatly increase the size of the translated programs.

The \inst{swapbr} instruction is used here to facilitate incoming jumps from more than one location, which would otherwise be impossible to achieve with PISA's paired-branch instructions. The return offset is swapped into register $r_{ro}$, negated (since the return offset is simply the negation of the incoming jump offset) and is then stored on the stack. When the method body finishes, the return offset is swapped back into the branch register, thereby returning the flow of execution to the caller. The arguments and offsets that are accumulated on the program stack during a (possibly nested or recursive) method invocation are cleared as the stack unwinds and the method returns. When the main method call eventually returns, just before the program terminates, the stack will have been returned to its initial, empty state.

\section{Method Invocations}
\label{sec:method-invocations}

In ROOPL, method invocations on the current object are always statically dispatched. This behaviour is known as \textit{closed recursion}. The effect of this is that local method invocations in a base class, will always dispatch to the method within that class, even if it has been overridden in a derived class. Using dynamic dispatch semantics for local method invocations (\textit{open recursion}) leads to increased program size, increased execution time and it makes program behaviour harder to reason about\footnote{Open recursion also breaks encapsulation and has been identified as the root cause of the \textit{fragile base class problem}~\cite{ja:recursion}}.

Figure~\ref{fig:pisa-local-method} shows the translation of local method invocations. The arguments are pushed on the stack in reverse order, followed by the object pointer. The jump itself is performed with an unconditional branch instruction to a statically determined label. After the method returns, the object pointer and the arguments are popped off the stack.

\begin{figure}[h]
\centering

\begin{subfigure}[t]{0.495\linewidth}
\vskip 0pt
\centering
\begin{equation*}
    \textbf{call}\ q\textbf{\texttt{(}}x_1,\ x_2\textbf{\texttt{)}}
\end{equation*}

\resizebox{.8\linewidth}{!}{
\begin{minipage}{1.025\linewidth}
\begin{alignat*}{5}
&\textbf{(1)}\quad&&\inst{push}\quad &&r_{x_2}\quad && &&\text{; Push $x_2$ onto stack}\\
&\textbf{(2)}\quad&&\inst{push}\quad &&r_{x_1}\quad && &&\text{; Push $x_1$ onto stack}\\
&\textbf{(3)}\quad&&\inst{push}\quad &&r_{t}\quad && &&\text{; Push '\textit{this}' onto stack}\\
&\textbf{(4)}\quad&&\inst{bra}\quad &&label_q\quad && &&\text{; Jump to method}\\
&\textbf{(5)}\quad&&\inst{pop}\quad &&r_{t}\quad && &&\text{; Pop '\textit{this}' from stack}\\
&\textbf{(6)}\quad&&\inst{pop}\quad &&r_{x_1}\quad && &&\text{; Pop $x_1$ from stack}\\
&\textbf{(7)}\quad&&\inst{pop}\quad &&r_{x_2}\quad && &&\text{; Pop $x_2$ from stack}
\end{alignat*}
\end{minipage}
}
\end{subfigure}
\begin{subfigure}[t]{0.495\linewidth}
\vskip 0pt
\centering
\begin{equation*}
    \textbf{uncall}\ q\textbf{\texttt{(}}x_1,\ x_2\textbf{\texttt{)}}
\end{equation*}

\resizebox{.8\linewidth}{!}{
\begin{minipage}{1.025\linewidth}
\begin{alignat*}{5}
&\textbf{(1)}\quad&&\inst{push}\quad &&r_{x_2}\quad && &&\text{; Push $x_2$ onto stack}\\
&\textbf{(2)}\quad&&\inst{push}\quad &&r_{x_1}\quad && &&\text{; Push $x_1$ onto stack}\\
&\textbf{(3)}\quad&&\inst{push}\quad &&r_{t}\quad && &&\text{; Push '\textit{this}' onto stack}\\
&\textbf{(4)}\quad&&\inst{rbra}\quad &&label_q\quad && &&\text{; Reverse jump to method}\\
&\textbf{(5)}\quad&&\inst{pop}\quad &&r_{t}\quad && &&\text{; Pop '\textit{this}' from stack}\\
&\textbf{(6)}\quad&&\inst{pop}\quad &&r_{x_1}\quad && &&\text{; Pop $x_1$ from stack}\\
&\textbf{(7)}\quad&&\inst{pop}\quad &&r_{x_2}\quad && &&\text{; Pop $x_2$ from stack}
\end{alignat*}
\end{minipage}
}
\end{subfigure}

\caption{PISA translation of local method invocations}
\label{fig:pisa-local-method}
\end{figure}

Uncalling a method is accomplished with the reverse branch instruction which flips the direction of execution after jumping to the method. Note that since we are using pass-by-reference semantics, we are in fact passing memory addresses as arguments to the method, which in turn points to the locations of the values of $x_1$ and $x_2$. The callee is responsible for dereferencing the arguments when they are used in the method body, using the \inst{exch} instruction.

Translation of non-local method calls always uses dynamic dispatch, which is slightly more involved than just jumping to a statically determined instruction label. The steps for dynamically dispatching to a method associated with a different object are:

\begin{enumerate}
    \itemsep 1pt
    \item Look up the address of the method in the object vtable and create a local copy
    \item Calculate the relative jump offset from the method invocation to the method prologue
    \item Push the arguments on the stack along with the new object pointer
    \item Perform the jump
    \item Pop the arguments from the stack
    \item Undo the jump offset calculation, to reobtain the absolute address of the method
    \item Look up the address of the method in the class vtable again, to clear the local copy
\end{enumerate}

Figure~\ref{fig:pisa-nonlocal-method} shows the translation of a dynamic method call. The first step is to dereference the callee-object to obtain the address of the class vtable. We then look up the address of the method by adding the vtable offset ($\mathit{offset}_q$) to the vtable address.

Note how this lookup involves \textit{swapping} the address stored in the vtable in static memory with the value of a register. This means the vtable is in fact altered and we need to return it to its original state before we perform the jump, since the callee might need to lookup the same method address later on. We can restore the vtable with a Lecerf-reversal by creating a copy of the method address in a register, and then undoing the lookup thereby swapping the original method address back into the vtable.

\begin{figure}[h]
\centering

\begin{equation*}
    \textbf{call}\ x\textbf{\texttt{::}}q\textbf{\texttt{(}}x_1,\ x_2\textbf{\texttt{)}}
\end{equation*}

\resizebox{.8\linewidth}{!}{
\begin{minipage}{\linewidth}
\begin{alignat*}{6}
&\textbf{(1)}\quad&& &&\inst{exch}\quad &&r_v\quad &&r_x &&\text{; Get address of vtable}\\
&\textbf{(2)}\quad&& &&\inst{addi}\quad &&r_v\quad &&\mathit{offset}_q\qquad &&\text{; Lookup $q$ in vtable}\\
&\textbf{(3)}\quad&& &&\inst{exch}\quad &&r_t\quad &&r_v &&\text{; Get address of $q$}\\
&\textbf{(4)}\quad&& &&\inst{xor}\quad &&r_{tgt}\quad &&r_t &&\text{; Copy address of $q$}\\
&\textbf{(5)}\quad&& &&\inst{exch}\quad &&r_t\quad &&r_v &&\text{; Place address back in vtable}\\
&\textbf{(6)}\quad&& &&\inst{subi}\quad &&r_v\quad &&\mathit{offset}_q\qquad &&\text{; Restore vtable pointer}\\
&\textbf{(7)}\quad&& &&\inst{exch}\quad &&r_v\quad &&r_x &&\text{; Restore object pointer}\\
&\textbf{(8)}\quad&& &&\inst{push}\quad &&r_{x_2}\quad && &&\text{; Push $x_2$ onto stack}\\
&\textbf{(9)}\quad&& &&\inst{push}\quad &&r_{x_1}\quad && &&\text{; Push $x_1$ onto stack}\\
&\textbf{(10)}\quad&& &&\inst{push}\quad &&r_x\quad && &&\text{; Push new '\textit{this}' onto stack}\\
&\textbf{(11)}\quad&& &&\inst{subi}\quad &&r_{tgt}\quad &&label_{jmp} &&\text{; Calculate jump offset}\\
&\textbf{(12)}\quad&&label_{jmp}\ \texttt{:}\quad &&\inst{swapbr}\quad &&r_{tgt}\quad && &&\text{; Jump to method}\\
&\textbf{(13)}\quad&& &&\inst{neg}\quad &&r_{tgt}\quad && &&\text{; Restore $r_{tgt}$ to original value}\\
&\textbf{(14)}\quad&& &&\inst{addi}\quad &&r_{tgt}\quad &&label_{jmp} &&\text{; Restore absolute jump value}\\
&\textbf{(15)}\quad&& &&\inst{pop}\quad &&r_x\quad && &&\text{; Pop new '\textit{this}' from stack}\\
&\textbf{(16)}\quad&& &&\inst{pop}\quad &&r_{x_1}\quad && &&\text{; Pop $x_1$ from stack}\\
&\textbf{(17)}\quad&& &&\inst{pop}\quad &&r_{x_2}\quad && &&\text{; Pop $x_2$ from stack}\\
&\textbf{(18)}\quad&& &&\inst{exch}\quad &&r_v\quad &&r_x &&\text{; Get address of vtable}\\
&\textbf{(19)}\quad&& &&\inst{addi}\quad &&r_v\quad &&\mathit{offset}_q\qquad &&\text{; Lookup $q$ in vtable}\\
&\textbf{(20)}\quad&& &&\inst{exch}\quad &&r_t\quad &&r_v &&\text{; Get address of $q$}\\
&\textbf{(21)}\quad&& &&\inst{xor}\quad &&r_{tgt}\quad &&r_t &&\text{; Clear address of $q$}\\
&\textbf{(22)}\quad&& &&\inst{exch}\quad &&r_t\quad &&r_v &&\text{; Place address back in vtable}\\
&\textbf{(23)}\quad&& &&\inst{subi}\quad &&r_v\quad &&\mathit{offset}_q\qquad &&\text{; Restore vtable pointer}\\
&\textbf{(24)}\quad&& &&\inst{exch}\quad &&r_v\quad &&r_x &&\text{; Restore object pointer}\\
\end{alignat*}
\end{minipage}
}

\caption{PISA translation of a non-local method invocation}
\label{fig:pisa-nonlocal-method}
\end{figure}

Since the usual branch instructions (\inst{bra}, \inst{rbra}, et cetera) can only jump to static instruction labels, we must use the \inst{swapbr} instruction to swap the jump offset into the branch register. Because the vtable only stores absolute method addresses, we have to calculate the jump offset manually for each method call. We can accomplish this by subtracting the memory address of the \inst{swapbr} instruction from the method address.

After the method returns, we negate the jump offset (to cancel out the negation done by the callee in the method prologue) and add the address of the \inst{swapbr} instruction to the jump offset to obtain the original absolute value of the method. To avoid leaving this method address in a register or on the stack as garbage data, we repeat the vtable lookup to clear the local method address copy. In total, the vtable is consulted 4 times per method invocation.

\begin{figure}[h]
\centering

\begin{equation*}
    \textbf{uncall}\ x\textbf{\texttt{::}}q\textbf{\texttt{(}}x_1,\ x_2\textbf{\texttt{)}}
\end{equation*}

\resizebox{.8\linewidth}{!}{
\begin{minipage}{\linewidth}
\begin{alignat*}{6}
&\textbf{(11)}\quad&& &&\inst{subi}\quad &&r_{tgt}\quad &&label_{jmp}\quad&&\text{; Calculate jump offset}\\
&\textbf{(12)}\quad&&top_{jmp} &&\inst{rbra}\quad &&bot_{jmp}\quad && &&\text{; Flip direction}\\
&\textbf{(13)}\quad&&label_{jmp}\ \texttt{:}\quad &&\inst{swapbr}\quad &&r_{tgt}\quad && &&\text{; Jump to method}\\
&\textbf{(14)}\quad&& &&\inst{neg}\quad &&r_{tgt}\quad && &&\text{; Restore $r_{tgt}$ to original value}\\
&\textbf{(15)}\quad&&bot_{jmp} &&\inst{bra}\quad &&top_{jmp}\quad && &&\text{; Paired branch}\\
&\textbf{(16)}\quad&& &&\inst{addi}\quad &&r_{tgt}\quad &&label_{jmp} &&\text{; Restore absolute jump value}
\end{alignat*}
\end{minipage}
}

\caption{PISA translation of a non-local reverse method invocation}
\label{fig:pisa-uncall-nonlocal-method}
\end{figure}

Uncalling a non-local method is analogous to calling a non-local method, with the added caveat that the direction of execution should be reversed before the jump occurs. Unlike BobISA (which has the \inst{RSWB} instruction, see Section~\ref{sec:bob} in Chapter~\ref{chp:survey}), PISA does not have a single instruction which swaps the branch register and flips the direction bit simultaneously. Figure~\ref{fig:pisa-uncall-nonlocal-method} shows how this is instead accomplished with an \inst{RBRA}/\inst{BRA} pair. The vtable lookup and cleanup is identical to the approach used in Figure~\ref{fig:pisa-nonlocal-method}.

\section{Object Blocks}
\label{sec:object-blocks}

Since the stack is maintained over (but not during) execution of a statement, we can store ROOPL objects on the program stack. The execution of an object block begins with allocation of a new object on the top of the stack. Then the block statement is executed, after which the object will again be on the top of the stack, ready for deallocation.

\begin{figure}[h]
\centering

\begin{equation*}
    \textbf{construct}\ c\ x\quad s\quad\textbf{destruct}\ x
\end{equation*}

\resizebox{.8\linewidth}{!}{
\begin{minipage}{\linewidth}
\begin{alignat*}{6}
&\textbf{(1)}\quad&&\inst{xor}\quad &&r_x\quad &&r_{sp}\qquad &&\text{; Store address of new object $x$ in $r_x$}\\
&\textbf{(2)}\quad&&\inst{xori}\quad &&r_v\quad &&label_{vt}\qquad &&\text{; Store address of vtable in $r_v$}\\
&\textbf{(3)}\quad&&\inst{exch}\quad &&r_v\quad &&r_{sp}\qquad &&\text{; Push address of vtable onto stack}\\
&\textbf{(4)}\quad&&\inst{addi}\quad &&r_{sp}\quad &&size_c\qquad &&\text{; Allocate space for new object}\\
&\textbf{(5)}\quad&&\cdots\cdots && && &&\text{; Code for statement $s$}\\
&\textbf{(6)}\quad&&\inst{subi}\quad &&r_{sp}\quad &&size_c\qquad &&\text{; Deallocate space occupied by zero-cleared object}\\
&\textbf{(7)}\quad&&\inst{exch}\quad &&r_v\quad &&r_{sp}\qquad &&\text{; Pop vtable address into $r_v$}\\
&\textbf{(8)}\quad&&\inst{xori}\quad &&r_v\quad &&label_{vt}\qquad &&\text{; Clear $r_v$}\\
&\textbf{(9)}\quad&&\inst{xor}\quad &&r_x\quad &&r_{sp}\qquad &&\text{; Clear $r_x$}
\end{alignat*}
\end{minipage}
}

\caption{PISA translation of an object block}
\label{fig:pisa-object-block}
\end{figure}

Figure~\ref{fig:pisa-object-block} illustrates how this is accomplished in practice. The immediate $label_{vt}$ is the address of the vtable for class $c$ and $size_c$ is the size of the class. The size of a class is the number of instance variables plus 1, for accomodating the vtable pointer. Within the block statement $s$, the register $r_x$ contains the address of the new object $x$.

\section{Local Blocks}
\label{sec:local-blocks}

Figure~\ref{fig:pisa-local-block} shows the translation of a local integer block. Local blocks are not part of the core language (See Section~\ref{subsec:local-blocks} in Chapter~\ref{chp:roopl}), but are included as a language extension, borrowed from Janus.

\begin{figure}[h]
\centering

\begin{equation*}
    \textbf{local int}\ x\ =\ e_1\quad s\quad\textbf{delocal}\ x\ =\ e_2
\end{equation*}

\resizebox{.8\linewidth}{!}{
\begin{minipage}{\linewidth}
\begin{alignat*}{5}
&\textbf{(1)}\quad&&\cdots\cdots && && &&\text{; Code for $r_e\ \leftarrow\ \llbracket e_1 \rrbracket$}\\
&\textbf{(2)}\quad&&\inst{xor}\quad &&r_x\quad &&r_{sp}\qquad &&\text{; Store address of new integer $x$ in $r_x$}\\
&\textbf{(3)}\quad&&\inst{xor}\quad &&r_t\quad &&r_e\qquad &&\text{; Copy value of $e_1$ into $r_t$}\\
&\textbf{(4)}\quad&&\inst{push}\quad &&r_t\quad && &&\text{; Push value of $e_1$ onto stack}\\
&\textbf{(5)}\quad&&\cdots\cdots && && &&\text{; Inverse of \textbf{(1)}}\\
&\textbf{(6)}\quad&&\cdots\cdots && && &&\text{; Code for statement $s$}\\
&\textbf{(7)}\quad&&\cdots\cdots && && &&\text{; Code for $r_e\ \leftarrow\ \llbracket e_2 \rrbracket$}\\
&\textbf{(8)}\quad&&\inst{pop}\quad &&r_t\quad && &&\text{; Pop value of $x$ into $r_t$}\\
&\textbf{(9)}\quad&&\inst{xor}\quad &&r_t\quad &&r_e\qquad &&\text{; Clear value of $r_t$ with $r_e$}\\
&\textbf{(10)}\quad&&\inst{xor}\quad &&r_x\quad &&r_{sp}\qquad &&\text{; Clear reference to $x$}\\
&\textbf{(11)}\quad&&\cdots\cdots && && &&\text{; Inverse of \textbf{(7)}}
\end{alignat*}
\end{minipage}
}

\caption{PISA translation of a local block}
\label{fig:pisa-local-block}
\end{figure}

Again, the translation can take advantage of the fact that the program stack is preserved over statement execution. This means we can place the local integers on the stack and pop them off after the block statement has been executed. Local integers are initialized with some expression $e_1$ and zero-cleared with another expression $e_2$. Evaluation of an irreversible expression in a reversible assembly language is bound to generate some amount of garbage data so we use a Lecerf-reversal to uncompute this garbage data after initializing the local variable with $e_1$, and again after clearing the local variable with $e_2$.

\section{Control Flow}
\label{sec:control-flow}

At the level of assembly language, control flow statements are usually realized via direct alteration of the program counter, which is clearly not an option for a translation targeting a reversible instruction set such as PISA. Another complication arises in the evaluation of the expressions acting as entry and exit conditions, since ROOPL expressions are irreversible.

\citeauthor{ha:translation} suggests a simple approach for arranging the translation of Janus CFOs in such a way that the garbage data produced by evaluation of the entry and exit expressions can be uncomputed without significant code duplication~\cite{ha:translation}. Since Janus (and ROOPL) uses the value $0$ for the boolean value \textit{false} and non-zero for the boolean value \textit{true}, we can safely reduce the result of evaluating the entry and exit expressions to either $0$ or $1$ while still preserving the semantics of the source program.

\begin{figure}[ht]
\centering

\resizebox{.8\linewidth}{!}{
\begin{minipage}{\linewidth}
\begin{equation*}
    \textbf{if}\ e_1\ \textbf{then}\ s_1\ \textbf{else}\ s_2\ \textbf{fi}\ e_2 \hspace{6cm} \textbf{from}\ e_1\ \textbf{do}\ s_1\ \textbf{loop}\ s_2\ \textbf{until}\ e_2
\end{equation*}
\end{minipage}
}

\resizebox{\linewidth}{!}{
\begin{minipage}{1.525\linewidth}
\begin{alignat*}{10}
&\textbf{(1)}\quad&& &&\cdots\cdots && &&\text{; Code for $r_e\ \leftarrow\ \llbracket e_1 \rrbracket_c$}&&\textbf{(1)}\quad&& &&\inst{xori}\quad &&r_t\ 1&&\text{; Set $r_t = 1$}\\
&\textbf{(2)}\quad&& &&\inst{xor}\quad &&r_t\ r_e&&\text{; Copy value of $e_1$ into $r_t$}\qquad\quad&&\textbf{(2)}\quad&&entry\ \texttt{:}\quad&&\inst{beq}&&r_t\ r_0\ assert\quad&&\text{; Receive jump}\\
&\textbf{(3)}\quad&& &&\cdots\cdots && &&\text{; Inverse of \textbf{(1)}}&&\textbf{(3)}\quad&& &&\cdots\cdots && &&\text{; Code for $r_e\ \leftarrow\ \llbracket e_1 \rrbracket_c$}\\
&\textbf{(4)}\quad&&test\ \texttt{:}\quad&&\inst{beq}&&r_t\ r_0\ test_{false}\quad&&\text{; Jump if $e_1 = 0$}&&\textbf{(4)}\quad&& &&\inst{xor}\quad &&r_t\ r_e&&\text{; Clear $r_t$}\\
&\textbf{(5)}\quad&& &&\inst{xori}\quad &&r_t\ 1&&\text{; Clear $r_t$}&&\textbf{(5)}\quad&& &&\cdots\cdots && &&\text{; Inverse of \textbf{(3)}}\\
&\textbf{(6)}\quad&& &&\cdots\cdots && &&\text{; Code for statement $s_1$}&&\textbf{(6)}\quad&& &&\cdots\cdots && &&\text{; Code for statement $s_1$}\\
&\textbf{(7)}\quad&& &&\inst{xori}\quad &&r_t\ 1&&\text{; Set $r_t = 1$}&&\textbf{(7)}\quad&& &&\cdots\cdots && &&\text{; Code for $r_e\ \leftarrow\ \llbracket e_2 \rrbracket_c$}\\
&\textbf{(8)}\quad&&assert_{true}\ \texttt{:}\quad&&\inst{bra}&&assert&&\text{; Jump}&&\textbf{(8)}\quad&& &&\inst{xor}\quad &&r_t\ r_e&&\text{; Copy value of $e_2$ into $r_t$}\\
&\textbf{(9)}\quad&&test_{false}\ \texttt{:}\quad&&\inst{bra}&&test&&\text{; Receive jump}&&\textbf{(9)}\quad&& &&\cdots\cdots && &&\text{; Inverse of \textbf{(7)}}\\
&\textbf{(10)}\quad&& &&\cdots\cdots && &&\text{; Code for statement $s_2$}&&\textbf{(10)}\quad&&test\ \texttt{:}\quad&&\inst{bne}&&r_t\ r_0\ exit&&\text{; Exit if $e_2 = 1$}\\
&\textbf{(11)}\quad&&assert\ \texttt{:}\quad&&\inst{bne}&&r_t\ r_0\ assert_{true}\quad&&\text{; Receive jump}&&\textbf{(11)}\quad&& &&\cdots\cdots && &&\text{; Code for statement $s_2$}\\
&\textbf{(12)}\quad&& &&\cdots\cdots && &&\text{; Code for $r_e\ \leftarrow\ \llbracket e_2 \rrbracket_c$}&&\textbf{(12)}\quad&&assert\ \texttt{:}\quad&&\inst{bra}&&entry&&\text{; Jump to top}\\
&\textbf{(13)}\quad&& &&\inst{xor}\quad &&r_t\ r_e&&\text{; Clear $r_t$}&&\textbf{(13)}\quad&& &&\inst{xori}\quad &&r_t\ 1&&\text{; Clear $r_t$}\\
&\textbf{(14)}\quad&& &&\cdots\cdots && &&\text{; Inverse of \textbf{(12)}}&& && && && &&
\end{alignat*}
\end{minipage}
}

\caption[PISA translation of control flow statements]{PISA translation of conditonals (left) and loops (right), from~\cite{ha:translation}}
\label{fig:pisa-control-flow}
\end{figure}

This allows us to perform the uncomputation of the expression evaluation (which clears extraneous garbage data) \textit{before} the branch is executed, while still being able to subsequently clear the register holding the result of the evaluation. Conditional statements and loops in ROOPL are essentially identical to those in Janus and this approach is therefore perfectly suitable for our ROOPL to PISA translation. Figure~\ref{fig:pisa-control-flow} shows the translation of both conditional statements and loops.

\section{Reversible Updates}
\label{sec:reversible-updates}

Figure~\ref{fig:pisa-swap-update} shows the translation of reversible variable updates and variable swapping. Since PISA does not have a built-in register swap instruction, we use the classic XOR-swap to exchange the contents of the two registers reversibly.

\begin{figure}[h]
\setlength{\abovedisplayskip}{0pt}
\setlength{\abovedisplayshortskip}{0pt}
\setlength{\belowdisplayskip}{0pt}
\setlength{\belowdisplayshortskip}{0pt}
\centering

\begin{subfigure}[b]{0.495\linewidth}
\vskip 0pt
\centering

\begin{equation*}
    x_1\ \textbf{\texttt{<=>}}\ x_2
\end{equation*}

\begin{alignat*}{4}
&\textbf{(1)}\quad&&\inst{xor}\quad &&r_{x_1}\quad &&r_{x_2}\\
&\textbf{(2)}\quad&&\inst{xor}\quad &&r_{x_2}\quad &&r_{x_1}\\
&\textbf{(3)}\quad&&\inst{xor}\quad &&r_{x_1}\quad &&r_{x_2}
\end{alignat*}
\end{subfigure}\hfill
\begin{subfigure}[b]{0.495\linewidth}
\vskip 0pt
\centering

\begin{equation*}
    x\ \odot\textbf{\texttt{=}}\ e
\end{equation*}

\begin{alignat*}{5}
&\textbf{(1)}\quad&&\cdots\cdots && && &&\text{; Code for $r_e\ \leftarrow\ \llbracket e \rrbracket$}\\
&\textbf{(2)}\quad&&\llbracket \odot \rrbracket_i\quad&&r_x\quad &&r_e\quad&&\text{; Assign $e$ to $x$}\\
&\textbf{(3)}\quad&&\cdots\cdots && && &&\text{; Inverse of \textbf{(1)}}
\end{alignat*}
\end{subfigure}

\caption{PISA translation of variable updates and variable swapping}
\label{fig:pisa-swap-update}
\end{figure}

Variable updates are accomplished with one of three instructions as well as an expression evaluation which is reversed after the update, in order to clear any accumulated garbage data. The update instruction in \textbf{(2)} is given by the function $\llbracket \odot \rrbracket_i\ :\ \mathit{ModOps}\ \rightarrow\ \mathit{Instructions}$:
\begin{equation*}
    \llbracket \textbf{\texttt{+}} \rrbracket_i\ =\ \inst{add} \qquad \llbracket \textbf{\texttt{-}} \rrbracket_i\ =\ \inst{sub} \qquad \llbracket \textbf{\texttt{\^}} \rrbracket_i\ =\ \inst{xor}
\end{equation*}
See Section~\ref{sec:syntax} in Chapter~\ref{chp:roopl} and Section~\ref{sec:pisa} in Chapter~\ref{chp:survey} for the ROOPL and PISA syntax domains.

\section{Expression Evaluation}
\label{sec:expression-evaluation}

When implementing evaluation of irreversible expressions in a reversible language, we have to accept the generation of some garbage data. Since ROOPL expressions are irreversible, every evaluation of an expression must be accompanied by a subsequent \textit{unevaluation} in order to clear any accumulated garbage data in registers and memory. This technique keeps the translation clean at the statement-level.

Code generation for evaluation of expressions is done by recursive descent over the structure of the expression tree. Numerical constants, variables and \textbf{\texttt{nil}}-nodes represent the base cases while binary expressions represent the recursive cases. A few of the binary operators supported in ROOPL (such as addition and bitwise exclusive-or) have single-instruction equivalents in PISA, but most operators are translated to more than one PISA instruction.

We consider the issue of register allocation for expression evaluation to be outside the scope of our translation. See~\cite[Section~4.5]{ha:translation} for an examination of reversible register allocation in PISA. A novel approach for reducing register pressure, by leveraging reversible computations to recompute registers instead of spilling them to memory, is presented in~\cite{mb:rematerialization}.

\section{Error Handling}
\label{sec:error-handling}

Aside from being syntactically correct and well-typed, a ROOPL program is required to meet a number of conditions that cannot, in general, be determined at compile time:

\begin{itemize}
    \renewcommand\labelitemi{\normalfont\bfseries \textendash}
    \item If the entry expression of a conditional is true, then the exit assertion should also be true after executing the \textbf{\texttt{then}}-branch.
    \item If the entry expression of a conditional is false, then the exit assertion should also be false after executing the \textbf{\texttt{else}}-branch.
    \item The entry expression of a loop should initially be true.
    \item If the exit assertion of a loop is false, then the entry expression should also be false after executing the \textbf{\texttt{loop}}-statement.
    \item All instance variables should be zero-cleared within an object block, before the object is deallocated.
    \item The value of a local variable should always match the value of the \textbf{\texttt{delocal}}-expression after the block statement has executed.
\end{itemize}

It is entirely up to the programmer to make sure these conditions are met by the program. If either of these conditions are not met, the program will silently continue with erroneous execution. To avoid such a situation, we can insert run time error checks that terminates the program or jumps to some error handler in case of programmer error.

\begin{figure}[ht]
\centering

\begin{equation*}
    \textbf{local int}\ x\ =\ e_1\quad s\quad\textbf{delocal}\ x\ =\ e_2
\end{equation*}

\resizebox{.8\linewidth}{!}{
\begin{minipage}{\linewidth}
\begin{alignat*}{6}
&\textbf{(1)}\quad&&\inst{bne}\quad &&r_t\quad &&r_0\quad&&label_{error}\quad &&\text{; Dynamic error check}\\
&\textbf{(2)}\quad&&\cdots\cdots && && && &&\text{; Code for $r_e\ \leftarrow\ \llbracket e_1 \rrbracket$}\\
&\textbf{(3)}\quad&&\inst{xor}\quad &&r_x\quad &&r_{sp} && &&\text{; Store address of new integer $x$ in $r_x$}\\
&\textbf{(4)}\quad&&\inst{xor}\quad &&r_t\quad &&r_e && &&\text{; Copy value of $e_1$ into $r_t$}\\
&\textbf{(5)}\quad&&\inst{push}\quad &&r_t\quad && && &&\text{; Push value of $e_1$ onto stack}\\
&\textbf{(6)}\quad&&\cdots\cdots && && && &&\text{; Inverse of \textbf{(1)}}\\
&\textbf{(7)}\quad&&\cdots\cdots && && && &&\text{; Code for statement $s$}\\
&\textbf{(8)}\quad&&\cdots\cdots && && && &&\text{; Code for $r_e\ \leftarrow\ \llbracket e_2 \rrbracket$}\\
&\textbf{(9)}\quad&&\inst{pop}\quad &&r_t\quad && && &&\text{; Pop value of $x$ into $r_t$}\\
&\textbf{(10)}\quad&&\inst{xor}\quad &&r_t\quad &&r_e && &&\text{; Clear value of $r_t$ with $r_e$}\\
&\textbf{(11)}\quad&&\inst{xor}\quad &&r_x\quad &&r_{sp} && &&\text{; Clear reference to $x$}\\
&\textbf{(12)}\quad&&\cdots\cdots && && && &&\text{; Inverse of \textbf{(7)}}\\
&\textbf{(13)}\quad&&\inst{bne}\quad &&r_t\quad &&r_0\quad&&label_{error}\quad &&\text{; Dynamic error check}
\end{alignat*}
\end{minipage}
}

\caption{PISA translation of a local block, with run time error checking}
\label{fig:pisa-error-check}
\end{figure}

Figure~\ref{fig:pisa-error-check} shows the translation of a local integer block with added dynamic error checks. In case the value of the local integer $v_i$ does not match the value of the \textbf{\texttt{delocal}}-expression $v_e$, the register $r_t$ will contain the non-zero value $v_i\ \oplus\ v_e$ at instruction \textbf{(13)}. If this is the case, we jump to an error routine at $label_{error}$.

The error check at \textbf{(1)} serves the same purpose as its counterpart, when the flow of execution is reversed, but has no effect otherwise since $r_t$ is empty before the statement is executed.

Dynamic error checks for conditionals, loops and object blocks can be implemented using a similar technique.

\section{Implementation}
\label{sec:implementation}

We implemented a ROOPL compiler (ROOPLC), utilizing the techniques presented in the preceding sections. The compiler serves as a proof-of-concept and does not perform any optimization of the target programs whatsoever. ROOPLC is written in Haskell (GHC, version 7.10.3) and the output was tested using the PendVM Pendulum simulator~\cite{cr:pendvm}.

Appendix~\ref{app:rooplc-source} contains the source code listings for the ROOPL compiler and Appendix~\ref{app:example-output} contains an example ROOPL program and the corresponding translated PISA program. The source code for the ROOPL compiler, additional test programs and the C source code for the PendVM simulator are also included in the enclosed ZIP archive.

The ROOPL compiler follows the PISA conventions that register $r_0$ is preserved as $0$, $r_1$ contains the stack pointer and $r_2$ stores return offsets for method invocations. Additionally, the compiler will always use $r_3$ to store the object pointer. The remaining 28 general purpose registers are used for variables, parameters and intermediate expression evaluation results.

In ROOPL, the class fields of the main class act as the program output. The program prelude, as described in Section~\ref{sec:program-structure}, leaves the value of these variables on the program stack after the program terminates. For the sake of convenience, the compiler instead copies these values from the program stack to static memory before termination. The compiler is structured as 6 separate compilation phases:

\begin{description}[leftmargin = 0pt]
\item[1. Parsing] The parsing phase transforms the input program from textual representation to an abstract syntax tree. The parser was implemented using the monadic parser combinators from the \texttt{Text.Parsec} library. See Section~\ref{sec:syntax} for details on the ROOPL syntax.

\item[2. Class Analysis] The class analysis phase verifies a number of properties of the classes in the program: Inheritance cycle detection, duplicate method names, duplicate field names and unknown base classes. The class analysis phase also computes the size of each class and constructs tables mapping class names to methods, instance variables et cetera.

\item[3. Scope Analysis] The scope analysis phase maps every occurrence of every identifier to a unique variable or method declaration. The scope analysis phase is also responsible for constructing the class virtual tables and the symbol table.

\item[4. Type Checking] The type checker uses the symbol table and the abstract syntax tree to verify that the program satisfies the ROOPL type system, as described in Section~\ref{sec:type-system}.

\item[5. Code Generation] The code generation phase translates the abstract syntax tree to a series of PISA instructions in accordance with the code generation schemes presented in this chapter. Rudimentary register allocation is also handled during code generation.

\item[6. Macro Expansion] The macro expansion phase is responsible for expanding macros left in the translated PISA program after code generation and for final processing of the output.
\end{description}

The size blowup from ROOPL to PISA is by a factor of 10 to 15 in terms of LOC. The nature of the target programs suggest that basic peephole optimization could reduce program size drastically.
\newpage

\chapter{Conclusion}
\label{chp:conclusion}

We described and formalized the reversible object-oriented programming language ROOPL and we discussed the considerations that went into its design. The language extends the design of existing imperative reversible languages in the literature and represents the first effort towards introducing OOP methodology to the field of reversible computing.

The combination of reversible computing and object-oriented programming is entirely uncharted territory and we identified the most interesting or novel points of intersection between the two disciplines, such as reversible class mutators and the proposed constructor/deconstructor extension.

Since ROOPL is the first imperative reversible language with non-trivial user-defined data types, we presented a complete static type system for the language and proved that well-typedness is preserved over statement inversion. We also demonstrated the computational strength of the language by implementing a reversible Turing machine simulator.

Finally, we established the techniques required for a clean translation from ROOPL to the reversible low-level machine language PISA and we demonstrated the feasibility of supporting core OOP features such as class inheritance and subtype polymorphism in a reversible programming language, by means of object layout prefixing and virtual function tables. We created a proof-of-concept compiler which fully implements our translation techniques.

If reversible computing is to contend with conventional computing models, we need reversibility at every level of abstraction. To this end, much has been accomplished at the circuit, gate and machine levels but aside from the work on reversible functional programming, there is little on offer in terms of high level languages and abstractions. The work presented in this thesis is a step in the direction of reconciling the abstraction techniques of conventional programming languages with the reversible programming paradigm. With ROOPL we have demonstrated that reversible object-oriented programming languages are both possible and practical.

\section{Future Work}
\label{sec:future-work}

In order to move away from the syntactically coupled allocation and deallocation mechanics used in ROOPL, more work is needed on the topics of reversible memory heaps and reversible dynamic memory management. Some work has already been done on these topics with regards to reversible functional languages~\cite{ha:constructor, tm:refcounting, tm:gc}.

ROOPL offers only the minimal toolset necessary for object-oriented programming. Advanced OOP features such as mixins, traits and generic classes could also prove to be useful in a reversible programming language and the implementation of such features could be the subject of further work.

Compilation of reversible languages is still in its infancy and the existing body of work focuses exclusively on correctness and avoiding garbage data. The practicality of reversible languages depends in part on compilation techniques that are not only correct but also \textit{performant}, both in terms of execution time and program size. In particular, optimization techniques that utilize the bidirecitonal nature of reversible programs to reduce code size shows promise and there is need for general and well-performing solutions to the reversible register allocation problem.
\newpage

\printreferences
\newpage

\begin{appendices}
\newcommand{\source}[1]
{\lstinputlisting[
    style = basic,
    language = haskell,
    frame = leftline,
    basicstyle = \scriptsize,
    numbers = left,
    breaklines = true]{#1}}

\chapter{\textsc{Rooplc} Source Code}
\label{app:rooplc-source}

\section{AST.hs}
\source{ROOPLC/AST.hs}

\newpage
\section{PISA.hs}
\source{ROOPLC/PISA.hs}

\newpage
\section{Parser.hs}
\source{ROOPLC/Parser.hs}

\newpage
\section{ClassAnalyzer.hs}
\source{ROOPLC/ClassAnalyzer.hs}

\newpage
\section{ScopeAnalyzer.hs}
\source{ROOPLC/ScopeAnalyzer.hs}

\newpage
\section{TypeChecker.hs}
\source{ROOPLC/TypeChecker.hs}

\newpage
\section{CodeGenerator.hs}
\source{ROOPLC/CodeGenerator.hs}

\newpage
\section{MacroExpander.hs}
\source{ROOPLC/MacroExpander.hs}

\newpage
\section{ROOPLC.hs}
\source{ROOPLC/ROOPLC.hs}
\newpage

\chapter{Example Output}
\label{app:example-output}

\section{LinkedList.rpl}
\begin{lstlisting}[
    style = basic,
    frame = leftline,
    language = roopl]
class Program
    int result
    int n
    Node foo

    method BuildList(Node head)
        if n = 0 then
            call head::sum(result)
        else
            construct Node next
                call next::constructor(n, head)
                n -= 1
                call BuildList(next)
                n += 1
                uncall next::constructor(n, head)
            destruct next
        fi n = 0

    method main()
        n += 7
        construct Node tail
        foo <=> tail
        call BuildList(tail)
        foo <=> tail
        destruct tail

class Node
    int data
    Node next

    method constructor(int d, Node n)
        next <=> n
        data ^= d

    method sum(int s)
        s += data
        if next = nil then
            skip
        else
            call next::sum(s)
        fi next = nil
\end{lstlisting}

\newpage
\section{LinkedList.pal}
\begin{lstlisting}[
    style = basic,
    frame = leftline,
    language = pisa]
;; pendulum pal file
top:                     BRA    start
l_r_result:              DATA   0
l_r_n:                   DATA   0
l_r_foo:                 DATA   0
l_Node_vt:               DATA   302
                         DATA   338
l_Program_vt:            DATA   15
                         DATA   242
l_BuildList_2_top:       BRA    l_BuildList_2_bot
                         ADDI   $1 -1
                         EXCH   $2 $1
                         EXCH   $4 $1
                         ADDI   $1 1
                         EXCH   $3 $1
                         ADDI   $1 1
l_BuildList_2:           SWAPBR $2
                         NEG    $2
                         ADDI   $1 -1
                         EXCH   $3 $1
                         ADDI   $1 -1
                         EXCH   $4 $1
                         EXCH   $2 $1
                         ADDI   $1 1
                         ADD    $6 $3
                         ADDI   $6 2
                         EXCH   $7 $6
cmp_top_8:               BNE    $7 $0 cmp_bot_9
                         XORI   $8 1
cmp_bot_9:               BNE    $7 $0 cmp_top_8
f_top_10:                BEQ    $8 $0 f_bot_11
                         XORI   $9 1
f_bot_11:                BEQ    $8 $0 f_top_10
                         XOR    $5 $9
f_bot_11_i:              BEQ    $8 $0 f_top_10_i
                         XORI   $9 1
f_top_10_i:              BEQ    $8 $0 f_bot_11_i
cmp_bot_9_i:             BNE    $7 $0 cmp_top_8_i
                         XORI   $8 1
cmp_top_8_i:             BNE    $7 $0 cmp_bot_9_i
                         EXCH   $7 $6
                         ADDI   $6 -2
                         SUB    $6 $3
test_4:                  BEQ    $5 $0 test_false_6
                         XORI   $5 1
                         XOR    $6 $4
                         EXCH   $7 $6
                         ADDI   $7 1
                         EXCH   $8 $7
                         XOR    $9 $8
                         EXCH   $8 $7
                         ADDI   $7 -1
                         EXCH   $7 $6
                         ADD    $7 $3
                         ADDI   $7 1
                         EXCH   $3 $1
                         ADDI   $1 1
                         EXCH   $4 $1
                         ADDI   $1 1
                         EXCH   $7 $1
                         ADDI   $1 1
                         EXCH   $6 $1
                         ADDI   $1 1
                         ADDI   $9 -63
l_jmp_12:                SWAPBR $9
                         NEG    $9
                         ADDI   $9 63
                         ADDI   $1 -1
                         EXCH   $6 $1
                         ADDI   $1 -1
                         EXCH   $7 $1
                         ADDI   $1 -1
                         EXCH   $4 $1
                         ADDI   $1 -1
                         EXCH   $3 $1
                         ADDI   $7 -1
                         SUB    $7 $3
                         EXCH   $7 $6
                         ADDI   $7 1
                         EXCH   $8 $7
                         XOR    $9 $8
                         EXCH   $8 $7
                         ADDI   $7 -1
                         EXCH   $7 $6
                         XOR    $6 $4
                         XORI   $5 1
assert_true_5:           BRA    assert_7
test_false_6:            BRA    test_4
                         XOR    $6 $1
                         XORI   $7 4
                         EXCH   $7 $1
                         ADDI   $1 3
                         XOR    $7 $6
                         EXCH   $8 $7
                         ADDI   $8 0
                         EXCH   $9 $8
                         XOR    $10 $9
                         EXCH   $9 $8
                         ADDI   $8 0
                         EXCH   $8 $7
                         ADD    $8 $3
                         ADDI   $8 2
                         EXCH   $3 $1
                         ADDI   $1 1
                         EXCH   $6 $1
                         ADDI   $1 1
                         EXCH   $8 $1
                         ADDI   $1 1
                         EXCH   $4 $1
                         ADDI   $1 1
                         EXCH   $7 $1
                         ADDI   $1 1
                         ADDI   $10 -112
l_jmp_13:                SWAPBR $10
                         NEG    $10
                         ADDI   $10 112
                         ADDI   $1 -1
                         EXCH   $7 $1
                         ADDI   $1 -1
                         EXCH   $4 $1
                         ADDI   $1 -1
                         EXCH   $8 $1
                         ADDI   $1 -1
                         EXCH   $6 $1
                         ADDI   $1 -1
                         EXCH   $3 $1
                         ADDI   $8 -2
                         SUB    $8 $3
                         EXCH   $8 $7
                         ADDI   $8 0
                         EXCH   $9 $8
                         XOR    $10 $9
                         EXCH   $9 $8
                         ADDI   $8 0
                         EXCH   $8 $7
                         XOR    $7 $6
                         ADD    $7 $3
                         ADDI   $7 2
                         EXCH   $8 $7
                         XORI   $9 1
                         SUB    $8 $9
                         XORI   $9 1
                         EXCH   $8 $7
                         ADDI   $7 -2
                         SUB    $7 $3
                         EXCH   $4 $1
                         ADDI   $1 1
                         EXCH   $6 $1
                         ADDI   $1 1
                         EXCH   $3 $1
                         ADDI   $1 1
                         BRA    l_BuildList_2
                         ADDI   $1 -1
                         EXCH   $3 $1
                         ADDI   $1 -1
                         EXCH   $6 $1
                         ADDI   $1 -1
                         EXCH   $4 $1
                         ADD    $7 $3
                         ADDI   $7 2
                         EXCH   $8 $7
                         XORI   $9 1
                         ADD    $8 $9
                         XORI   $9 1
                         EXCH   $8 $7
                         ADDI   $7 -2
                         SUB    $7 $3
                         XOR    $7 $6
                         EXCH   $8 $7
                         ADDI   $8 0
                         EXCH   $9 $8
                         XOR    $10 $9
                         EXCH   $9 $8
                         ADDI   $8 0
                         EXCH   $8 $7
                         ADD    $8 $3
                         ADDI   $8 2
                         EXCH   $3 $1
                         ADDI   $1 1
                         EXCH   $6 $1
                         ADDI   $1 1
                         EXCH   $8 $1
                         ADDI   $1 1
                         EXCH   $4 $1
                         ADDI   $1 1
                         EXCH   $7 $1
                         ADDI   $1 1
                         ADDI   $10 -188
l_rjmp_top_15:           RBRA   l_rjmp_bot_16
l_jmp_14:                SWAPBR $10
                         NEG    $10
l_rjmp_bot_16:           BRA    l_rjmp_top_15
                         ADDI   $10 188
                         ADDI   $1 -1
                         EXCH   $7 $1
                         ADDI   $1 -1
                         EXCH   $4 $1
                         ADDI   $1 -1
                         EXCH   $8 $1
                         ADDI   $1 -1
                         EXCH   $6 $1
                         ADDI   $1 -1
                         EXCH   $3 $1
                         ADDI   $8 -2
                         SUB    $8 $3
                         EXCH   $8 $7
                         ADDI   $8 0
                         EXCH   $9 $8
                         XOR    $10 $9
                         EXCH   $9 $8
                         ADDI   $8 0
                         EXCH   $8 $7
                         XOR    $7 $6
                         ADDI   $1 -3
                         EXCH   $7 $1
                         XORI   $7 4
                         XOR    $6 $1
assert_7:                BNE    $5 $0 assert_true_5
                         ADD    $6 $3
                         ADDI   $6 2
                         EXCH   $7 $6
cmp_top_17:              BNE    $7 $0 cmp_bot_18
                         XORI   $8 1
cmp_bot_18:              BNE    $7 $0 cmp_top_17
f_top_19:                BEQ    $8 $0 f_bot_20
                         XORI   $9 1
f_bot_20:                BEQ    $8 $0 f_top_19
                         XOR    $5 $9
f_bot_20_i:              BEQ    $8 $0 f_top_19_i
                         XORI   $9 1
f_top_19_i:              BEQ    $8 $0 f_bot_20_i
cmp_bot_18_i:            BNE    $7 $0 cmp_top_17_i
                         XORI   $8 1
cmp_top_17_i:            BNE    $7 $0 cmp_bot_18_i
                         EXCH   $7 $6
                         ADDI   $6 -2
                         SUB    $6 $3
l_BuildList_2_bot:       BRA    l_BuildList_2_top
l_main_3_top:            BRA    l_main_3_bot
                         ADDI   $1 -1
                         EXCH   $2 $1
                         EXCH   $3 $1
                         ADDI   $1 1
l_main_3:                SWAPBR $2
                         NEG    $2
                         ADDI   $1 -1
                         EXCH   $3 $1
                         EXCH   $2 $1
                         ADDI   $1 1
                         ADD    $4 $3
                         ADDI   $4 2
                         EXCH   $5 $4
                         XORI   $6 7
                         ADD    $5 $6
                         XORI   $6 7
                         EXCH   $5 $4
                         ADDI   $4 -2
                         SUB    $4 $3
                         XOR    $4 $1
                         XORI   $5 4
                         EXCH   $5 $1
                         ADDI   $1 3
                         ADD    $5 $3
                         ADDI   $5 3
                         EXCH   $6 $5
                         XOR    $6 $4
                         XOR    $4 $6
                         XOR    $6 $4
                         EXCH   $6 $5
                         ADDI   $5 -3
                         SUB    $5 $3
                         EXCH   $4 $1
                         ADDI   $1 1
                         EXCH   $3 $1
                         ADDI   $1 1
                         BRA    l_BuildList_2
                         ADDI   $1 -1
                         EXCH   $3 $1
                         ADDI   $1 -1
                         EXCH   $4 $1
                         ADD    $5 $3
                         ADDI   $5 3
                         EXCH   $6 $5
                         XOR    $6 $4
                         XOR    $4 $6
                         XOR    $6 $4
                         EXCH   $6 $5
                         ADDI   $5 -3
                         SUB    $5 $3
                         ADDI   $1 -3
                         EXCH   $5 $1
                         XORI   $5 4
                         XOR    $4 $1
l_main_3_bot:            BRA    l_main_3_top
l_constructor_0_top:     BRA    l_constructor_0_bot
                         ADDI   $1 -1
                         EXCH   $2 $1
                         EXCH   $4 $1
                         ADDI   $1 1
                         EXCH   $5 $1
                         ADDI   $1 1
                         EXCH   $3 $1
                         ADDI   $1 1
l_constructor_0:         SWAPBR $2
                         NEG    $2
                         ADDI   $1 -1
                         EXCH   $3 $1
                         ADDI   $1 -1
                         EXCH   $5 $1
                         ADDI   $1 -1
                         EXCH   $4 $1
                         EXCH   $2 $1
                         ADDI   $1 1
                         ADD    $6 $3
                         ADDI   $6 2
                         EXCH   $7 $6
                         XOR    $7 $5
                         XOR    $5 $7
                         XOR    $7 $5
                         EXCH   $7 $6
                         ADDI   $6 -2
                         SUB    $6 $3
                         ADD    $6 $3
                         ADDI   $6 1
                         EXCH   $7 $6
                         EXCH   $8 $4
                         XOR    $7 $8
                         EXCH   $8 $4
                         EXCH   $7 $6
                         ADDI   $6 -1
                         SUB    $6 $3
l_constructor_0_bot:     BRA    l_constructor_0_top
l_sum_1_top:             BRA    l_sum_1_bot
                         ADDI   $1 -1
                         EXCH   $2 $1
                         EXCH   $4 $1
                         ADDI   $1 1
                         EXCH   $3 $1
                         ADDI   $1 1
l_sum_1:                 SWAPBR $2
                         NEG    $2
                         ADDI   $1 -1
                         EXCH   $3 $1
                         ADDI   $1 -1
                         EXCH   $4 $1
                         EXCH   $2 $1
                         ADDI   $1 1
                         EXCH   $5 $4
                         ADD    $6 $3
                         ADDI   $6 1
                         EXCH   $7 $6
                         ADD    $5 $7
                         EXCH   $7 $6
                         ADDI   $6 -1
                         SUB    $6 $3
                         EXCH   $5 $4
                         ADD    $6 $3
                         ADDI   $6 2
                         EXCH   $7 $6
cmp_top_25:              BNE    $7 $0 cmp_bot_26
                         XORI   $8 1
cmp_bot_26:              BNE    $7 $0 cmp_top_25
f_top_27:                BEQ    $8 $0 f_bot_28
                         XORI   $9 1
f_bot_28:                BEQ    $8 $0 f_top_27
                         XOR    $5 $9
f_bot_28_i:              BEQ    $8 $0 f_top_27_i
                         XORI   $9 1
f_top_27_i:              BEQ    $8 $0 f_bot_28_i
cmp_bot_26_i:            BNE    $7 $0 cmp_top_25_i
                         XORI   $8 1
cmp_top_25_i:            BNE    $7 $0 cmp_bot_26_i
                         EXCH   $7 $6
                         ADDI   $6 -2
                         SUB    $6 $3
test_21:                 BEQ    $5 $0 test_false_23
                         XORI   $5 1
                         XORI   $5 1
assert_true_22:          BRA    assert_24
test_false_23:           BRA    test_21
                         ADD    $6 $3
                         ADDI   $6 2
                         EXCH   $7 $6
                         XOR    $8 $7
                         EXCH   $9 $8
                         ADDI   $9 1
                         EXCH   $10 $9
                         XOR    $11 $10
                         EXCH   $10 $9
                         ADDI   $9 -1
                         EXCH   $9 $8
                         EXCH   $7 $6
                         ADDI   $6 -2
                         SUB    $6 $3
                         EXCH   $3 $1
                         ADDI   $1 1
                         EXCH   $4 $1
                         ADDI   $1 1
                         EXCH   $8 $1
                         ADDI   $1 1
                         ADDI   $11 -400
l_jmp_29:                SWAPBR $11
                         NEG    $11
                         ADDI   $11 400
                         ADDI   $1 -1
                         EXCH   $8 $1
                         ADDI   $1 -1
                         EXCH   $4 $1
                         ADDI   $1 -1
                         EXCH   $3 $1
                         ADD    $6 $3
                         ADDI   $6 2
                         EXCH   $7 $6
                         EXCH   $9 $8
                         ADDI   $9 1
                         EXCH   $10 $9
                         XOR    $11 $10
                         EXCH   $10 $9
                         ADDI   $9 -1
                         EXCH   $9 $8
                         XOR    $8 $7
                         EXCH   $7 $6
                         ADDI   $6 -2
                         SUB    $6 $3
assert_24:               BNE    $5 $0 assert_true_22
                         ADD    $6 $3
                         ADDI   $6 2
                         EXCH   $7 $6
cmp_top_30:              BNE    $7 $0 cmp_bot_31
                         XORI   $8 1
cmp_bot_31:              BNE    $7 $0 cmp_top_30
f_top_32:                BEQ    $8 $0 f_bot_33
                         XORI   $9 1
f_bot_33:                BEQ    $8 $0 f_top_32
                         XOR    $5 $9
f_bot_33_i:              BEQ    $8 $0 f_top_32_i
                         XORI   $9 1
f_top_32_i:              BEQ    $8 $0 f_bot_33_i
cmp_bot_31_i:            BNE    $7 $0 cmp_top_30_i
                         XORI   $8 1
cmp_top_30_i:            BNE    $7 $0 cmp_bot_31_i
                         EXCH   $7 $6
                         ADDI   $6 -2
                         SUB    $6 $3
l_sum_1_bot:             BRA    l_sum_1_top
start:                   BRA    top
                         START 
                         ADDI   $1 480
                         XOR    $3 $1
                         XORI   $4 6
                         EXCH   $4 $1
                         ADDI   $1 4
                         EXCH   $3 $1
                         ADDI   $1 1
                         BRA    l_main_3
                         ADDI   $1 -1
                         EXCH   $3 $1
                         ADDI   $1 -1
                         EXCH   $4 $1
                         XORI   $5 3
                         EXCH   $4 $5
                         XORI   $5 3
                         ADDI   $1 1
                         ADDI   $1 -2
                         EXCH   $4 $1
                         XORI   $5 2
                         EXCH   $4 $5
                         XORI   $5 2
                         ADDI   $1 2
                         ADDI   $1 -3
                         EXCH   $4 $1
                         XORI   $5 1
                         EXCH   $4 $5
                         XORI   $5 1
                         ADDI   $1 3
                         ADDI   $1 -4
                         EXCH   $4 $1
                         XORI   $4 6
                         XOR    $3 $1
                         ADDI   $1 -480
finish:                  FINISH
\end{lstlisting}
\end{appendices}

\end{document}